\documentclass{amsart}
\usepackage{bc}
\begin{document}
\title{Algorithmic aspects of branched coverings}
\author{Laurent Bartholdi}
\email{laurent.bartholdi@gmail.com}
\author{Dzmitry Dudko}
\email{dzmitry.dudko@gmail.com}
\address{\'Ecole Normale Sup\'erieure, Paris \emph{and} Mathematisches Institut, Georg-August Universit\"at zu G\"ottingen}
\thanks{Partially supported by ANR grant ANR-14-ACHN-0018-01 and DFG grant BA4197/6-1}
\date{May 18, 2017}
\begin{abstract}
  This is a survey, and a condensed version, of a series of articles
  on the algorithmic study of Thurston maps. We describe branched
  coverings of the sphere in terms of group-theoretical objects called
  bisets, and develop a theory of decompositions of bisets.

  We introduce a canonical ``Levy'' decomposition of an arbitrary
  Thurston map into homeomorphisms, metrically-expanding maps and maps
  doubly covered by torus endomorphisms. The homeomorphisms decompose
  themselves into finite-order and pseudo-Anosov maps, and the
  expanding maps decompose themselves into rational maps.

  As an outcome, we prove that it is decidable when two Thurston maps
  are equivalent. We also show that the decompositions above are
  computable, both in theory and in practice.
\end{abstract}

\maketitle
\tableofcontents
\let\myboldmath=\boldmath

\setcounter{section}{-1}\section{Introduction}
Nielsen~\cite{nielsen:surfaces}, and later
Thurston~\cite{thurston:surfaces}, have achieved an impressive
classification of surface self-homeomorphisms. Given a compact surface
$X$ and a self-homeomorphism $f\colon X\selfmap$, there exists a
canonical set $\CC$ of curves, invariant up to isotopy by $f$, that
separate $X$ into simpler surfaces, and such that the induced first
return maps on these pieces are isotopic to either finite-order or
\emph{pseudo-Anosov} maps.

These induced maps all preserve a geometric structure: finite-order
transformations are hyperbolic isometries (so preserve a complex
structure), while pseudo-Anosov maps preserve a pair of transverse
foliations, expanding one and contracting the other.

This classification result may also be viewed as a bridge between
topology and group theory: $f$ naturally acts by automorphisms on the
fundamental group $G$ of $X$; the collection $\CC$ of curves
determines a splitting of $G$ as amalgamated free product over cyclic
subgroups, and the induced automorphisms of the pieces are either
finite-order or irreducible in a strong sense (``iwip'': irreducible
with irreducible powers). Thus the decomposition of $X$ as an amalgam
over circles naturally parallels a decomposition of $G$ as an amalgam
over cyclic subgroups.

Our aim is to do the same for branched self-coverings of compact surfaces,
namely maps $f\colon X\selfmap$ that are coverings away from a finite set
of \emph{branch points}, where they admit local models of the form
$z\mapsto z^d$ in complex charts. If $f$ has degree $>1$, the only surfaces
to consider are the sphere and the torus, by the Riemann-Hurwitz
formula. Examples of branched self-coverings of the sphere include rational
maps in $\C(z)$ and their compositions with homeomorphisms; branched
self-coverings of the torus admit no branch points, so are genuine
coverings and can all be represented on the torus $\R^2/\Z^2$ (after
punctures are filled in) as $z\mapsto M z+b$ for some $M\in\Z^{2\times2}$
and some $b\in\R^2$. These maps may descend to the sphere: if $2b\in\Z^2$
and $\wp\colon\R^2/\Z^2\to S^2$ is a branched covering satisfying
$\wp(-p)=\wp(p)$ then $\wp\circ f\circ\wp^{-1}$ is a branched self-covering
of $S^2$.

An extra ingredient, besides topology and group theory, becomes
available if $f$ has degree $>1$: it may happen that $X$ admits a
metric that is expanded by $f$. Even better, $X$ may admit a complex
structure that is preserved by $f$, in which case there exists a
conformal metric that is expanded by $f$.

\subsection{Overview of results}\label{ss:overview}
We consider self-branched coverings $f\colon(S^2,A)\selfmap$, with $A$ a
finite subset of $S^2$ containing $f(A)$ and the critical values of $f$;
such maps are called \emph{marked Thurston maps}. For example, $A$ could be
the \emph{post-critical set} $P_f=\bigcup_{n>0}f^n(\text{critical
  points})$.

Since Thurston's fundamental work, it is customary to consider such
maps $f$ up to \emph{combinatorial equivalence}: two maps $f_0,f_1$
are equivalent if they can be deformed smoothly into one another along
a path $f_t\colon(S^2,A_t)\selfmap$ of marked Thurston maps.

This notion is a combination of two stricter notions: \emph{conjugacy}
by a homeomorphism $(S^2,A_0)\to(S^2,A_1)$ and \emph{isotopy rel $A$},
namely along a path of Thurston maps with constant marked set $A$. The
\emph{centralizer} of a Thurston map $f\colon(S^2,A)\selfmap$ is the
group of pure mapping classes $g\in\Mod(S^2,A)$ such that $g^{-1}f g$
is isotopic to $f$.

Our work makes essential use of a fundamental invariant introduced by
Nekrashevych, the \emph{biset} of a branched covering. Choose a
basepoint $*\in S^2\setminus A$ and write $G=\pi_1(S^2\setminus A,*)$.
Then the biset of a branched self-covering $f\colon(S^2,A)\selfmap$ is
a set $B(f)$ equipped with two commuting actions of $G$, whose
(appropriately defined) isomorphism class is a complete invariant of
$f$ up to combinatorial equivalence.  Since $G$ is a free group,
calculations in $B(f)$ are easy to perform.

The theory is slightly complicated by a family of branched
self-coverings $f\colon(S^2,A)\selfmap$ that come from a self-covering
of the torus $\tilde f\colon T^2\selfmap$ via a degree-$2$ branched
covering $\wp\colon T^2\to S^2$; i.e.\ $\wp\circ\tilde f=f\circ\wp$.
Suppose we have $\tilde f(z)=M z+b$ on the model $\R^2/\Z^2$ of $T^2$,
and assume that the eigenvalues of $M$ are real but not rational. We
call the map $f$ \emph{irrational doubly covered by a torus
  endomorphism}. It may furthermore be \emph{expanding}, if all
eigenvalues of $M$ have norm $>1$.

A \emph{Levy obstruction} is a cycle of curves
$\gamma_0,\gamma_1,\dots,\gamma_n=\gamma_0$ on $S^2\setminus A$ with $f$
mapping $\gamma_i$ to $\gamma_{i+1}$ by degree $1$, up to isotopy. It is
clearly an obstruction to the existence of an $f$-expanding metric, and we
show that it is the only one:
\begin{mainthm}
  Suppose $f\colon (S^2,A)\selfmap$ is a Thurston map with degree at
  least $2$ such that $f$ admits no Levy obstruction. Then either $f$
  is isotopic to a map expanding a metric on $S^2\setminus A$, or $f$
  is isotopic to the quotient by the involution $z\mapsto -z$ of an
  affine map on $\R^2/\Z^2$ whose eigenvalues are different from
  $\pm1$.
\end{mainthm}

By a \emph{decomposition} of a map $f\colon(S^2,A)\selfmap$
we mean a decomposition of $S^2\setminus A$ into punctured spheres
along an $f$-invariant multicurve; the \emph{pieces} of the
decomposition are the return maps of $f$ on the sub-spheres.

We first show that every Thurston map $f\colon(S^2,A)\selfmap$ may
canonically be decomposed into pieces that, up to isotopy, (1) are
homeomorphisms, or (2) expand a metric on $S^2\setminus A$, or (3) are
non-expanding irrational doubly covered by torus endomorphisms; see
Figure~\ref{fig:geometric}.

The second case (expanding a metric) is equivalent to (2') a
topological property of $f$ (it does not admit Levy cycles), and to
(2'') a group-theoretical property (the biset of $f$ is contracting).

According to the classical Nielsen-Thurston theory, homeomorphisms in
case (1) may canonically be further decomposed, again up to isotopy,
into maps of finite order and pseudo-Anosov homeomorphisms, namely
homeomorphisms that preserve a pair of transverse foliations on
$S^2\setminus A$.

The decomposition theory of Pilgrim~\cite{pilgrim:combinations} lets
us decompose expanding maps that are not doubly covered by irrational
torus endomorphisms into pieces that preserve a complex structure,
namely are rational maps. Therefore (2) implies another
group-theoretical property (4): the biset of $f$ decomposes as an
amalgam over cyclic bisets (i.e.\ transitive bisets over cyclic
groups), with rational pieces, and an algebraic-geometry property (5):
there exists a complex stable curve (an algebraic variety $\mathbb X$
consisting of complex spheres with marked points, arranged as a
cactoid) and a rational map $\mathbb X\dashrightarrow\mathbb X$ that
becomes isotopic to $f$ when the nodes of $\mathbb X$ are resolved;
see Theorem~\ref{thm:G}.

We show how decision problems -- isotopy rel $A$ and computation of
centralizers -- can be promoted from pieces in a decomposition to the global
map. We also show that the points in $A\setminus P_f$ can be encoded in
group-theoretical language, as finite sequences of biset elements. Finally,
isotopy and centralizers can be computed for each of the pieces:
finite-order, pseudo-Anosov, rational maps and maps doubly covered by
torus endomorphisms. Our main result follows:
\begin{mainthm}\label{thm:decidable}
  It is decidable whether or not two Thurston maps are combinatorially
  equivalent.

  Furthermore, the centralizer of a Thurston map $f$ (i.e.\ the set of
  homeomorphisms that commute with $f$ up to isotopy) is effectively
  computable.
\end{mainthm}
This extends a series of partial results: Bonnot, Braverman and
Yampolsky show in~\cite{bonnot-braverman-yampolsky:thurs tondecidable}
that equivalence to a rational map is decidable, and Selinger and
Yampolsky show in~\cite{selinger-yampolsky:geometrization}*{Main
  Theorem~III} that it is decidable whether $g$ is equivalent to $f$
provided that all return maps in the canonical decomposition of $f$
are rational maps with hyperbolic orbifolds.

Even though the first examples of branched self-coverings of the sphere of
degree $>1$ are rational maps, they benefit greatly from sometimes being
considered as topological maps. This is because surgery may be performed on
topological maps: one may decompose them into smaller pieces, glue two maps
together (``mating'', see~\S\ref{ss:matings}), etc. A fundamental theorem
of Thurston asserts that, once a topological condition is satisfied (``no
annular obstruction''), the topological map can be isotoped back into a unique
rational map. What we are proposing, in this research, is to express the
topology in a group-theoretical language in which fundamental questions
become decidable and effectively computable, see~\S\ref{ss:examples}.

Note that, in contrast to non-invertible branched self-coverings, the
algorithmic theory of $\Mod$ is fairly advanced, and in particular the
Nielsen-Thurston decomposition is known to be efficiently computable,
starting with train tracks as shown by Bestvina and
Handel~\cite{bestvina-h:tt}, and actually in polynomial-time as
recently announced by Margalit, Strenner and Yurttas.

\subsection{Structure of the papers}
In the first article~\cite{bartholdi-dudko:bc1}, we develop the
general machinery of bisets and decompositions of bisets. The main
definitions are graphs of groups and graphs of bisets, and the main
result is a van Kampen-like theorem: given a correspondence
$Y\leftarrow Z\to X$ and appropriately compatible covers of $X,Y,Z$, a
graph of bisets is obtained by restricting the correspondence to the
sets in the cover; and the van Kampen theorem expresses the biset of
the correspondence as the ``fundamental biset'' of the graph of
bisets, just as the fundamental group of a space is the fundamental
group of its graph of groups.

In the second article~\cite{bartholdi-dudko:bc2}, we specialize to
punctured spheres, or more generally orbispace structures on spheres,
which we treat as groups of the form
\[G=\langle \gamma_1,\dots,\gamma_n\mid
\gamma_1^{e_1}=\cdots=\gamma_n^{e_n}=\gamma_1\cdots\gamma_n=1\rangle
\]
with all $e_i\in\{2,3,\dots,\infty\}$. Curves on $S^2\setminus A$ are
treated as conjugacy classes in $G$, and multicurves as collections of
conjugacy classes. We use detailed structure about pure mapping class
groups to explain how conjugacy and centralizer problems for pieces of
a decomposition can be promoted to the original Thurston map.

In the third article~\cite{bartholdi-dudko:bc3}, we study the effect
of erasing punctures (periodic cycles or marked preimages of
post-critical points) from a Thurston map. We show how these erased
points can be encoded by a finite subset of the biset called a
\emph{portrait}. We use this language to study more carefully the maps
that are doubly covered by torus endomorphisms and characterize them
via elementary group theory.

In the fourth article~\cite{bartholdi-dudko:bc4}, we prove the first
decomposition theorem, of an arbitrary Thurston map into
homeomorphisms, expanding maps and maps doubly covered by torus
endomorphisms. We also give the characterization of expanding maps as
Levy-free maps and as maps with contracting biset. In particular, the
above decomposition is along a minimal Levy multicurve such that all
pieces are Levy-free or homeomorphisms.

In the fifth article~\cite{bartholdi-dudko:bc5}, we describe
algorithms, and their implementations, that
\begin{itemize}
\item convert the Poirier description of a complex polynomial by its
  external angles into a biset;
\item convert the Hubbard tree description of a polynomial into a
  biset;
\item convert a polynomial biset into its Poirier description by
  external angles;
\item convert a floating-point approximation of a rational map into a
  biset;
\item convert a sphere biset into a complex rational map with
  algebraic coefficients, or produce an invariant multicurve that
  testifies to the inexistence of a rational map.
\end{itemize}
The first three algorithms are entirely symbolic, while the last two
require floating-point calculations as well as manipulations of
triangulations on the sphere. All these algorithms have been
implemented in the software package \textsc{Img}~\cite{img:manual}
within the computer algebra system \textsc{Gap}~\cite{gap4.5:manual}.

Finally, we prove Theorem~\ref{thm:decidable}
in~\S\ref{ss:decidability}, assuming all the results in the previously
mentioned articles.

We give below, in separate sections, condensates of the contents of
these articles, with relevant definitions and sketches of proofs, and
conclude this article with a series of examples seen from the
topological, group-theoretical and algebraic perspectives.

\subsection{Pilgrim's decomposition}\label{ss:canonical decomposition}
Pilgrim develops in~\cite{pilgrim:combinations} a decomposition theory
for branched coverings. In particular, he constructs a \emph{canonical
  obstruction}, which is a multicurve $\Gamma_f$ associated to a
branched self-covering $f\colon(S^2,A)\selfmap$ that is not doubly
covered by a torus endomorphism and which has the property that $f$ is
combinatorially equivalent to a rational map if and only if
$\Gamma_f=\emptyset$.

Let us review the construction, omitting on purpose the case of maps
doubly covered by torus endomorphisms. The \emph{Teichm\"uller space}
$\mathscr T_A$ of $(S^2,A)$ is the space of complex structures on the
marked sphere $(S^2,A)$, or equivalently Riemannian metrics on
$S^2\setminus A$ of curvature $-1$. Thurston associates with
$f\colon(S^2,A)\selfmap$ a self-map
$\sigma_f\colon\mathscr T_A\selfmap$ defined by pulling back complex
structures through $f$. He shows (see~\cite{douady-h:thurston}) that
$\sigma_f$ is weakly contracting for the Teichm\"uller metric on
$\mathscr T_A$, so that (starting from an arbitrary point
$\tau\in\mathscr T_A$) either $\sigma_f^n(\tau)$ converges to a fixed
point (which is then a complex structure preserved by $f$, so $f$ is
combinatorially equivalent to a rational map) or degenerates to the
boundary of $\mathscr T_A$, in which case some curves on
$S^2\setminus A$ become very short in the hyperbolic metric defined by
$\sigma_f^n(\tau)$. The canonical obstruction $\Gamma_f$ is simply
defined as the collection of simple closed curves on $S^2\setminus A$
whose length goes to $0$ as $n\to\infty$.

Selinger gives in~\cite{selinger:canonical}*{Theorem~5.6} a topological
characterization of the canonical obstruction. As a consequence, we
deduce that Pilgrim's canonical obstruction is the union of the
\emph{Levy obstruction} (the multicurve along which $S^2$ is pinched
to produce the Levy decomposition) and the \emph{rational obstruction}
(the multicurve along which the expanding maps of the Levy
decomposition should be further pinched to give rational maps).
Selinger and Yampolsky show
in~\cite{selinger-yampolsky:geometrization}*{Main Theorem I} that the
canonical obstruction is computable.

\subsection{Remarks}
The main objects of study are \emph{bisets} (called
\emph{combinatorial bimodules} in~\cite{nekrashevych:ssg}), which we
generalize by allowing different groups $H,G$ to act on the left
and right respectively. Bisets may be thought of as generalizations of
group homomorphisms, up to pre- and post-composition by inner
automorphisms.  Indeed, if $\phi\colon H\to G$ is a group
homomorphism, written $h\mapsto h^\phi$, one associates with it the
$H$-$G$-set $B_\phi$, which, qua right $G$-set, is plainly $G$; the
left $H$-action is by
\[h\cdot b=h^\phi b.
\]
Conversely, if $B$ is a transitive $H$-$G$ biset (a general biset
splits as a disjoint union of its transitive components), then there
is a group $K$ and there are homomorphisms $\phi\colon K\to G$ and
$\psi \colon K\to H$ such that $B\cong B_\psi^{\vee}\otimes B_\phi$,
where $ B_\psi^{\vee}$ is the contragredient of $ B_\psi$,
see~\S\ref{ss:bisets}, so bisets may also be thought of
as correspondences of groups. In fact, to every topological
correspondence $Y\leftarrow Z\to X$ there is a naturally associated
$\pi_1(Y)$-$\pi_1(X)$-biset, independent of basepoints up to
isomorphism.

The main decision problems we study are, in the context of bisets, the
\emph{conjugacy problem} that can be asked in any category with
multiplication (given $B,C$, does there exist $X$ with $X B=C X$?),
\emph{witnessed conjugacy problem} (given $B,C$, find an $X$ with
$X B=C X$ or prove that there are none), and the \emph{centralizer
  problem} (given $B$, describe the set of $X$ with $B X=X B$). We
proceed from the best-behaved bisets (that of rational Thurston maps)
to the general case by following diverse reductions, in particular
introducing \emph{portraits of bisets} to add and erase marked points,
and \emph{trees of bisets} to glue and cut along multicurves.

In fact, by a general trick, only the centralizer problem needs to be
considered: assuming that the centralizer problem is solvable and given
$B,C$, compute the centralizer of $B\sqcup C$, and check whether it
contains an element that switches $B$ and $C$; if so, its restriction to
$C$ is a witness for conjugacy of $B$ and $C$. This is easy to check e.g.\
if centralizers are finitely generated subgroups of well-understood
groups. We show, however, that centralizers are only computable in the
weaker sense of being expressible as kernels of maps from well-understood
groups to Abelian groups. In particular, they can be infinitely generated,
see Example~\ref{ex:infinitely generated}. For extra clarity, we treat all
three decision problems in parallel.

We restrict ourselves to studying actions of the pure mapping class group:
punctures on our marked spheres may be permuted by Thurston maps, but are
fixed by the mapping classes. In this manner, the \emph{portrait} of
Thurston maps, namely the dynamics on their marked points, is preserved by
pre- and post-composition by mapping classes. One could extend the
statements and algorithms to non-pure mapping class groups, at the cost of
introducing finite groups in a few places. The action of non-pure mapping
classes would better capture the notions of conjugation, centralizer and
combinatorial equivalence of maps. In particular, the action of non-pure
(i.e.\ fractional) Dehn twists could lead to valuable systematic
constructions of Thurston maps, in the spirit of near-Euclidean Thurston
maps~\cite{cannon-floyd-parry-pilgrim:net}.

\subsection{Acknowledgments}
We are grateful to Kevin Pilgrim and Thomas Schick for enlightening
discussions.

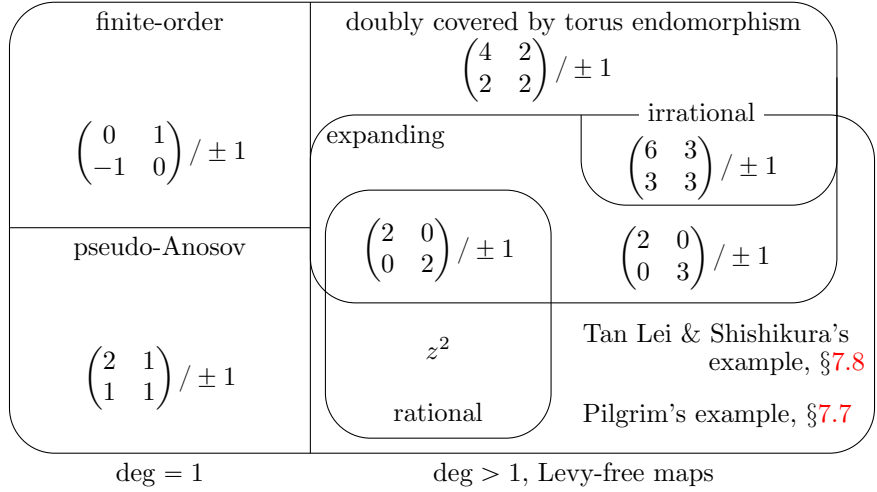
\begin{figure}
  \centering{\begin{tikzpicture}
      \draw[rounded corners=6mm] (11.5,2) -- (11.5,0) -- (0,0) -- (0,6) -- (11,6) -- (11,4) -- (11,2) -- (4,2) -- (4,4.5) -- (11.5,4.5) -- (11.5,2);
      \draw (4,0) -- +(0,6);
      \draw (0,3) -- +(4,0);
      \node[anchor=north] at (2,0) {$\deg=1$};
      \node[anchor=north] at (7.5,0) {$\deg>1$, Levy-free maps};
      \node[anchor=north] at (2,6) {finite-order};
      \node at (2,4) {$\begin{pmatrix}0&1\\-1&0\end{pmatrix}/\pm1$};
      \node[anchor=north] at (2,3) {pseudo-Anosov};
      \node at (2,1) {$\begin{pmatrix}2&1\\1&1\end{pmatrix}/\pm1$};
      \node[anchor=north] at (7.5,6) {doubly covered by torus endomorphism};
      \node at (7,5.1) {$\begin{pmatrix}4&2\\2&2\end{pmatrix}/\pm1$};
      \draw[rounded corners=6mm] (11,5) -- (11,3.3) -- (7.6,3.3) -- (7.6,4.5);
      \node at (9.2,3.8) {$\begin{pmatrix}6&3\\3&3\end{pmatrix}/\pm1$};
      \node[anchor=south,fill=white] at (9.2,4.3) {irrational};
      \node[anchor=north] at (5,4.5) {expanding};
      \draw[rounded corners=6mm] (4.2,0.2) rectangle (7.2,3.5);
      \node at (9.1,2.6) {$\begin{pmatrix}2&0\\0&3\end{pmatrix}/\pm1$};
      \node[anchor=south] at (5.7,0.3) {rational};
      \node at (5.7,1.4) {$z^2$};
      \node at (5.7,2.7) {$\begin{pmatrix}2&0\\0&2\end{pmatrix}/\pm1$};
      \node[anchor=west] at (7.5,1.6) {Tan Lei \& Shishikura's};
      \node[anchor=east] at (11.5,1.2) {example, \S\ref{ex:tls}};
      \node[anchor=west] at (7.5,0.5) {Pilgrim's example, \S\ref{ex:pilgrim}};
    \end{tikzpicture}}
  \caption{Geometric maps}\label{fig:geometric}
\end{figure}

\subsection{Notations}
Here are some notations that shall be used throughout the texts:
\begin{itemize}
\item The symmetric group on a set $S$ is written $S\perm$;
\item Concatenation of paths is written $\gamma\#\delta$ for ``first
  $\gamma$, then $\delta$''; inverses of paths are written $\gamma^{-1}$;
\item The identity map is written $\one$. Composition of maps,
  permutations etc.\ is in the algebraic, left-to-right order, unless
  explicitly written as $(f\circ g)(x)=f(g(x))$. The restriction of a
  map $f\colon A\to B$ to a subset $C\subseteq A$ is written
  $f\restrict C$. Self-maps are written $f\colon X\selfmap$ in
  preference to $f\colon X\to X$;
\item We write $\approx$ for isotopy of paths, maps etc, $\sim$ for
  conjugacy or combinatorial equivalence, and $\cong$ for isomorphism
  of algebraic objects;
\item We try to use similar fonts for similar objects: script $\CC$
  for multicurves, $\gf$ for graphs of groups, $G_z$ for its vertex
  and edge groups (even if the graph of groups is not called $G$),
  $\gfB$ for graph of bisets, $B_z$ for its vertex and edge
  bisets, usually Greek letters for functions.
\end{itemize}

We also establish a kind of ``dictionary'' between topological and
algebraic notions:

\begin{tabular}{l|l}
  Topology & Algebra\\[2pt]\hline\rule{0pt}{3ex}\relax

  Continuous map & Right-principal biset $B(f)$\\
  Covering map & Left-free right-principal biset\\
  Composition of maps $g\circ f$ & Tensor product of bisets $B(f)\otimes B(g)$\\
  Topological correspondence $(f,i)$ & Biset $B(f,i)=B(i)^\vee\otimes B(f)$\\
  Covering pair ($f$ is a covering) & Left-free biset\\
  Decomposition of correspondence & Graph of bisets\\[1ex]

  Punctured sphere $(S^2,A)$ & Sphere group $G=\langle \gamma_1,\dots,\gamma_n\mid\gamma_1\cdots\gamma_n\rangle$\\
  Puncture $a_i\in A$ & Peripheral conjugacy class $\gamma_i^G$ in $G$\\
  Mapping class group $\Mod(S^2,A)$ & Outer automorphism group $\Mod(G)$\\[1ex]

  Multicurve $\CC$ & Family of essential conjugacy classes in $G$\\
  Decomposition of $S^2$ along $\CC$ & Decomposition of $G$ as sphere tree of groups $\gf$\\
  $\CC\sqcup A$ & Distinguished conjugacy classes $X$ of $\gf$\\
  Homeomorphism $(S^2,B,\DD)\to(S^2,A,\CC)$ & Conjugator $\subscript\gfY\gfI_\gf$ between trees of groups\\
  Group of Dehn twists along $\CC$ &$\Z^{\CC}$\\
  Stabilizer $\Mod(S^2,A,\CC)$ of $\CC$ & Stabilizer $\Mod(\gf)$ of $\gf$\\[1ex]

  Branched covering $f\colon(S^2,B)\to(S^2,A)$ & Sphere $H$-$G$-biset $B(f)$\\
  Isotopy rel $A$ & Isomorphism of sphere bisets\\
  Mapping class biset $M(f)$ & Mapping class biset $M(B)$\\[1ex]

  Thurston map $f\colon(S^2,A)\selfmap$ & Sphere $G$-$G$-biset $B(f)$\\
  Expanding map $f$ & Contracting biset $B(f)$\\
  Restriction $f\colon A\selfmap$ & Portrait $B_*\colon A\selfmap$\\
  Combinatorial equivalence & Conjugacy of sphere bisets\\
  Centralizer $Z(f)$ & Centralizer $Z(B)$\\
  Extra marked points of expanding map $f$ & Portrait of bisets for $B(f)$\\
  Torus endomorphism $\R^2/\Z^2\selfmap$ & Biset $\subscript{\Z^2}B_{\Z^2}$ of a linear map\\
  Map doubly covered by $\R^2/\Z^2\selfmap$ & Crossed product $\subscript{\Z^2}B_{\Z^2}\rtimes\{\pm1\}$\\[1ex]

  Decomposition of $f$ along multicurve & Decomposition of $B(f)$ as tree of bisets $\gfB$\\
  Sub-mapping class biset $M(f,B,A,\DD,\CC)$ & $M(\gfB)$\\
  Renormalization of $f$ w.r.t. $\CC$ & Return bisets of $\gfB$
\end{tabular}

\renewcommand\thesection{\Roman{section}}
\section{Bisets and van Kampen's theorem~\cite{bartholdi-dudko:bc1}}\label{ss:bisets}
Let $f\colon Y\to X$ be a continuous map between topological
spaces. Fix basepoints $\dagger\in Y$ and $*\in X$, and consider the
fundamental groups $H=\pi_1(Y,\dagger)$ and $G=\pi_1(X,*)$. Then the
map $f$ may be encoded into an \emph{$H$-$G$-biset}, namely a set
$B(f)$ with commuting left $H$-action and right $G$-action. As a set,
$B(f)$ is the set of homotopy classes of paths, in $X$, from
$f(\dagger)$ to $*$. The actions of $H$ and $G$ are respectively given
by pre-catenation of the $f$-image and by post-catenation. To recall
the acting groups, we sometimes write $\subscript H B(f)_G$. Bisets
can be multiplied; the \emph{product} $B\otimes_G C$ of an
$H$-$G$-biset $B$ with a $G$-$F$-biset $C$ is the $H$-$F$-biset
$(B\times C)/\{(b g,c)=(b,g c)\}$. The \emph{contragredient} of the
$H$-$G$-biset $B$ is the $G$-$H$-biset $B^\vee$, which is $B$ as a set
with actions $g\cdot(b^\vee)\cdot h=(h^{-1}b g^{-1})^\vee$.

An \emph{intertwiner} from an $H$-$G$-biset $B$ to an $H'$-$G'$-biset
$B'$ is a map $\beta\colon B\to B'$ and a pair of homomorphisms
$\gamma\colon H\to H'$ and $\alpha\colon G\to G'$ with
\[\beta(h b g)=\gamma(h)\beta(b)\alpha(g)\text{ for all }h\in H,b\in B,g\in G.\]
If $H=G$, $H'=G'$ and $\gamma=\alpha$, then the intertwiner is a
\emph{semiconjugacy}. If $H=H'$ and $G=G'$ and $\gamma=\one$ and
$\alpha=\one$, then the intertwiner is a \emph{morphism}. \emph{Congruences}, \emph{conjugacies}, and
\emph{isomorphisms} are invertible intertwiners, semiconjugacies, and
morphisms respectively.
\[\begin{tikzcd}
    \begin{array}{c}\text{semiconjugacies}\\ (\alpha,\beta,\alpha)\\ G=H, G'=H'\end{array}\rar[draw=none,"\displaystyle\subset" description]\dar[draw=none,"\displaystyle\cup" description] 
   & \begin{array}{c}\text{intertwiners}\\ (\gamma,\beta,\alpha)\\ \end{array}\rar[draw=none,"\displaystyle\supset" description]\dar[draw=none,"\displaystyle\cup" description] &
    \begin{array}{c}\text{morphisms}\\ (\one,\beta,\one)\\ G=G', H=H'\end{array}\dar[draw=none,"\displaystyle\cup" description]\\
    \begin{array}{c}\text{conjugacies}\end{array}\rar[draw=none,"\displaystyle\subset" description] &
    \begin{array}{c}\text{congruences}\\ \exists (\gamma^{-1},\beta^{-1}, \alpha^{-1})\end{array}\rar[draw=none,"\displaystyle\supset" description] &
    \begin{array}{c}\text{isomorphisms}\end{array}
  \end{tikzcd}
\]
\noindent Note that if $G=H$ then every morphism is also a semiconjugacy.

Let $X,Y$ be path connected topological spaces. Bisets are well
adapted to encode more general objects than continuous maps $Y\to X$,
namely \emph{topological correspondences}.  These are triples
$(Z,f,i)$ consisting of a topological space $Z$ and continuous maps
$f\colon Z\to X$ and $i\colon Z\to Y$, and are simply written
$Y\leftarrow Z\to X$. If $Z$ is path connected, then the biset of the
correspondence is $B(f,i)=B(i)^\vee\otimes B(f)$; in general, it is
the disjoint union of the bisets on all path connected components of $Z$.

\subsection{Coverings and left-free bisets}\label{ss:coverings}
The biset $\subscript H{B(f)}_G$ of a continuous map $f\colon Y\to X$
is, by construction, isomorphic to $G_G$ qua right $G$-set. If
furthermore $f$ is a covering, say of degree $d$, then
$\subscript H {B(f)}$ is free of degree $d$ qua left $H$-set. In
particular, if in the topological correspondence $Y\leftarrow Z\to X$
the map $f$ is a degree-$d$ covering, then $B(f,i)$ is left-free of
degree $d$. The correspondence is called a \emph{covering pair}.

Choose then a subset $S\subseteq B(f,i)$ of cardinality $d$ that
intersects once every left $H$-orbit; such a subset is called a
\emph{basis}. Using the isomorphism $\subscript H{B(f,i)}=H\times S$, we may
write the right $G$-action on $B(f,i)$ in the form
\[s\cdot g=h\cdot s'\text{ for some }h\in H,s'\in S.\]
This is a map $S\times G\to H\times S$, which (writing $S\perm$ for
the permutation group on $S$) yields a group homomorphism
\[\psi\colon G\mapsto H^S\rtimes S\perm\]
called the \emph{wreath recursion} of $B(f,i)$. Writing
$S=\{\ell_1,\dots,\ell_d\}$, we may write $\psi$ as
\[g\mapsto\pair{h_1,\dots,h_d}\pi\]
for elements $h_1,\dots,h_d$ and a permutation $\pi$. It is sufficient
to specify $\psi$ on generators of $g$, and these data are called a
\emph{presentation} of the biset.

We shall see in~\S\ref{ss:examples} many examples of biset
presentations. We stress here that they are eminently computable, and
in particular with pencil and paper. Here is the concrete recipe, for
a covering correspondence $Y\leftarrow Z\to X$. Fix basepoints
$*\in X$ and $\dagger\in Y$, and write $G=\pi_1(X,*)$ and
$H=\pi_1(Y,\dagger)$. Choose for each $z\in f^{-1}(*)$ a path $\ell_z$
in $Y$ from $i(z)$ to $\dagger$, and set
$S=\{\ell_z\mid z\in f^{-1}(*)\}$.  For every (generator) $g\in G$,
represented as a curve $g\colon[0,1]\to X$, and for every
$\ell_z\in S$, there exists a unique lift $\tilde g_z\colon[0,1]\to Z$
of $g$ that starts at $z$. Let $z'$ be the endpoint of $\tilde g_z$.
Then the concatenation
$h_z\coloneqq\ell_z^{-1}\#(i\circ \tilde g_z)\#\ell_{z'}$ is a loop at
$\dagger$, and its class in $H$ is independent of the choice of
representative for $g$. The wreath recursion of $B(f,i)$ is the map
$g\mapsto\pair{h_{z_1},\dots,h_{z_d}}\pi$ with $\pi\in S\perm$ the
permutation $z\mapsto z'$.

\subsection{Graphs of bisets}
The van Kampen theorem expresses the fundamental group of a
topological space in terms of the fundamental groups of subspaces. A
convenient algebraic object that captures the data is a graph of
groups -- a graph decorated with groups such that each edge group has
two morphisms into the neighboring vertex groups. It is therefore
convenient~\cite{serre:trees} to double each edge -- to replace it with
a pair of directed edges of the opposite orientation.

We view \emph{graphs} as sets $\gf$ endowed with two maps
$x\mapsto x^-$ and $x\mapsto\overline x$, with axioms
$\overline{\overline x}=x$ and $(x^-)^-=x^-$ and
$x^-=x\Leftrightarrow\overline x=x$. The vertex set $V(\gf)$ is then
$\{x\mid x^-=x\}$, and the edge set $E(\gf)$ is $\gf\setminus
V(\gf)$. The object $\overline x$ is called the \emph{reverse} of
$x$. Setting $x^+\coloneqq (\overline x)^-$, the vertices $x^-$ and
$x^+$ are respectively the \emph{origin} and \emph{terminus} of $x$. A
\emph{graph morphism} is a map $h\colon \gf \to \gf'$ such that
$h(x^-)=h(x)^-$ and $h(\overline x)=\overline x$ for all $x\in
\gf$. Note that the image of a vertex is a vertex while the image of
an edge is either an edge or a vertex.

Recall (e.g.\ from~\cite{serre:trees}*{\S4}) that a graph of groups is
a graph $\gf$ with a group $G_x$ associated with each $x\in\gf$, and
for each edge $x\in\gf$ homomorphisms $G_x\to G_{x^-}$ and
$G_x\to G_{\overline x}$, written respectively $g\mapsto g^-$ and
$g\mapsto\overline g$.  Its \emph{fundamental group} $\pi_1(\gf,v)$ is
the set of group-decorated loops $g_0e_1g_1\cdots e_n g_n$ with
$e_1\cdots e_n$ a loop at $v\in V(\gf)$ in $\gf$ and
$g_i\in G_{e_i^+}$ for all $i$, up to the relations $e g^+=g^-e$ for
all edges $e\in E$ and all $g\in G_e$. More generally, given
$v, w \in \gf$ we define $\pi_1(\gf, v,w )$ as the set of
group-decorated paths $g_0e_1g_1\cdots e_n g_n$ with $e_1\cdots e_n$ a
path from $v$ to $w$ and $g_i\in G_{e_i^+}$ for all $i$, up to the
relations $e g^+=g^-e$ as above. The set $\pi_1(\gf,v,w)$ is naturally
a $\pi_1(\gf,v)$-$\pi_1(\gf,w)$ biset. Given $p\in \pi_1(\gf,u,v)$ and
$q\in \pi_1(\gf,v,w)$ their product $p q$ is in $\pi_1(\gf,u,w)$.

\begin{defn}[Graph of bisets]\label{defn:graphofbisets}
  Let $\gf,\gfY$ be two graphs of groups. A \emph{graph of
    bisets} $\subscript \gfY\gfB_\gf$ between them is the
  following data:
  \begin{itemize}
  \item a graph $\gfB$;
  \item graph morphisms $\lambda\colon\gfB\to\gfY$ and
    $\rho\colon\gfB\to\gf$;
  \item for every $z\in\gfB$, a
    $G_{\lambda(z)}$-$G_{\rho(z)}$-biset $B_z$, an intertwiner
    $()^-\colon B_z\to B_{z^-}$ with respect to the homomorphisms
    $G_{\lambda(z)}\to G_{\lambda(z)^-}$ and
    $G_{\rho(z)}\to G_{\rho(z)^-}$, and an intertwiner
    $\overline{()}\colon B_z\to B_{\overline z}$ with respect to the
    homomorphisms $G_{\lambda(z)}\to G_{\overline{\lambda(z)}}$ and
    $G_{\rho(z)}\to G_{\overline{\rho(z)}}$. These intertwiners satisfy
    natural axioms: the composition
    $B_z\to B_{\overline z}\to B_{\overline{\overline z}}=B_z$ is the
    identity for every $z\in\gf$, and if $z\in V(\gf)$, then the
    homomorphisms $B_z\to B_{z^-}$ and $B_z\to B_{\overline z}$ are
    the identity. For $b\in B_z$ we write $b^+={\overline b}^-$.
  \end{itemize}
  We call $\gfB$ a \emph{$\gfY$-$\gf$-biset}.
\end{defn}

\begin{defn}[Fundamental biset of graph of bisets]\label{defn:p1biset}
  Let $\gfB$ be a $\gfY$-$\gf$-biset; choose
  $*\in V(\gf)$ and $\dagger\in V(\gfY)$. Write
  $G=\pi_1(\gf,*)$ and $H=\pi_1(\gfY,\dagger)$. The
  \emph{fundamental biset} of $\gfB$ is an $H$-$G$-biset
  $B=\pi_1(\gfB,\dagger,*)$, constructed as follows.
  \begin{equation}\label{eq:freebiset}
    B=\frac{\bigsqcup_{z\in V(\gfB)}\pi_1(\gfY,\dagger,\lambda(z))\otimes_{G_{\lambda(z)}}B_z\otimes_{G_{\rho(z)}}\pi_1(\gf,\rho(z),*)}{\left\{q b^-p = q\lambda(z)b^+\overline{\rho(z)}p\quad\forall\begin{array}{c}q\in\pi_1(\gfY,\dagger,\lambda(z)^-),b\in B_z,\\p\in\pi_1(\gf,\rho(z)^-,*),z\in E(\gfB)\end{array}\right\}}.
  \end{equation}
  In other words, elements of $B$ are sequences
  $h_0y_1h_1\cdots y_n\,b\,x_1\cdots g_{m-1} x_n g_n$ subject to the
  equivalence relations used previously to define $\pi_1(\gf)$, as
  well as
  $y_n h b^-g x_1\leftrightarrow y_n h\lambda(z)b^+\overline{\rho(z)}g x_1$
  for all $z\in\gfB$, $b\in B_z$, $h\in G_{\lambda(z)^-}$,
  $g\in G_{\rho(z)^-}$.
\end{defn}
\noindent Up to congruence, the fundamental biset is independent on
the choice of basepoints:
$\pi_1(\gfB, \dagger,*)= \pi_1(\gfY,
\dagger,\dagger')\otimes \pi_1(\gfB, \dagger',*') \otimes
\pi_1(\gf, *',*)$.

The definition is a bit unwieldy, but it has a simpler version in case
the graph of bisets is \emph{left-fibrant},
see~\cite{bartholdi-dudko:bc1}*{Definition~\ref{bc1:dfn:FibrationGrBis}}. Such
graphs of bisets arise from correspondences where one of the maps is a
fibration (such as a covering). A left-fibrant graph of bisets
possesses a lifting property: any
$p b q \in \pi_1(\gfB, \dagger,*)$ can be rewritten in an
essentially unique way as $p q' b''$ for some $b'$ in a vertex
biset. Thus~\eqref{eq:freebiset} takes the form
\cite{bartholdi-dudko:bc1}*{\eqref{bc1:eq:cor:NormalForm}}
\begin{equation}
  \label{eq:cor:NormalForm}
  \pi_1(\gfB,\dagger,*)=\bigsqcup_{z\in \rho^{-1}(*)}\pi_1(\gfY,\dagger,\lambda(z))\otimes_{G_{\lambda(z)}}B_z,
\end{equation}
with right action given by lifting of paths in $\pi_1(\gf,*)$.

A biset $\subscript H B_G$ is \emph{biprincipal} if both actions are free and
transitive. A graph of bisets $\subscript \gfY \gfI_{\gf}$ is
\emph{biprincipal} if
\begin{enumerate}
\item $\lambda\colon \gfI\to\gfY$ and $\rho\colon
  \gfI\to\gf$ are graph isomorphisms; and
\item $B_z$ are biprincipal for all objects $z\in \gfI$.\qedhere
\end{enumerate}

\noindent We use this notion to define congruence and conjugacy of
graphs of bisets:
\begin{defn}\label{defn:conjugategob}
  Two graphs of groups $\gfY$, $\gf$ are called \emph{congruent} if
  there is a biprincipal graph of bisets $\subscript\gfY\gfI_\gf$.

  Isomorphism of graphs of bisets is meant in the strongest possible
  sense: isomorphism of the underlying graphs, and isomorphisms of the
  respective bisets. There is a general notion of tensor product of
  graphs of bisets, which in the cases below simply amounts to
  tensoring the vertex and edge bisets together.

  Two graphs of bisets $\subscript\gfY\gfB_\gf$ and
  $\subscript {\gfY'} \gfC_{\gf'}$ are \emph{congruent} if there are
  biprincipal graph of bisets $\subscript\gfY\gfI_{\gfY'}$ and
  $\subscript\gf{\mathfrak L}_{\gf'}$ such that
  $\subscript \gfY\gfB_\gf\otimes \mathfrak L$ and
  $\gfI \otimes \subscript{\gfY'}\gfC_{\gf'} $ are isomorphic.

  Two graphs of bisets $\subscript \gf\gfB_\gf$ and
  $\subscript \gfY\gfC_\gfY$ are \emph{conjugate} if
  there is a biprincipal graph of bisets
  $\subscript \gf\gfI_\gfY$ such that
  $\gfB\otimes_\gf \gfI$ and
  $\gfI\otimes_\gfY\gfC$ are isomorphic.
\end{defn}

\subsection{\myboldmath Graphs of bisets from $1$-dimensional covers}
\begin{defn}[Finite $1$-dimensional covers]\label{defn:1dimcovers}
  Consider a path connected space $X$, covered by a finite collection
  of path connected (not necessarily open) subspaces $(X_v)_{v\in
    V}$. It is a \emph{finite $1$-dimensional cover} of $X$ if
  \begin{itemize}
  \item for every $u,v\in V$ and for every path connected component $X'$ of
    $X_{u}\cap X_{v}$ there are an open neighbourhood
    $\widetilde X'\supset X'$ and an $X_{w}\subset X'$ such that
    $X_{w}\hookrightarrow \widetilde X'$ is a homotopy equivalence;
  \item if $X_{u}\subseteq X_{v}\subseteq X_{w}$ then $u=v$ or $v=w$.
  \end{itemize}
  We order $V$ by writing $u<v$ if $X_{u}\varsubsetneqq X_{v}$.
\end{defn}

\begin{defn}[Graphs of groups from covers]\label{defn:gog_1dimcovers}
  Consider a path connected space $X$ with a $1$-dimensional cover
  $(X_v)_{v\in V}$. It has an associated graph of groups $\gf$, defined as
  follows. The vertex set of $\gf$ is $V$. For every pair $u<v$ there are
  edges $e$ and $\overline e$ connecting $u=e^-=\overline e^+$ and $v=e^+=\overline e^-$,
  and we let $E$ be the set of these edges. Set $\gf=V\sqcup E$.

  Choose basepoints $*_v\in X_v$ for all $v\in V$. Choose for each
  edge $e$ a path $\ell_e$ from $*_{e^-}$ to $*_{e^+}$ such that
  $\ell_{\overline e}=\ell^{-1}_{e}$.  Set $G_v\coloneqq\pi_1(X_v,*_v)$ for
  every $v\in V$. For every edge $e$ with $e^-<e^+$ set
  $G_e\coloneqq G_{e^-}$; define $G_e\to G_{e^-}\coloneqq\one$ and
  $G_e\to G_{e^+}$ by $\gamma\mapsto\ell_e^{-1}\#\gamma\#\ell_e$. For
  every edge $e$ with $e^->e^+$ define $G_{e}\coloneqq G_{\overline e}$ and
  define morphisms $G_e\to G_{e^-}$ and $G_e\to G_{e^+}$ as
  $G_{\overline e}\to G_{\overline e^+}$ and $G_{\overline e}\to G_{\overline e^-}$
  respectively.
\end{defn}

Consider now a correspondence $(Z,f,i)$, with $f\colon Z\to X$ and
$i\colon Z\to Y$, between path connected spaces $X$ and $Y$.  Suppose
that $(U_{\alpha})$, $(V_\beta)$, and $(W_\gamma)$ are finite
$1$-dimensional covers of $X$, $Y$, and $Z$ respectively,
\emph{compatible with} $f$ and $i$: for every $\gamma$ there are
$\lambda(\gamma)$ and $\rho(\gamma)$ such that
$f(W_\gamma)\subset U_{\rho(\gamma)}$ and
$i(W_\gamma)\subset V_{\lambda(\gamma)}$.  Then the graph of
bisets $\subscript \gfY\gfB_\gf$ of $(f,i)$ with respect to
the above data is as follows:
\begin{itemize}
\item the graphs of groups $\gf$ and $\gfY$ are constructed as
  in Definition~\ref{defn:gog_1dimcovers} using the covers
  $(U_\alpha)$ and $(V_\beta)$ of $X,Y$ respectively. Choices of paths
  $\ell_e$, $m_e$ were made for edges $e$ in $\gf$, $\gfY$
  respectively;
\item the underlying graph of $\gfB$ is similarly constructed
  using the cover $(W_\gamma)$ of $Z$. For every vertex
  $z\in \gfB$ the biset
  $\subscript{G_{\lambda(z)}}{(B_z)}_{G_{\rho(z)}}$ is
  $B(f\restrict{W_z}, i\restrict{W_z})$;
\item for every edge $e\in \gfB$ representing the embedding
  $W_{z'}\varsubsetneqq W_z$ the biset $B_e$ is $B_{z'}$, and if $e$
  is oriented so that $e^-=z'$ then the intertwiners $()^\pm$ are the
  maps $()^-=\one\colon B_e\to B_{z'}$ and $()^+\colon B_e\to B_z$
  given by
  $(\gamma^{-1},\delta)\mapsto(m_{\lambda(e)}^{-1}\#\gamma^{-1},\delta\#\ell_{\rho(e)})$
  in the description of $B_e$ as
  $B(i\restrict{W_{e^-}})^\vee\otimes B(f\restrict{W_{e^-}})$.
\end{itemize}
The graphs of groups $\gf,\gfY$ and the graph of bisets
$\gfB$ are independent of the choices of basepoints and
connecting paths $\ell_e,m_e$, up to congruence.

\begin{thm}[Van Kampen's theorem for correspondences]\label{thm:vankampenbis}
  Let $(f,i)$ be a topological correspondence from a path connected
  space $Y$ to a path connected space $X$, and let
  $\subscript\gfY\gfB_\gf$ be the graph of bisets subject to
  compatible finite $1$-dimensional covers of spaces in question.
  
  Then for every $v\in\gfY$ and $u\in\gf$ we have an
  isomorphism
  \[B(f,i,\dagger_v,*_u)\cong \pi_1(\gfB, v,u),\]
  where $\dagger_v$ and $*_u$ are basepoints. 
\end{thm}
If in Theorem~\ref{thm:vankampenbis} the map $f\colon Z\to Y$ is a
covering and all restrictions $f\colon W_\gamma\to U_{\rho(\gamma)}$
are covering maps, then $\gfB$ is left-fibrant. In this case
its fundamental biset is computed by~\eqref{eq:cor:NormalForm}.

\subsection{Hubbard trees}\label{ss:hubbard tree}
Consider a polynomial $p(z)\in\C[z]$. We will see how to construct a
graph of bisets out of $p$'s ``Hubbard tree'', see
Figure~\ref{fig:julia of z^2+i}. We first recall some basic definitions and
properties; see~\cites{douady-h:edpc1,douady-h:edpc2} for details.

The \emph{post-critical set} $P(p)$ is the forward orbit of $p$'s
critical values:
\[P(p)\coloneqq\{p^n(z)\mid p'(z)=0,\,n\ge1\}.\]
The polynomial $p$ is \emph{post-critically finite} if $P(p)$ is
finite. The \emph{Julia set} $J(p)$ of $p$ is the boundary of the
filled-in Julia set $K(p)$, and the \emph{Fatou set} is its
complement:
\[K(p)\coloneqq\{z\in \C\mid \{p^n(z)\mid n\in\N\}\text{ is bounded}\},\quad J_c=\partial K_c,\quad F(p)=\C\setminus J(p).
\]
The \emph{Hubbard tree} of $p$ is the smallest tree in $K(p)$ that
contains $P(p)$ and all of $p$'s critical points; it intersects $F(p)$
along radial arcs.  It is a simplicial graph, with some distinguished
vertices corresponding to $P(p)$. All its vertices have an
\emph{order} in $\N\cup\{\infty\}$: by definition, $p$ behaves locally
as $z\mapsto z^{\deg_z(p)}$ at a point $z\in\C$, and
\[\ord(v)=\lcm\{\deg_z(p^n)\mid n\ge0,z\in p^{-n}(v)\}.\]
Thus in particular $\ord(v)=\infty$ if $v$ is critical and periodic.

This order function defines an \emph{orbispace} structure on $\C$,
see~\ref{ss:orbispheres}: a topological space with the extra data of a
non-trivial group $G_v$ attached at a discrete set of points $v$, in
canonical neighbourhoods of which the fundamental group is isomorphic
to $G_v$. In our situation, the group attached to $v\in P(p)$ is
cyclic of order $\ord(v)$.

For each $z\in P(p)$, let $\gamma_z$ denote a small loop around $z$, and
identify $\gamma_z$ with a representative of a conjugacy class in
$\pi_1(\C\setminus P(p),*)$, see~\S\ref{defn:sphere groups}. It follows
that the fundamental group of the orbispace defined by $\ord$ is given as
follows:
\[G_p = \pi_1(\C\setminus P(p),*)/\langle \gamma_z^{\ord(z)}:z\in P(p)\rangle.\]
The biset $B(p)$ of $p$ is the biset of the orbispace-correspondence
\[\left(p\colon (\C, p^{-1}(P(p)))\to (\C, P(p)) , (\C,
    p^{-1}(P(p)))\hookrightarrow (\C, P(p)) \right).
\]
Since $p$ is an orbispace-covering, the biset $B(p)$ is left-free,
see~\S\ref{ss:coverings}.

Out of the Hubbard tree $T$ of $p$, we may construct a graph of groups
$\gf$. Qua graph, $\gf$ is $T$. A cyclic group of order $\ord(v)$ is
attached to every vertex $v\in T$, and the edges of $\gf$ all carry
trivial groups.

We may also construct a graph of bisets $\subscript\gf{\mathfrak T}_\gf$ as
follows, using the Hubbard tree $T$.  The underlying graph of
$\mathfrak T$ is $p^{-1}(T)$ and $\rho\colon\mathfrak T\to\gf$ is
given by the covering map $p\colon p^{-1}(T)\to T$. The map
$\lambda\colon\mathfrak T\to\gf$ is  the canonical retraction of
$p^{-1}(T)$ to its subtree $T$. There is a degree-$\deg_z(p)$ cyclic biset
attached to each vertex $z\in\mathfrak T$, and trivial bisets attached
to edges of $\mathfrak T$. The biset inclusions are determined by an
additional piece of information: \emph{angles} between incoming edges
at vertices of $T$. We shall give in
Algorithm~\ref{algo:hubbardtree2biset} a procedure that constructs the
graph of cyclic bisets directly out of the combinatorial data of $T$,
see also~\cite{richter:hubbard}.

\begin{prop}
  Let $p$ be a complex polynomial. Then the groups $G_p$ and
  $\pi_1(\gf,*)$ are isomorphic, and the bisets $B(p)$ and
  $\pi_1(\mathfrak T)$ are conjugate via the group isomorphism
  between $G_p$ and $\pi_1(\gf,*)$.
\end{prop}

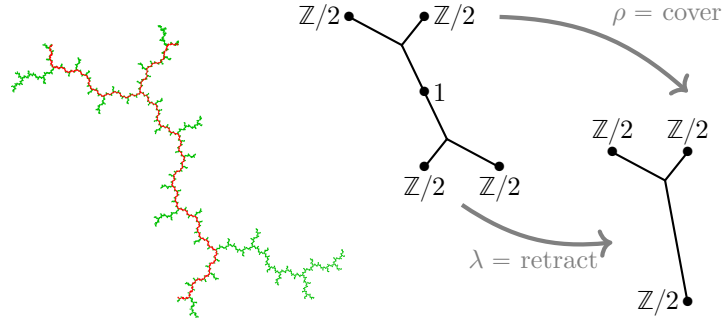
\begin{figure}[h]
\begin{center}
  \begin{tikzpicture}
    \begin{scope}
      \node[anchor=south west,inner sep=0] (z2pi) at (0,0) {\includegraphics[width=0.35\textwidth]{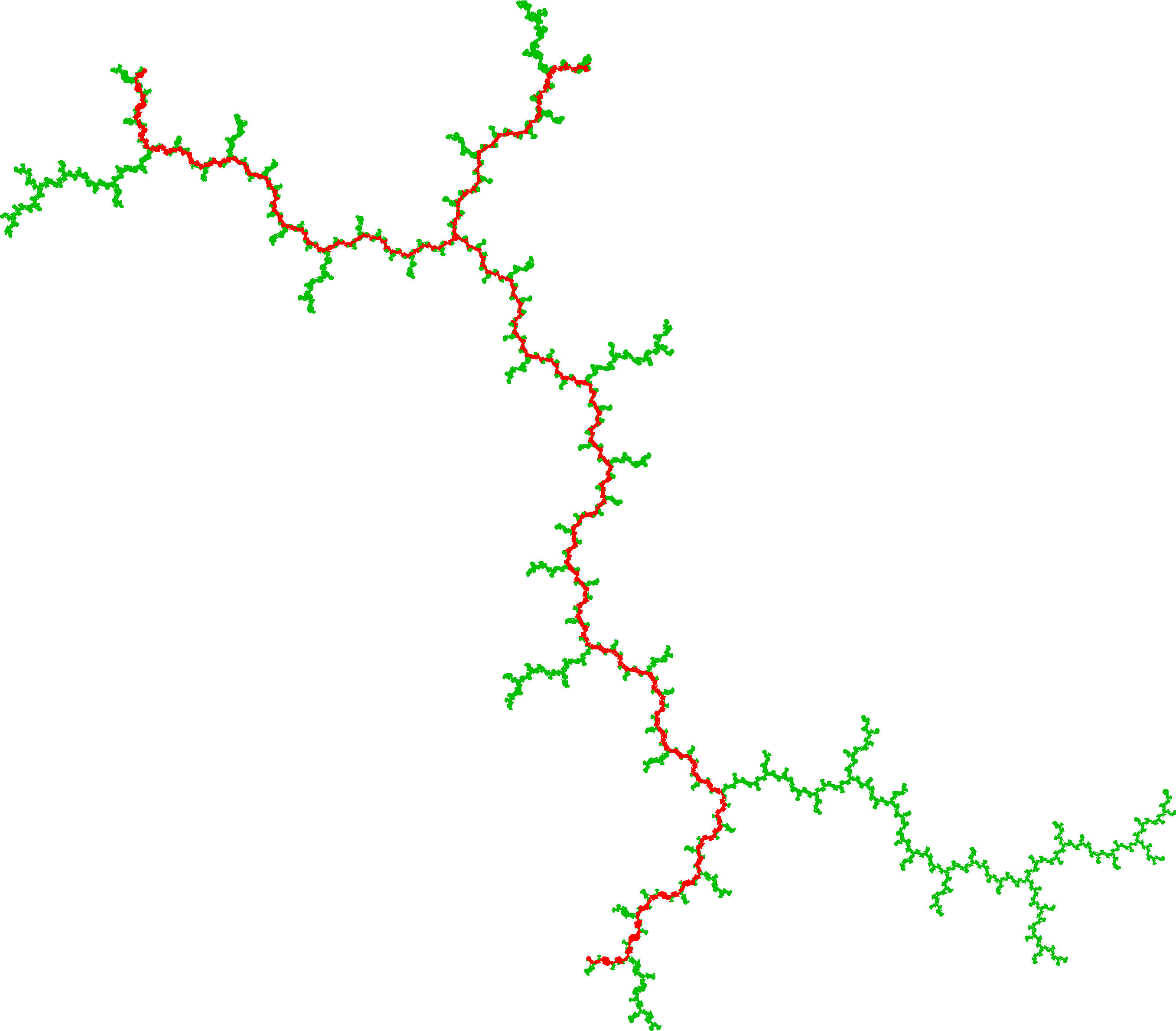}};
      \begin{scope}[x={(z2pi.south east)},y={(z2pi.north west)}]
        \coordinate (alpha_z2pi) at (0.387,0.232);
      \end{scope}
    \end{scope}

    \begin{scope}[xshift=5.5cm,yshift=3cm,thick,inner sep=0pt]
      \coordinate (alpha) at (-0.3,0.625);
      \coordinate (malpha) at (0.3,-0.625);
      \draw (-1,1) node[label={left:$\Z/2$}] {$\bullet$} -- (alpha)
      -- (0,1) node[label={right:$\Z/2$}] {$\bullet$}
      (alpha) -- (0,0) node[label={right:$1$}] {$\bullet$}
      -- (malpha) -- (0,-1) node[label={below:$\Z/2$}] {$\bullet$}
      (malpha) -- (1,-1) node[label={below:$\Z/2$}] {$\bullet$};
    \end{scope}

    \begin{scope}[xshift=9cm,yshift=1.2cm,thick,inner sep=0pt]
      \coordinate (alpha) at (-0.3,0.625);
      \draw (-1,1) node[label={above:$\Z/2$}] {$\bullet$} -- (alpha)
      -- (0,1) node[label={above:$\Z/2$}] {$\bullet$};
      \draw (0,-1) node[label={left:$\Z/2$}] {$\bullet$} -- (alpha);
    \end{scope}

    \path[ultra thick,gray,->] (6.5,4) edge [bend left=20] node[auto] {$\rho=$ cover} (9,3);
    \path[ultra thick,gray,->] (6,1.5) edge [bend right=20] node[below] {$\lambda=$ retract} (8,1);
  \end{tikzpicture}
\end{center}
\caption{The Julia set of $p(z)=z^2+i$, its Hubbard tree (in red), and its
  associated graph of bisets above graph of groups. The vertex groups and
  bisets are indicated on the picture, all edge bisets and groups are
  trivial, and the embeddings of edge bisets into vertex bisets are
  irrelevant.}\label{fig:julia of z^2+i}
\end{figure}

\section{Spheres and their decompositions~\cite{bartholdi-dudko:bc2}}
We specialize the maps we consider to branched coverings
$(S^2,B)\to(S^2,A)$ between spheres with finitely many marked
points. Decompositions of $S^2$ are given by \emph{multicurves},
namely collections of disjoint simple closed curves on $S^2$. Once
these small curves are pinched to points, one obtains a new collection
of topological spheres attached at these points.

The main results of this part are decidability statements: conjugacy
and isotopy questions for sphere maps can be translated to group
theory. In the presence of a multicurve, these questions can be
reduced to simpler questions on the restrictions of the maps to the
topological spheres in the complement of the multicurve.

\subsection{Sphere groups and maps}\label{ss:spheregroups}
Spheres with marked points are described, within group theory, by
their fundamental group. Let $(S^2,A)$ be a topological sphere, marked
by a finite subset $A\subset S^2$ with $\#A\ge 2$. Choose a basepoint
$*\in S^2\setminus A$. Then the fundamental group of $S^2\setminus A$
may be computed as follows. Order $A=\{a_1,\dots,a_n\}$. Choose for
each $a_i\in A$ a path $\overline\gamma_i$ from $*$ to $a_i$, such
that $\overline\gamma_i$ and $\overline\gamma_j$ intersect only at $*$
for $i\neq j$, and such that the $\overline\gamma_i$ are cyclically
ordered as $\overline\gamma_1,\dots,\overline\gamma_n$
counterclockwise around $*$. Let $\gamma_i$ be the loop at $*$ that
travels on the right of $\overline\gamma_i$, circles once
counterclockwise around $a_i$, and returns to $*$ on the right of
$\overline\gamma_i^{-1}$. We then have
\begin{equation}\label{eq:spheregroup}
  G=\pi_1(S^2\setminus A,*)=\langle \gamma_1,\dots,\gamma_n\mid \gamma_1\cdots\gamma_n\rangle.
\end{equation}

A cycle (loop without basepoint) in $S^2\setminus A$ is represented by
a conjugacy class in $G$. The group $G$ comes with extra data: the
collection $\{\gamma_1^G,\dots,\gamma_n^G\}$ of conjugacy classes,
called \emph{peripheral conjugacy classes}, defined by the property
that $\gamma_i^G$ represents can be homotoped to a loop
circling $a_i$ once counterclockwise.

A \emph{sphere map} $f\colon(S^2,B)\to(S^2,A)$ is a branched covering
$f\colon S^2\to S^2$ such that
$f(B\cup\{\text{critical points}\})\subseteq A$. The \emph{biset}
$B(f)$ is the biset $B(f,i)$ of the correspondence
\[(f\colon S^2\setminus f^{-1}(A)\to S^2\setminus A, i\colon S^2\setminus
  f^{-1}(A)\hookrightarrow S^2\setminus B).
\]
Fixing basepoints $\dagger\in S^2\setminus B$ and
$*\in S^2\setminus A$ for the fundamental groups
$H\coloneqq\pi_1(S^2\setminus B,\dagger)$ and
$G\coloneqq\pi_1(S^2\setminus A,*)$, the $H$-$G$-biset $B(f)$ may be
concretely seen as
\[B(f)=\{\gamma\colon[0,1]\to S^2\setminus B\mid \gamma(0)=\dagger,
  f(\gamma(1))=*\}.
\]
The left- and right-actions are respectively by pre-catenation and
post-catenation of the appropriate $f$-lift.

Two sphere maps $f_0,f_1\colon(S^2,B)\to(S^2,A)$ are \emph{isotopic},
written $f_0\approx f_1$, if there exists a path
$(f_t\colon(S^2,B)\to(S^2,A))_{t\in[0,1]}$ of sphere maps connecting
$f_0$ to $f_1$. Clearly, all $B(f_t)$ are isomorphic.

A \emph{Thurston map} is a self-sphere map
$f\colon(S^2,A)\selfmap$. In that dynamical setting, we naturally
assume that the basepoints $*$ and $\dagger$ coincide. Two Thurston
maps $f\colon(S^2,A)\selfmap$ and $g\colon(S^2,B)\selfmap$ are
\emph{combinatorially equivalent} if there exists a homeomorphism
$\phi\colon(S^2,A)\to(S^2,B)$ with $\phi\circ f\approx g\circ\phi$.

We consider algebraic counterparts to these
notions. In~\cite{hurwitz:ramifiedsurfaces}, Hurwitz describes an
elegant classification of degree-$d$ branched coverings $S^2\selfmap$
with critical values contained in $\{a_1,\dots,a_n\}$ in terms of
\emph{admissible} $n$-tuples of permutations
$(\sigma_i\in d\perm)_{i=1,\dots,n}$. A $n$-tuple is \emph{admissible}
if $\sigma_1\cdots\sigma_n=1$ and
$\langle\sigma_1,\dots,\sigma_n\rangle$ is a transitive subgroup of
$d\perm$ and the cycle lengths of the $\sigma_i$ satisfy the condition
\begin{equation}\label{eq:riemannhurwitz}
  \sum_{i=1}^n\sum_{\substack{c\text{ cycle}\\\text{of }\sigma_i}}\bigl(\text{length}(c)-1\bigr)=2d-2.
\end{equation}

\begin{defn}[Sphere groups]\label{defn:sphere groups}
  A \emph{sphere group} is a tuple $(G,\Gamma_1,\dots,\Gamma_n)$
  consisting of a group and $n$ conjugacy classes $\Gamma_i$ in $G$,
  such that $G$ admits a presentation as in~\eqref{eq:spheregroup} for
  some choice of $\gamma_i\in\Gamma_i$. The $\Gamma_i$ are called
  \emph{peripheral conjugacy classes}.
\end{defn}
If $(S^2,A)$ is a marked sphere, we note that
$\pi_1(S^2\setminus A,*)$ is a sphere group for each
$*\in S^2\setminus A$. For every $d\in\N$ we denote by $\Gamma_i^d$
the subset $\{g^d \mid g\in\Gamma_i\}$.

Let $\subscript H B_G$ be a left-free biset of finite degree, and choose a
basis $S$ of $B$, namely a set of representatives for the left
action. Consider $g\in G$. Then $S\cong\{\cdot\}\otimes_H B$
decomposes into orbits $S_1\sqcup\cdots\sqcup S_\ell$ under the action
of $g$, of respective cardinalities $d_1,\dots,d_\ell$; and for all
$i=1,\dots,\ell$, choosing $s_i\in S_i$ there are elements $h_i\in H$
with $h_is_i=s_i g^{d_i}$. The multiset
$\{(d_i,h_i^H)\mid i=1,\dots,\ell\}$ consisting of degrees and
conjugacy classes in $H$ is independent of the choice of $S$, and
depends only on the conjugacy class of $g$; it is called the
\emph{lift} of $g^G$.

\begin{defn}[Sphere bisets]\label{defn:sphere bisets}
  Let $(G,\{\Gamma_i\})$ and $(H,\{\Delta_j\})$ be sphere groups. A
  \emph{sphere biset} is an $H$-$G$-biset $B$ such that the following
  hold:
  \begin{enumerate}
  \item $B$ is left-free and right-transitive;
  \item the permutations of $\{\cdot\}\otimes_H B$ induced by the
    right action of representatives of $\Gamma_1,\dots,\Gamma_n$ form
    an admissible tuple as in~\eqref{eq:riemannhurwitz};
  \item the multiset of all lifts of $\Gamma_1,\dots,\Gamma_n$
    contains exactly once every $\Delta_j$, the other conjugacy
    classes being all trivial.
  \end{enumerate}
  By the last condition, to every peripheral conjugacy class
  $\Delta_j$ in $H$ is associated a well-defined \emph{degree}
  $\deg_{\Delta_j}(B)\in\N$ and conjugacy class
  $\Gamma_i\eqqcolon B_*(\Delta_j)$, such that
  $(\deg_{\Delta_j}(B),\Delta_j)$ belongs to the lift of
  $\Gamma_i$. We define in this manner a map $B_*$ from the peripheral
  conjugacy classes in $H$ to those of $G$, called the \emph{portrait}
  of $B$.
\end{defn}
In case the peripheral classes of $G,H$ are indexed as
$(\Gamma_a)_{a\in A}$ and $(\Delta_c)_{c\in C}$ respectively, we write
$B_*(c)=a$ rather than $B_*(\Delta_c)=\Gamma_a$, defining in this
manner a map $B_*\colon C\to A$.

If $G=\pi_1(S^2\setminus A)$ and $H=\pi_1(S^2\setminus B)$ and
$f\colon(S^2,B)\to(S^2,A)$ is a sphere map, then $B(f)$ is a sphere
$H$-$G$-biset. The following result extends the Dehn-Nielsen-Baer
Theorem~\ref{thm:dehn-nielsen-baer} to non-invertible maps; its first
 part (in the dynamical setting $A=B$) is due to
 Kameyama~\cite{kameyama:thurston}. Recall that an isomorphism between sphere
bisets is required to preserve the peripheral conjugacy classes.
\begin{thm}
  The isomorphism class of $B(f)$ depends only on the isotopy class of
  $f$, and conversely every sphere $H$-$G$-biset is of the form $B(f)$
  for a sphere map $f\colon(S^2,B)\to(S^2,A)$.
\end{thm}
In summary, there is a bijective correspondence between isotopy
classes of sphere maps and isomorphism classes of sphere bisets.

\subsection{Multicurves}\label{ss:multicurves}
A \emph{multicurve} on $(S^2,A)$ is a disjoint collection $\CC$ of
non-trivial, non-peripheral simple closed curves on $S^2\setminus A$.
Algebraically, each curve in $\CC$ is expressed as a conjugacy class
in $\pi_1(S^2\setminus A)$; and one may choose in each $\Gamma\in\CC$
a representative $c_\Gamma\in\Gamma$ such that $\pi_1(S^2\setminus A)$
decomposes as a graph of groups, with one vertex per connected
component $S$ of $S^2\setminus\CC$, with group $\pi_1(S)$, and one
edge per curve $\Gamma\in\CC$, with group $\langle c_\Gamma\rangle$.
The underlying graph of the graph of groups is a tree. Following
Definition~\ref{defn:gog_1dimcovers}, we consider the barycentric
subdivision $\gf$ of this graph of groups, with one vertex per curve
in $\CC$ and one per connected component of $S^2\setminus\CC$.

Let $f\colon(S^2,A)\selfmap$ be a Thurston map. A multicurve $\CC$ is
\emph{$f$-invariant}\footnote{It is sometimes called ``completely
  invariant''} if every component of $f^{-1}(\CC)$ is either trivial,
peripheral, or homotopic to a curve in $\CC$, and every curve in $\CC$
appears in this manner; i.e.\ $f^{-1}(\CC)=\CC$ up to isotopy.

As soon as $f^{-1}(\CC)\subseteq\CC$ up to isotopy, one may construct
the \emph{transition matrix} of $f$ with respect to $\CC$, also called
\emph{Thurston matrix}. It is the endomorphism $T_f$ of $\Q\CC$
defined by
\begin{equation}\label{eq:thurston matrix}
  T_f(\gamma)=\sum_{\substack{\delta\in
      f^{-1}(\gamma)\\\delta\approx\varepsilon\in\CC}}\frac1{\deg(f\restrict\delta\colon\delta\to\gamma)}\varepsilon.
\end{equation}
Here by $\deg$ one means the usual positive degree of $f$; i.e.\ the
degree of $z^d\colon\{|z|=1\}\selfmap$ is $|d|$. A multicurve $\CC$ is
called an \emph{annular obstruction} if the spectral radius of its
Thurston matrix is $\ge1$; see Theorem~\ref{thm:thurston}.

Let $f\colon (S^2,A)\selfmap$ be a Thurston map and let $\CC$ be an
$f$-invariant multicurve. The van Kampen theorem lets us decompose
$B(f)$ as a sphere tree of bisets. Denote by $S'_1,\dots,S'_n$ the
connected components of $S^2\setminus\CC$, set
$S_j\coloneqq \overline {S'_j}\setminus A$ and call $S_j$ a
\emph{small sphere}. View each $S_j$ as a punctured sphere, so
$\pi_1(S_j)$ is a sphere group. Observe that $\{S_j\}\sqcup \CC$ is a
finite $1$-dimensional cover of $S^2\setminus A$ and denote by $\gf$
the associated \emph{sphere tree of groups}, see
Definition~\ref{defn:gog_1dimcovers}, with the sphere structure given
by the set of peripheral conjugacy classes in every $\pi_1(S_j)$. By
construction, each vertex of $\gf$ represents either a sphere $S_j$ or
a curve in $\CC$. The former is called a \emph{sphere vertex} and the
latter is called a \emph{curve vertex}.

Let $T'_1,\dots,T'_m$ be the connected components of
$S^2\setminus f^{-1}(\CC)$ and set
$T_j\coloneqq \overline{T'_j}\setminus f^{-1}(A)$. Then
$\{T_j\}\sqcup f^{-1}(\CC)$ is a $1$-dimensional cover of
$S^{2}\setminus f^{-1}(A)$. Using an isotopy rel $A$ modify the
inclusion $S^2\setminus f^{-1}(A)\hookrightarrow S^2\setminus A$ so
that the new map $i\colon S^2\setminus f^{-1}(A)\to S^2\setminus A$
squeezes all annuli between the essential curves in $f^{-1}(\CC)$ that
are isotopic rel $A$ and maps them to the corresponding curve in
$\CC$. If $i(T_j)\subset \gamma\in \CC$, then define
$\lambda(T_j)\coloneqq \gamma$; otherwise there is a unique $S_k$ such
that $\lambda(T_j)\subset S_k$, and define
$\lambda(T_j)\coloneqq S_k$. The map $\lambda$ is defined similarly
for curves in $f^{-1}(\CC)$. Since
$f\colon S^2\setminus f^{-1}(A) \to S^2\setminus A$ is a covering,
there is a unique
$\rho \colon \{T_j\}\sqcup f^{-1}(\CC)\to \{S_k\}\sqcup \CC $ such
that $f\colon T_j\to \rho(T_j)$ and $f\colon \gamma\to \rho(\gamma)$
are coverings. In this way we obtain a covering correspondence
$f,i \colon S^2\setminus f^{-1}(A) \rightrightarrows S^2\setminus A$
compatible with the $1$-dimensional covers. The \emph{sphere tree of
  bisets} $\subscript \gf\gfB_{\gf}$ is the associated graph of
bisets. A \emph{conjugacy} of a sphere tree of bisets
$\subscript \gf\gfB_{\gf}$ is required to respect the sphere
structure; namely the $\gfI$ in Definition~\ref{defn:conjugategob} is
a sphere tree of bisets. As with sphere trees of groups, vertices of
$\mathfrak B$ representing spheres $T_j$ are called \emph{sphere
  vertices} and vertices representing curves in $f^{-1}(\CC)$ are
\emph{curve vertices}.

\noindent We prove that the decomposition of $B(f)$ as a sphere tree
of bisets is computable:
\begin{algo}\label{algo:decompose}
  \textsc{Given} $f\colon(S^2,A)\selfmap$ a Thurston map by its sphere biset, and given $\CC$ an $f$-invariant multicurve as a collection of conjugacy classes,\\
  \textsc{Compute} the decomposition of $B(f)$ as a sphere tree of
  bisets.
\end{algo}

We return to small spheres $T_i\subset S^2\setminus A$ defining sphere
vertices of $\mathfrak B$.  Each $\overline T_i$ can be homotoped rel
$A$ either to a point (i.e.~the map $\overline T_i\hookrightarrow S^2$
is homotopic rel $A$ to a constant map), or to a curve in $\CC$, or to
a component $S_j$ (after filling in trivial rel $A$ discs), and is
respectively called \emph{trivial}, \emph{annular} or
\emph{essential}. Every $S_j$ contains up to homotopy a single
essential $T_i$, and $T_i$ covers via $f$ a single piece $S_k$, so
that we have a map $S_j\to S_k$ induced by $f$, well-defined up to
isotopy. We also write $k=f(j)$ so that $f\colon\{1,\dots,n\}\selfmap$
describes also how the components of $S^2\setminus\CC$ are mapped by
$f$. We define finally the set of \emph{return maps}
\[R(f,\CC)\coloneqq\{f^e\colon S_j\selfmap\;\mid f^e(j)=j\text{ and }f^{e'}(j)\neq j\text{ for all }e'<e\}.\]

All these notions have algebraic counterparts: consider a sphere biset
$B$. A multicurve $\CC$ is \emph{$B$-invariant} if, for every
conjugacy class $\Gamma\in\CC$ and every $b\in B$ there exist $d\in\N$
and $\Delta\in\CC\cup\{\Gamma_i\}\cup\{1\}$ such that
$\Delta^{\pm} b\subseteq b\Gamma^{\pm d}$, and if all curves in $\CC$
occur as such a $\Delta$. Note the similarity to
Definition~\ref{defn:sphere bisets}.  Consider the sphere tree of
bisets decomposition $\subscript {\gf}\gfB_\gf$ of $B$ along $\CC$,
and denote again by $B_*\colon\{1,\dots,n\}\selfmap$ the dynamics on
the essential sphere vertices of $\gfB$. A sphere vertex biset $B_z$
of $\gfB$ is called \emph{trivial} or \emph{annular} if it is of the
form $G_{\lambda(z)}\otimes_P B'$ for a $P$-$G_{\rho(z)}$-set $B'$ and
a subgroup $P\le G_{\lambda(z)}$ generated by a representative of a
peripheral or trivial conjugacy class, respectively a class in $\CC$;
and is called \emph{essential} otherwise. Let us denote by
$B_1,\dots,B_n$ the bisets associated with essential vertices in
$\gfB$, and let
\[R(\gfB)=R(B,\CC)\coloneqq\{B_j\otimes B_{B_*(j)}\otimes\cdots\otimes
B_{B_*^{e-1}(j)}\selfmap\;\mid B_*^e(j)=j\text{ with $e$ minimal}\}
\]
denote the return bisets of $B$, namely the bisets obtained by
following a cycle in the sphere tree of bisets into which $B$
decomposes.

\noindent We generalize a result by Kameyama~\cite{kameyama:thurston}
to marked spheres with multicurves:
\begin{thm}\label{thm:kameyama}
  Let $f\colon(S^2,A,\CC)\selfmap$ and $g\colon(S^2,B,\DD)\selfmap$ be
  two maps with respective invariant multicurves $\CC$ and $\DD$. Then
  $f$ and $g$ are combinatorially equivalent along an isotopy carrying
  $\CC$ to $\DD$ if and only if the sphere tree of bisets
  decompositions of $B(f)$ and $B(g)$ are conjugate (by a sphere tree
  of bisets, see Definition~\ref{defn:conjugategob}).

  If $A=B$ and $\CC=\DD$, then $f$ and $g$ are isotopic rel $A\cup\CC$
  if and only if their sphere trees of bisets decompositions are isomorphic.
\end{thm}

\subsection{Mapping class bisets}\label{ss:mcb}
Let $(S^2,A)$ be a marked sphere, and denote by $\Mod(S^2,A)$ the pure
mapping class group of $(S^2,A)$, namely the set of isotopy classes of
homeomorphisms of $S^2$ fixing $A$. For $\CC$ a multicurve on
$(S^2,A)$, denote by $\Mod(S^2,A,\CC)$ the subgroup of $\Mod(S^2,A)$
fixing each curve (together with its orientation) of $\CC$ up to
isotopy. Similarly, if $G$ is a sphere group $\pi_1(S^2\setminus A)$
then we write $\Mod(G)$ for $\Mod(S^2,A)$; by the Dehn-Nielsen-Baer
theorem (see~\cite{farb-margalit:mcg}*{Theorem~8.8}
and~\S\ref{ss:orbispheres}), the group $\Mod(G)$ is actually a group
of outer automorphisms of $G$. If $\gf$ is a sphere tree of groups
decomposition of $G$ along a multicurve $\CC$, then we write
$\Mod(\gf)$ for $\Mod(S^2,A,\CC)$; it is a group of self-conjugators
of $\gf$ in the sense of Definition~\ref{defn:conjugategob}.

Let $f\colon(S^2,B)\to(S^2,A)$ be a sphere map, and let $\DD,\CC$ be
multicurves on $S^2\setminus B,S^2\setminus A$ respectively with
$\DD\subseteq f^{-1}(\CC)$.
\begin{defn}[Mapping class bisets]
  The $\Mod(S^2,B)$-$\Mod(S^2,A)$-biset $M(f,B,A)$ is defined as
  \[M(f,B,A)=\{m'fm''\mid m'\in\Mod(S^2,B),m''\in\Mod(S^2,A)\}\,/\,{\approx}.\]
  It admits as a subbiset the  $\Mod(S^2,B,\DD)$-$\Mod(S^2,A,\CC)$-biset
  \[M(f,B,A,\DD,\CC)=\{m'fm''\mid
  m'\in\Mod(S^2,B,\DD),m''\in\Mod(S^2,A,\CC)\}\,/\,{\approx}.\]
  The left- and right-actions are given by
  $m'fm''=m''\circ f\circ m'$, in keeping with using the algebraic
  order of operations in bisets.
\end{defn}

By Theorem~\ref{thm:kameyama}, the biset $M(f,B,A)$ is also the set of
isomorphism classes of bisets of the form $B(m')\otimes B(f)\otimes
B(m'')$ with $m'\in\Mod(S^2,B),m''\in\Mod(S^2,A)$, and similarly for
$M(f,B,A,\DD,\CC)$. Therefore, for a sphere biset $\subscript H B_G$ we
introduce the notation
\begin{equation}
\label{eq:M(B)}
 M(B)=\{B_\psi\otimes B\otimes B_\phi\mid
\psi\in\Mod(H),\phi\in\Mod(G)\}\,/\,{\cong},
\end{equation}
and for a sphere tree of bisets $\subscript \gfY\gfB_\gf$ we
define
\begin{equation}\label{eq:conjgob}
  M(\gfB)=\{\mathfrak N' \otimes \gfB \otimes
  \mathfrak N''\mid \mathfrak N'\in\Mod(\gfY),\mathfrak
  N''\in\Mod(\gf)\}\,/\,{\cong}.
\end{equation}

We prove that $M(f,B,A)$ and $M(f,B,A,\DD,\CC)$ are left-free of
finite degree. The groups $\Mod(S^2,A)$ and $\Mod(S^2,A,\CC)$ are
computable: by the Dehn-Nielsen-Baer theorem
(see~\cite{farb-margalit:mcg}*{Theorem~8.8}
and~\S\ref{ss:orbispheres}), the group $\Mod(S^2,A)$ is the group of
outer automorphisms of $\pi_1(S^2\setminus A,*)$ that preserve
peripheral conjugacy classes, and $\Mod(S^2,A,\CC)$ is the subgroup that
also preserves the classes in $\CC$.
\begin{algo}\label{algo:mcb}
  \textsc{Given} $f$ a sphere map $(S^2,B,\DD)\to(S^2,A,\CC)$,\\
  \textsc{Compute} the biset $M(f,B,A,\DD,\CC)$.
\end{algo}

The mapping class group $\Mod(S^2,A,\CC)$ naturally fits into a split
exact sequence
\begin{equation}\label{eq:Mp_Gr_seq}
  1\longrightarrow\Mod[e](S^2,A,\CC)\overset \iota\longrightarrow
  \Mod(S^2,A,\CC)\overset\pi\longrightarrow\Mod[v](S^2,A,\CC)\longrightarrow1
\end{equation}
whose kernel $\Mod[e](S^2,A,\CC)$ is generated by Dehn twists about
curves in $\CC$, and thus is isomorphic to $\Z^{\#\CC}$, and whose
quotient is isomorphic to the direct product of the mapping class
groups of the path connected components of
$S^2\setminus(\CC\sqcup A)$, where all removed discs are shrunk to
punctures.

The structure of $M(f,B,A,\DD,\CC)$ qua
$\Mod[e](S^2,B,\DD)$-$\Mod[e](S^2,A,\CC)$-biset is described by the
Thurston matrix $T_f$ of $f$, see~\eqref{eq:thurston matrix}, in the
sense that if $f m=m' f$ with $m\in\Mod[e](S^2,A,\CC)$, then
$m'\in\Mod[e](S^2,B,\DD)$ is the image of $m$ under the Thurston
matrix.

We arrive finally at the main results of this part: decision problems
for mapping class bisets. Consider a finitely generated group $P$, and
a $P$-$P$-biset $B$ that is left-free of finite degree. We study the
following decision problems:
\begin{description}
\item[The conjugacy problem] Given $b,c\in B$, are they conjugate? If
  so, give a witness $g\in P$ with $b g=g c$.
\item[The centralizer problem] Given $b\in B$, compute its centralizer
  $Z(b)\coloneqq\{g\in P\mid g b=b g\}$.
\end{description}
By a \emph{computable group} we mean a finitely generated group with
solvable word problem. A subgroup $H$ of a computable group $P$ is
\emph{computable} if $H$ is finitely generated and has solvable
membership problem (i.e.\ there is an algorithm that decides, given
$g\in G$, whether $g\in H$). We say that a subgroup $L\le P$ is
\emph{sub-computable} if there is a computable subgroup $H\le P$ and a
computable homomorphism $H\to A$ to an Abelian group such that
$L=\ker(H\to A)$. (It follows from the definition that $L$ also has
solvable membership problem because it is decidable if $h\in H$ is in
$\ker(H\to A)$.) We say that a centralizer problem is \emph{solvable},
respectively \emph{sub-solvable}, if there is an algorithm calculating
the centralizer group as a computable, respectively sub-computable,
subgroup. When elements of a left-free $H$-$G$-biset
$B\cong H\times S$ are supplied to an algorithm, they are given in the
form $h s$ with $h\in H$ and $s\in S$; a \emph{computable biset} is
one such that the groups $G,H$ and the map $S\times G\to H\times S$
are computable.

Let $f\colon(S^2,A,\CC)\selfmap$ be a Thurston map, and set for
brevity $M\coloneqq M(f,A,A,\CC,\CC)$. From the exact
sequence~\eqref{eq:Mp_Gr_seq} we derive an ``exact sequence'' of
bisets
\begin{equation}\label{eq:extension mcb}
  \subscript{\Mod[e](S^2,A,\CC)}M_{\Mod[e](S^2,A,\CC)}\hookrightarrow M\twoheadrightarrow
  {}_{\Mod[e](S^2,A,\CC)}\backslash M/_{\Mod[e](S^2,A,\CC)}
\end{equation}
in which the first term is $M$ with restricted left and right actions,
and the third term is the quotient of $M$ by the left and right
actions of $\Mod[e](S^2,A,\CC)$. There is a finite-to-one map from
this third term to the product of mapping class bisets of spheres in
$(S^2,A)\setminus\CC$. We prove a general result about conjugacy and
centralizer problems in extensions, which gives the
\begin{thm}\label{thm:extension}
  Let $M$ be a mapping class biset as in~\eqref{eq:extension
    mcb}. There is an algorithm that computes the following. It
  receives as input two elements $b,c\in M$, a conjugator
  $g\in \Mod[v](S^2,A,\CC)$ with $\overline b g =g\overline c$ in
  ${}_{\Mod[e](S^2,A,\CC)}\backslash M/_{\Mod[e](S^2,A,\CC)}$, and the
  centralizer $Z(\overline b)\le \Mod[v](S^2,A,\CC)$ as a computable
  group. It computes whether $b\sim c$ in $M$, if so finds a
  conjugator, and produces $Z(b)$ as a sub-computable group.  If
  furthermore $Z(\overline b)$ is finite, then $Z(b)$ is computable.
\end{thm}
It is therefore sufficient, to solve conjugacy problems in $M$, to
solve them for return bisets in $R(B,\CC)$. We note that the
centralizer problem in $\Mod(S^2,A)$ is solvable while the centralizer
problem in $M$ is only sub-solvable, see the example
in~\S\ref{ex:infinitely generated}.

\subsection{Distinguished conjugacy classes}\label{ss:distinguished cc}
Let $(S^2,A)$ be a marked sphere, and let $\CC$ be a multicurve. It is
often convenient to treat similarly the conjugacy classes describing
elements of $A$ and of $\CC$. Consider a sphere tree of groups
$\gf$. The set of \emph{distinguished conjugacy classes} $X$ of $\gf$
is the set of all peripheral conjugacy classes of all sphere vertex groups
with two conjugacy classes identified if they are related by an
edge. Equivalently, if $\gf$ is the sphere tree of groups
decomposition of $(S^2,A,\CC)$, then $X$ is in bijection with
$A\cup \CC$. Note that the set of geometric edges of $\gf$ is
naturally a subset of the distinguished conjugacy classes.

The following algorithm determines when a bijection between (possibly
peripheral) multicurves is induced by a homeomorphism between the
underlying spheres:
\begin{algo}\label{algo:MarkCl:Prom}
  \textsc{Given} $\gf$ and $\gfY$ two sphere trees of
  groups with distinguished conjugacy classes $X$ and $Y$, and given a
  bijection $h\colon X\to Y$,\\
  \textsc{Decide} whether $h\colon X\to Y$ promotes to a conjugator
  $\gfI$ from $\gf$ to $\gfY$, and if so
  \textsc{construct} $\gfI$ \textsc{as follows:}\\\upshape
  \begin{enumerate}
  \item Check whether $h$ restricts to an isomorphism between the geometric
    edge sets of $\gf$ and $\gfY$. If not return \texttt{fail}.
  \item Check whether the isomorphism between the edge sets promotes
    into a graph-isomorphism $h\colon\gf\to\gfY$. If
    not, return \texttt{fail}.
  \item For a sphere vertex $v\in \gf$ let $\Gamma_v\subset X$ be the set
    of peripheral conjugacy classes of $G_v$. Check whether
    $h(\Gamma_v)=\Gamma_{h(v)}$ for all vertices $v\in \gf$. If
    not, return \texttt{fail}.
  \item For every sphere vertex $v\in\gf$ choose an isomorphism
    $\phi(v)\colon G_v\to G_{h(v)}$ compatible with $h\colon \Gamma_v\to
    \Gamma_{h(v)}$. For every edge $e\in \gfY$ choose an isomorphism
    $\phi(e)\colon G_e\to G_{h(e)}$.
  \item\label{algo:MarkCl:Prom:step5} Set $\gfI\coloneqq \gf$,
    $\lambda\coloneqq\one$, $\rho\coloneqq h$ and
    $B_z\coloneqq G_{h(z)}$ for all $z\in \gfI$; the left action of
    $G_z$ on $B_z$ is via $\phi(z)$, the right action is natural, and
    the inclusion of $B_e$ into $B_{e^-}$ is via $1\mapsto g$, for any
    $g\in G_{h(z)}$ with $()^g\circ\phi(e)=\phi(e^-)$, if we identify
    $G_e$ with a subgroup of $G_{e^-}$.
  \item Return $\gfI$.
  \end{enumerate}
\end{algo}

Let $\subscript \gf\gfB_\gf$ and $\subscript \gfY\gfC_\gfY$ be two
sphere trees of bisets. If $\gfI$ is a conjugator between $\gf$
and $\gfY$, then
\[(\gfC)^\gfI\coloneqq \gfI \otimes \gfC
\otimes \gfI^{\vee}
\]
is an $\gf$-tree of bisets. Recall from~\eqref{eq:conjgob} the
notations $\Mod(\gf)$ and $M(\gfB)$. The following two algorithms
determine, given two trees of bisets that stabilize the same
multicurve, whether they are twists of one another by mapping classes
respecting the multicurve.

The first algorithm expresses, if possible, a sphere tree of bisets as
a left multiple of another one. It relies on the following
observation. Suppose that we want to construct a biprincipal sphere
tree of bisets $\mathfrak T$ over a sphere tree of groups $\gf$, and
that its vertex and edge bisets are already given, so that only the
intertwiners $T_e\to T_{e^-}$ need be specified at edges of
$\mathfrak T$. Consider an edge pair $\{e,\overline e\}$. The bisets
$T_e$ and $T_{e^\pm}$ may be identified with the groups $G_e$ and
$G_{e^\pm}$ respectively; then the intertwiners $T_e\to T_{e^\pm}$ are
defined by $1\mapsto g_\pm$ for some $g_\pm\in G_{e^\pm}$ which
commutes with the image of $G_e$. Since the $G_{e^\pm}$ are free while
$G_e$ is Abelian, the element $g_\pm$ may be chosen arbitrarily in
$(G_e)^\pm$. All resulting choices of maps $T_e\to T_{e^\pm}$ are
called \emph{legal intertwiners}. In fact, writing $g_\pm=(h_\pm)^\pm$
for some $h_\pm\in G_e$, the isomorphism class of $\mathfrak T$
depends only on $h_+(h_-)^{-1}$. See Example~\ref{ss:ex:DehnTwist} for
the interpretation of these intertwiners as ``Dehn twists''.
\begin{algo}\label{algo:chech:InMCB:pre}
  \textsc{Given} ${}_\gf\mathfrak B_\gf$ and ${}_\gf\mathfrak C_\gf$
  two sphere $\gf$-trees of bisets,\\
  \textsc{Decide} whether there is an  $\mathfrak M\in \Mod(\gf)$ such that 
  $\mathfrak C\cong \mathfrak M \otimes \mathfrak B $, and if so
  \textsc{construct} $\mathfrak M$ and the isomorphism \textsc{as follows:}\\\upshape
  \begin{enumerate} 
  \item Try to construct an isomorphism of trees
    $h\colon \mathfrak B\to\mathfrak C$ mapping essential vertices
    into essential vertices such that
    $\lambda_\gfB (z)=\lambda_\gfC \circ h(z)$ and
    $\rho_\gfB(z)=\rho_\gfC \circ h(z)$ for every $z\in\gfB$. If $h$
    does not exist, then return \texttt{fail}. Otherwise $h$ is
    unique.
  \item Choose an essential sphere vertex $v\in \gfB$. Try to find
    $M_{\lambda(v)} \in \Mod(G_{\lambda(v)})$ such that
    $M_{\lambda(v)} \otimes B_v \cong C_{h(v)}$. If such
    $M_{\lambda(v)}$ do not exist, return \texttt{fail}. Otherwise set
    $S\coloneqq \{v\}$ and run
    Steps~\ref{algo:chech:InMCB:St3}--\ref{algo:chech:InMCB:St6} over
    all pairs $\{e,\overline e\}\not\subset S$ but with $e^- \in S$.
  \item If $\lambda(e)\notin \lambda(S)$, then do the
    following. (Note that in this case $\lambda(e)$ is an edge in
    $\gf$.) Add $e$ and $\overline e$ to $S$, let $M_{\lambda(e)}$ be
    a principal $\Z$-biset, choose any legal intertwiner
    $()^-\colon M_{\lambda(e)}\to M_{\lambda(e)^-}$, and define
    $M_{\overline e}$ similarly.\label{algo:chech:InMCB:St3}
  \item Try to find an isomorphism between
    $M_{\lambda(e)}\otimes B_{e}$ and $C_{h(e)}$ compatible with the
    intertwiner maps. If it does not exist, return \texttt{fail}.
  \item If $e^+$ is not an essential sphere vertex, then do the
    following. If $\lambda(e^+)\notin S$, then (in this case
    $\lambda(e^+)$ is a curve vertex) choose a biprincipal $\Z$-biset
    $M_{\lambda(e^+)}$, choose a legal intertwiner from
    $M_{\lambda(e)}$ to $M_{{\lambda(e^+)}}$, and add $e^+$ to
    $S$. Try to find an isomorphism between
    $M_{\lambda(e^+)}\otimes B_{e^+}$ and $C_{h(e^+)}$ that is
    compatible with the isomorphism between
    $M_{\lambda(e)}\otimes B_{e}$ and $C_{h(e)}$ via the
    intertwiner maps. If it does not exist, return
    \texttt{fail}.
  \item If $e^+$ is an essential sphere vertex, then do the
    following. Try to find $M_{\lambda(z)} \in \Mod(G_{\lambda(z)})$
    such that $M_{\lambda(z)} \otimes B_z \cong C_{h(z)}$. If such
    $M_{\lambda(z)}$ do not exist, return \texttt{fail}. Try to find a
    legal intertwiner $()^-\colon M_{\lambda(e)}\to M_{\lambda(e^+)}$
    such that the isomorphism between $M_{\lambda(e)}\otimes B_{e}$
    and $C_{h(e)}$ is compatible with the isomorphism between
    $M_{\lambda(e)}\otimes B_{e}$ and $C_{h(e)}$ via the
    intertwiner maps. If no such
    intertwiner exists, return
    \texttt{fail}. Add $e^+$ to $S$. \label{algo:chech:InMCB:St6}
  \item Return the principal sphere tree of bisets $\mathfrak M$ and
    the isomorphism between $\gfC$ and $\mathfrak M \otimes \gfB$
    constructed via $h$.
  \end{enumerate}
\end{algo}

\begin{algo}\label{algo:chech:InMCB}
  \textsc{Given} $\subscript \gf\gfB_\gf$ and $\subscript \gf\gfC_\gf$
  two sphere $\gf$-trees of bisets,\\
  \textsc{Decide} whether $\gfC\in M(\gfB)$, and if so
  \textsc{construct} $\mathfrak M, \mathfrak N\in \Mod(\gf)$
  such that $\gfC\cong \mathfrak M \otimes \gfB \otimes
  \mathfrak N$ \textsc{as follows:}\\\upshape
  \begin{enumerate}
  \item Follow Algorithm~\ref{algo:mcb} to compute a basis of the
    mapping class biset $M(\gfB)$.
  \item For each $\mathfrak N$ in the basis, do the following. Run
    Algorithm~\ref{algo:chech:InMCB:pre} on $\gfC$ and
    $\gfB\otimes\mathfrak N$. If there exists
    $\mathfrak M\in\Mod(\gf)$ with
    $\gfC\cong\mathfrak M\otimes(\mathfrak B\otimes\mathfrak N)$,
    return $(\mathfrak M,\mathfrak N)$.
  \item Return \texttt{fail}.
  \end{enumerate}
\end{algo}

\subsection{Orbispheres}\label{ss:orbispheres}
In fact, a slightly more general situation than marked spheres, that
of \emph{orbispheres}, can be considered at almost no cost. Let there
also be given a function $\ord\colon A\to\{2,3,\dots,\infty\}$,
assigning a positive or infinite order to each marked point. This
describes an \emph{orbispace} structure: if $\ord(a)=\infty$, the
point $a\in A$ is punctured, while if $\ord(a)=n$ then the space has a
cone-type singularity of angle $2\pi/n$ at $a$. It is convenient to
extend $\ord$ to $S^2$ so that $\ord(p)=1\Leftrightarrow p\notin A$.
We call $(S^2,A,\ord)$ an \emph{orbisphere}, and write its fundamental
group, called an \emph{orbisphere group}, as
\begin{equation}\label{eq:orbispheregp}
  G=\pi_1(S^2,A,\ord,*)=\langle \gamma_1,\dots,\gamma_n\mid \gamma_1^{\ord(a_1)},\dots,\gamma_n^{\ord(a_n)}, \gamma_1\cdots\gamma_n\rangle.
\end{equation}

A \emph{map} $f\colon(S^2,B,\ord_B)\to(S^2,A,\ord_A)$ between
orbispheres is a branched covering between the underlying spheres,
with $f(B)\cup\{\text{critical values}(f)\}\subseteq A$. It is locally
modeled at $p\in S^2$ by $z\mapsto z^{\deg_p(f)}$, in charts
respectively centered at $p$ and $f(p)$, and the orbispace structures
satisfy $\ord_B(p)\deg_p(f)\mid\ord_A(f(p))$ for all $p\in S^2$.

Sphere groups and maps are subsumed in these definitions, by setting
$\ord(a)=\infty$ for all $a\in A$. On the other hand, let
$f\colon(S^2,A)\selfmap$ be a Thurston map. There is then a
\emph{minimal} orbisphere structure $(S^2,P_f,\ord_f)$, with
$P_f\subseteq A$ the post-critical set of $f$, given by
\[\ord_f(p)=\lcm\{\deg_q(f^n)\mid n\ge0,q\in f^{-n}(p)\}.\]

These notions have algebraic counterparts. Let $G$ be an orbisphere
group with peripheral conjugacy classes $\{\Gamma_a\}_{a\in A}$. For $a\in
A$ set
\[\ord_G(a)=\min\{d\in \N\mid\Gamma^d_a=\{1\}\}.\]
Orbisphere bisets are defined exactly as in
Definition~\ref{defn:sphere bisets}. Let $\subscript GB_G$ be an
orbisphere biset. It has a \emph{portrait} $B_*\colon A\selfmap$
induced by the map on peripheral conjugacy classes, and a \emph{local
  degree} $\deg_a(B)$. There is a \emph{minimal orbisphere quotient}
of $G$ associated with $B$, which is the quotient $\overline G$ of $G$
by the additional relations $\Gamma_a^{\ord_B(a)}=1$, for
\begin{equation}\label{eq:ordB}
  \ord_B(a)=\lcm\{d\mid n\ge0\text{ and $\Gamma_a$ has a lift of degree $d$ under }B^{\otimes n}\},
\end{equation}
and clearly $\ord_{B}(a) \mid \ord_{G}(a)$. We call the quotient biset
$\overline G\otimes_G B\otimes_G\overline G$ the \emph{minimal
  orbisphere biset} of $\subscript GB_G$.

The mapping class group $\Mod(G)$ of an orbisphere group is naturally
defined as the group of outer automorphisms of $G$ that preserve its
peripheral conjugacy classes classwise. It turns out that $\Mod(G)$ is
just a usual mapping class group. More precisely, consider a marked
sphere $(S^2,A)$ with $\widetilde G=\pi_1(S^2,A,*)$ and an orbisphere
structure $(S^2,A,\ord)$ with non-positive Euler characteristic and
with corresponding orbisphere group $G$. There is a natural map
$\widetilde G\to G$ mapping each generator $\gamma_i\in\widetilde G$
to $\gamma_i\in G$.

\begin{thm}[Dehn-Nielsen-Baer-Zieschang-Vogt-Coldewey~\cite{zieschang-vogt-coldewey:spdg}*{Theorem~5.8.3}]\label{thm:dehn-nielsen-baer}
  If $G$ has at least one peripheral class of order $\ge 3$, then the
  natural map $\widetilde G\to G$ induces an isomorphism
  $\Mod(\widetilde G)\to\Mod(G$).\qed
\end{thm}
There is a small difficulty if all peripheral classes in $G$ have
order $2$, because then the orientation of a mapping class is
difficult to read in $\Mod(G)$. In that case, we replace $\Mod(G)$ by
its orientation-preserving index-$2$ subgroup. In the
$(2,2,2,2)$-case, every mapping class is of the form $M^{(0,0)}$,
see~\eqref{eq:InjEndOfK}, with $\det(M)=1$ and $M$ fixing peripheral
conjugacy
classes. (See~\cite{bartholdi-dudko:bc2}*{\S\ref{bc2:ss:orbispheres}}
for details.)

An equivalent and useful formulation of orbispheres is via
\emph{planar discontinuous group actions}. To every orbisphere
$(S^2,A,\ord)$ there corresponds a universal cover which, if
$\sum_{a\in A}(1-\ord(a)^{-1})\ge2$, is a topological plane $\Pi$, and
an action of $G=\pi_1(S^2,A,\ord,*)$ on $\Pi$ by homeomorphisms such
that $(S^2,A,\ord)=\Pi/G$; for every $p\in\Pi$ that projects to a
marked point $a_i\in A$, its stabilizer $G_p$ is a cyclic group of
order $\ord(a_i)$, conjugate to $\langle\gamma_i\rangle$. Maps between
orbispace can be lifted to equivariant maps between their universal covers.
\begin{thm}[Baer~\cite{zieschang-vogt-coldewey:spdg}*{Theorem~5.14.1}]\label{thm:Baer}
  An orientation preserving homeomorphism of a plane commuting with
  planar discontinuous group is isotopic to the identity relative to
  the group action.\qed
\end{thm}

We consider now in more detail the operation on sphere groups and
sphere bisets consisting of changing the orbispace structure while
retaining the marked points. The less innocuous operation of erasing
punctures (i.e.\ setting their order to $1$) will be treated
in~\S\ref{ss:forgetful}.

Let $\widetilde G$, $G$ be two orbisphere groups with respective
peripheral conjugacy classes $(\widetilde\Gamma_a)_{a\in A}$ and
$(\Gamma_a)_{a\in A}$. An \emph{inessential forgetful morphism}
$\widetilde G\to G$ is a homomorphism
$\iota \colon \widetilde G\to G$ such that for every $a\in A$ we have
$\iota(\widetilde\Gamma_a)=\Gamma_a$.  Therefore the order of $a$ in
$G$ divides the order of $a$ in $\widetilde G$ for all
$a\in A$.

Let $G$ be an orbisphere group and let $B$ be a $G$-$G$-biset. We
naturally get a right action of $G$ on
\[T(B)\coloneqq\bigsqcup_{n\ge0}\{\cdot\}\otimes_G B^{\otimes n}.
\]
If $B$ is left-free of degree $d$ then $T(B)$ naturally has the
structure of a $d$-regular rooted tree: if $S$ is a basis of $B$, then
$T(B)$ is in bijection with the set of words $S^*$, which forms a
$\#S$-regular tree if one puts an edge between $s_1\dots s_n$ and
$s_1\dots s_{n+1}$ for all $s_i\in S$. The action of $G$ on $T(B)$
need not be free; following~\cite{nekrashevych:ssg}*{5.1.1} we denote
by $\IMG_B(G)$ the quotient of $G$ by the kernel of this action.

\begin{lem}
  Let $\subscript{\widetilde G}{\widetilde B}_{\widetilde G}$ be an
  orbisphere biset and let $\iota\colon \widetilde G\to G$ be an
  inessential forgetful morphism between orbisphere groups. Then
  \begin{equation}\label{eq:tildeB to B}
    B\coloneqq G\otimes_{\widetilde G}\widetilde B_{\widetilde G}\otimes G
  \end{equation}
  is an orbisphere $G$-$G$-biset if and only if the kernel of
  $\iota\colon\widetilde G\to G$ is contained in the kernel of
  $\widetilde G\to\IMG_{\widetilde B}(\widetilde G)$ if and only if
  the kernel of $\iota\colon\widetilde G\to G$ is contained in the
  kernel of $\widetilde G\to \overline G$, where $\overline G$ is the
  \emph{minimal orbisphere quotient} of $G$ associated with $B$.
\end{lem}

For an orbisphere biset $\subscript GB_G$ its mapping class biset
$M(B)$ is defined in the same way as in the sphere
case~\eqref{eq:M(B)}; namely $M(B)$ is the set of isomorphism classes
of twists $B_\psi\otimes B\otimes B_\phi$ under all
$\psi\in\Mod(H),\phi\in\Mod(G)$.

\begin{thm}
  Suppose that $\widetilde G\to G$ is an inessential forgetful
  morphism such that
  $B\coloneqq G\otimes_{\widetilde G}\widetilde B_{\widetilde
    G}\otimes G$ as in~\eqref{eq:tildeB to B} is an orbisphere biset.
  Then the natural map
  $\widetilde b\mapsto 1\otimes\widetilde b\otimes 1$ induces an
  isomorphism between the mapping class bisets $M(\widetilde B)$ and
  $M(B)$.
\end{thm}

\section{Forgetful morphisms and geometric maps~\cite{bartholdi-dudko:bc3}}\label{ss:forgetful}
We consider in this part the operation of erasing punctures from a
sphere. Consider $f\colon(S^2, A\sqcup E,\widetilde \ord)\selfmap$,
and assume that $f$ induces a branched covering $f\colon(S^2,A,\ord)\selfmap$
such that $\ord \colon A\to\{2,3,\dots,\infty\}$ satisfies
$\ord(a)\mid \widetilde \ord(a)$ for all $a\in A$. There is a natural
\emph{forgetful epimorphism}
\begin{equation}
\label{eq:defn:ForgGrMorph}\widetilde G\coloneqq\pi_1(S^2, A\sqcup E,  \widetilde  \ord,*)\to\pi_1(S^2,A,\ord, *)\eqqcolon G,
\end{equation}
as well as a natural \emph{forgetful biset epimorphism} 
\begin{equation}
\label{eq:defn:ForgBisMorph}
\subscript{\widetilde G}{\widetilde B}_{\widetilde G} \coloneqq B(f\colon(S^2, A\sqcup E,\widetilde  \ord)\selfmap )\to B(f\colon(S^2, A,  \ord)\selfmap)\eqqcolon \subscript {G}B_{G}
\end{equation}
given by 
\begin{equation}\label{eq:defn:ForgBisMorph2}
  \left\{\begin{array}{cl}
  \subscript{\widetilde G}{\widetilde B}_{\widetilde G} &\to G\otimes_G\widetilde B\otimes_{ G} G=\subscript G B_G,\\
  b &\mapsto 1\otimes b\otimes 1.
\end{array}\right.
\end{equation}  
We say that~\eqref{eq:defn:ForgGrMorph}
and~\eqref{eq:defn:ForgBisMorph} are \emph{essential} if $A\subsetneq
\widetilde A$; otherwise~\eqref{eq:defn:ForgGrMorph}
and~\eqref{eq:defn:ForgBisMorph} are \emph{inessential.}  A
forgetful morphism~\eqref{eq:defn:ForgBisMorph} is \emph{maximal} if
$(S^2,A,\ord)=(S^2,P_f,\ord_f)$.

We show that the geometric (see the next~\S) biset
$\subscript{\widetilde G}{\widetilde B}_{\widetilde G}$ of degree $>1$
can be described, and recovered, in terms of $\subscript GB_G$ and
extra data which we call a \emph{portrait of bisets}. This allows us,
in particular, to understand algorithmically maps doubly covered by
torus endomorphisms. In that case we note that
$(S^2, A,\ord)\coloneqq (S^2,P_f,\ord_f)$ is a $(2,2,2,2)$-orbisphere
and we show that $\subscript GB_G$ is a crossed product of an Abelian
biset with an order-$2$ group, see~\S\ref{ss:lattes}.

\subsection{Geometric maps}
Let us specify, by geometric conditions, a class of maps that will be
central to our study; see Figure~\ref{fig:geometric}. Ultimately, we will
show that it is equivalent to a topological condition, being ``Levy-free'',
see Definition~\ref{defn:levyfree}.

\begin{defn}
  A homeomorphism $f\colon (S^2,A)\selfmap$ is \emph{geometric} if $f$ is
  either
  \begin{itemize}
  \item[\FO] of finite order: there is an $n>0$ such that $f^n=\one$; or
  \item[\PA] pseudo-Anosov~\cite{thurston:surfaces}: there are two
    transverse measured foliations preserved by $f$ such that one
    foliation is expanded by $f$ while another is contracted.
  \end{itemize}

  Consider now a non-invertible map $f\colon (S^2,A)\selfmap$. Let
  $A^\infty\subseteq A$ denote the forward orbit of the periodic critical
  points of $f$. The map $f$ is \emph{B\"ottcher expanding} if there exists
  a metric on $S^2\setminus A^\infty$ that is expanded by $f$, and such
  that $f$ is locally conjugate to $z\mapsto z^{\deg_a(f)}$ at all $a\in
  A^\infty$. The map $f$ is \emph{geometric} if $f$ is either
  \begin{itemize}
  \item[\Exp] B\"ottcher expanding; or
  \item[\Tor] a quotient of a torus endomorphism
    $M z+v\colon\R^2\selfmap$ by the involution $z\mapsto-z$, such that
    the eigenvalues of $M$ are different from $\pm 1$.
  \end{itemize}
\end{defn}

The two cases are not mutually exclusive. A map $f\in\Tor$ is
expanding if and only if the absolute values of the eigenvalues of $M$
are greater than $1$.

A distinguished property of a geometric map is \emph{rigidity}: two
geometric maps are combinatorially equivalent if and only if they are topologically
conjugate.

An orbisphere biset $\subscript G B_{G}$ is \emph{geometric} if it is the biset
of a geometric map, and \Tor\ and \Exp\ bisets are defined
similarly. By rigidity there is a map $f_B\colon (S^2,A,\ord)\selfmap$, unique
up to conjugacy, such that the biset of $f_B$ is $B$.

If $\subscript GB_G$ is geometric and
$\subscript{\widetilde G}{\widetilde B}_{\widetilde G}\to \subscript
G B_G$ is a forgetful morphism as in~\eqref{eq:defn:ForgBisMorph}, then
elements of $E$ (which \emph{a priori} are defined up to homotopy) can
be interpreted dynamically as extra marked points on $S^2\setminus
A$. This is an instance of \emph{homotopy shadowing},
see~\cite{ishii-smillie:shadowing}.

More precisely, suppose that $\subscript{\widetilde G}{\widetilde
B}_{\widetilde G}\to \subscript G B_G$ is a forgetful morphism as
in~\eqref{eq:defn:ForgBisMorph}; suppose that $\subscript GB_G$ is the biset
of a geometric map $f_B\colon (S^2,A,\ord)\selfmap$; and suppose that
$\subscript{\widetilde G}{\widetilde B}_{\widetilde G}$ is a Levy-free biset,
see Definition~\ref{defn:levyfree}.  Then a finite set $E\subset
S^2\setminus A$ with $f_B(E)\subset A\sqcup E$ can be added to $A$ in
such a way that $\subscript{\widetilde G}{\widetilde B}_{\widetilde G}$ is
conjugate to the biset of $f_B\colon (S^2,A\sqcup E,\widetilde
\ord)\selfmap$. Moreover, if $G$ is not a cyclic group, then the set
$E$ is unique.
 
\subsection{Portraits of groups and bisets}\label{ss:portraits}
Let $\subscript GB_G$ be an orbisphere biset with peripheral conjugacy classes
$(\Gamma_a)_{a\in A}$ and portrait $B_*\colon A\selfmap$.
Suppose that $E$ is a finite set and suppose that
$B_*\colon A\sqcup E\selfmap$ is an extension of
$B_*\colon A\selfmap$.

\begin{defn}[Portraits of groups and bisets]\label{dfn:PrtrOfBst}
  Let $G$ be an orbisphere group with peripheral conjugacy classes
  $(\Gamma_a)_{a\in A}$ and let $E$ be a finite set. A \emph{portrait
    of groups} $(G_a)_{a\in A\sqcup E}$ in $G$ is a collection of
  cyclic subgroups $G_a\le G$ such that
  \[G_a=\begin{cases} \langle g \rangle &\mbox{for some } g\in \Gamma_a \mbox{ if } a\in A, \\
    \langle1 \rangle & \mbox{if } a \in E.  \end{cases}
  \]
  If $E=\emptyset$, then $(G_a)_{a\in A}$ is a \emph{minimal} portrait
  of groups.

  Let $\subscript GB_G$ be an orbisphere biset and let
  $B_*\colon A\sqcup E\selfmap$ be an extension of
  $B_*\colon A\selfmap$. A \emph{portrait of bisets in $\subscript GB_G$
    parameterized by $B_*\colon A\sqcup E\selfmap$} is a collection
  $(G_a,B_a)_{a\in A \sqcup E}$ such that
  \begin{itemize}
  \item[(A)] $(G_a)_{a\in A \sqcup E}$ is a portrait of groups in $G$;
    and
  \item[(B)] $B_a$ is a transitive $G_a$-$G_{B_*(a)}$-subbiset of
    $\subscript GB_G$ such that if $B_*(a)=B_*(c)$ and
    $G\otimes_{G_a} B_a= G\otimes_{G_c} B_c$ as subsets of $B$, then
    $a=c$.
  \end{itemize}
  If $E=\emptyset$, then $(G_a,B_a)_{a\in A \sqcup E}$ is a
  \emph{minimal} portrait of bisets.

  Two portraits of bisets $(G_a,B_a)_{a\in A\sqcup E}$ and
  $(G'_a,B'_a)_{a\in A\sqcup E}$ parameterized by
  $B_*\colon A\sqcup E\selfmap$ are \emph{conjugate} if there exist
  $(\ell_a)_{a\in A\sqcup E}$ such that
  \begin{equation}
    \label{eq:ConjOfPrtsOfBisets}G_a=\ell_a^{-1}G'_a\ell_a \text{ and } B_a=\ell_a^{-1}B'_a\ell_{B_*(a)}.\qedhere
  \end{equation}
\end{defn}
\noindent Every biset admits a minimal portrait, unique up to conjugacy.

Let us give a geometric interpretation of portraits of
bisets. Consider a branched covering $f\colon (S^2,A)\selfmap$. For
every $a\in A$ choose a small neighbourhood $D_a$ of it; up to isotopy
we may assume that
$f\colon D_a \setminus \{a\}\to D_{f(a)} \setminus \{f(a)\}$ is a
covering. Making appropriate choices we may embed
$\pi_1(D_a\setminus \{a\})$ and
$B(f\colon D_a\setminus \{a\} \to D_{f(a)}\setminus \{f(a)\})$ into
$\pi_1(S^2,A)$ and $B(f)$ respectively; calling the images of these
embeddings $G_a$ and $B_a$ we get a minimal portrait of bisets
$(G_a,B_a)_{a\in A}$ in $B(f)$.

Recall from~\S\ref{ss:bisets} that a self-conjugacy of a $G$-$G$-biset
$B$ is a pair of maps $(\phi\colon G\selfmap,\beta\colon B\selfmap)$
with $\beta(h b g)=\phi(h)\beta(b)\phi(g)$; and an automorphism of $B$
is a self-conjugacy with $\phi=\one$. We show that, for sphere bisets,
every self-conjugacy $(\phi,\beta)$ is determined by its map $\phi$,
or equivalently that $B$ admits no non-trivial automorphism:
\begin{thm}[No automorphisms]\label{thm:NoGhostAut}
  Suppose $\subscript GB_G$ is a non-cyclic orbisphere biset. Then the
  automorphism group of $\subscript GB_G$ is trivial.
\end{thm}

\begin{defn}\label{defn:forgetting}
  Let $\iota\colon\subscript{\widetilde G}{\widetilde B}_{ \widetilde G}\to
  \subscript G B_G$ be a forgetful morphism of orbisphere bisets as
  in~\eqref{eq:defn:ForgBisMorph}.  Let $(G_a,B_a)_{a\in A\sqcup E}$
  be a minimal portrait of bisets in $\widetilde B$. Then
  \[(G_a,B_a)_{a\in A\sqcup E} \coloneqq
  (\iota(G_{a}),\iota(B_{a}))_{a\in A\sqcup E} \]
  is the \emph{induced portrait of bisets}. It is parameterized by
  $B_*\coloneqq  \widetilde B_* \colon
  A\sqcup E \selfmap$.
\end{defn}

Consider a forgetful morphism $\widetilde G\to G$ as
in~\eqref{eq:defn:ForgGrMorph}.  For every $e\in E$ set \[\widetilde
G_e\coloneqq \pi_1(S^2, A\sqcup\{e\}, \widetilde
\ord\mid_{A\sqcup\{e\}},*)
\]
so that the forgetful morphism $\widetilde G\to G$ factors as
$\widetilde G\to\widetilde G_e\to G$.  We say that a biset
$\subscript{\widetilde G}K_{\widetilde G}\in \Mod(\widetilde G)$ is
\emph{knitting} if for every $e\in E$ we have
\[ \widetilde G_e \otimes_{\widetilde G} K \otimes_{\widetilde G}
  \widetilde G_e \cong \subscript{\widetilde G_e}{(\widetilde
    G_e)}_{\widetilde G_e}.
\]
If $m\in \Mod(S^2,A\sqcup E)$, then the biset of $m$ is knitting if
and only if $m$ is trivial in $\Mod(S^2,A\sqcup \{e\})$ for all
$e\in E$.

\begin{thm}
  Let $\widetilde G\to G$ be a forgetful morphism as
  in~\eqref{eq:defn:ForgGrMorph}.  Suppose that $\subscript GB_G$ is an
  orbisphere biset and that $(G_a,B_a)_{a\in A\sqcup E}$ is a portrait
  of bisets parameterized by $B_*\colon A\sqcup E\selfmap$.

  Then there is an orbisphere biset
  $\subscript{\widetilde G}{\widetilde B}_{\widetilde G}$ such that
  $(G_a,B_a)_{a\in A\sqcup E}$ is induced by the forgetful morphism 
  $\subscript{\widetilde G}{\widetilde B}_{\widetilde G}\to
  \subscript GB_G$ defined by~\eqref{eq:defn:ForgBisMorph2},
  and if $\subscript{\widetilde G}{\widetilde B'}_{\widetilde G}$ is another
  such biset then there is a knitting biset
  $\subscript{\widetilde G} K_{\widetilde G}$ with
  \[\widetilde B' \cong K \otimes_{\widetilde G} \widetilde B.\]
\end{thm}

\noindent Conjugacy and centralizers of bisets of the form
$\widetilde B$ may be studied as follows:
\begin{thm}\label{thm:RedConjCentrProb}
  Let $\widetilde G\to G$ be a forgetful morphism of groups as
  in~\eqref{eq:defn:ForgGrMorph} and let
  \[\subscript{\widetilde G}{\widetilde
      B}_{\widetilde G}\to \subscript G B_G\;\text{ and }\; \subscript{\widetilde
      G}{\widetilde C}_{\widetilde G}\to \subscript G C_G
  \]
  be two forgetful biset morphisms as in~\eqref{eq:defn:ForgBisMorph}.
  Suppose furthermore that~$\widetilde B$ is geometric of degree
  $>1$. Denote by $(G_a,B_a)_{a\in A\sqcup E}$ and
  $(G_a,C_a)_{a\in A\sqcup E}$ the portraits of bisets induced by
  $\widetilde B$ and $\widetilde C$ in $B$ and $C$ respectively.

  Then $\widetilde B,\widetilde C$ are conjugate by
  $\Mod(\widetilde G)$ if and only if there exists
  $\phi\in\Mod(G)$ such that $B^\phi\cong C$ and the portraits
  $(G_a^\phi,B_a^\phi)_{a\in A\sqcup E}$ and $(G_a,C_a)_{a\in A\sqcup
    E}$ are conjugate qua portraits in $C$.

  Furthermore, the centralizer $Z(\widetilde B)$ of $\widetilde B$ is
  isomorphic, via the forgetful map $\Mod(\widetilde
  G)\to\Mod(G)$, to
  \[\big\{\phi\in Z(B)
  \big|\,(G_a^{\phi},B_a^{\phi})_{a\in A\sqcup E}\sim(G_a,B_a)_{a\in
    A\sqcup E}\big\}
  \]
  and is a finite-index subgroup of $Z(B)$.
\end{thm}
It follows that the conjugacy and centralizer problems for
$\widetilde B$ are decidable as soon as they are decidable for
$(B,\text{portrait})$.

\begin{prop}[Contracting case]\label{prop:PortrBis:ContrCase}
  Suppose that $B$ is an orbisphere contracting $G$-biset, see
  Definition~\ref{defn:contracting}. Then for every
  $B_*\colon A\sqcup E\selfmap$ the number of conjugacy classes of
  portraits of bisets parameterized by $B_*$ is finite.

  Moreover, there is an algorithm that, given
  $B_*\colon A\sqcup E\selfmap$, decides whether two portraits of
  bisets parameterized by $B_*$ are conjugate.
\end{prop}
The algorithm of Proposition~\ref{prop:PortrBis:ContrCase} reduces the
conjugacy problem for portraits to conjugacy problems of elements in
$B^{\otimes n}$. Here is a simple example illustrating this main step
in the algorithm:
\begin{exple}
  Suppose $E=\{e\}$ with $B_*(e)=e$, and let $(G_a,B_a)_{a\in A\sqcup
    E}$ and $(G_a,C_a)_{a\in A\sqcup E}$ be two portraits of
  bisets. Then $G_e=1$ and $B_e=\{b\}$ and $C_e=\{c\}$.

  The portraits $(G_a,B_a)_{a\in A\sqcup E}$ and $(G_a,C_a)_{a\in
    A\sqcup E}$ are conjugate if and only if there exists $\ell \in G$
  such that $\ell b \ell^{-1}=c$.

  This is a conjugacy problem in the biset $B$ which can effectively
  be solved using contraction in $B$, since the length of $\ell$ can
  be bounded in terms of the lengths of $b$ and $c$.
\end{exple}

Let $B$ be the biset of a map $f$. In case $E=\{e_1,\dots,e_n\}$
consists of a single cycle for $B_*$, all the bisets $B_e$ are
singletons and the portrait of bisets contains precisely the same
information as a ``homotopy pseudo orbit'', namely a sequence of
points $z_1,\dots,z_n$ with homotopy classes of paths $\gamma_n$ from
$z_n$ to an $f$-preimage of $z_{n+1}$, indices read modulo $n$;
see~\cite{ishii-smillie:shadowing}. In case $f$ is expanding, these
authors prove that $(z_1,\dots,z_n)$ is homotopic to a unique
period-$n$ cycle.

\begin{algo}\label{algo:erase}
  \textsc{Given} bisets $\widetilde B$, $\widetilde C$ together with
  the forgetful morphisms onto $B$ and $C$ as in
  Theorem~\ref{thm:RedConjCentrProb} such that, in addition, $B$ is
  contracting; and \textsc{given} $\phi\in\Mod(G)$ such that $B^\phi\cong C$ and a finite generating set of $Z(B)$,\\
  \textsc{Decide} whether $\widetilde B$ and $\widetilde C$ are
  conjugate, and \textsc{compute} $Z(\widetilde B)$.
\end{algo}

\subsection{\myboldmath Bisets of minimal $(2,2,2,2)$-maps}\label{ss:lattes}
A $(2,2,2,2)$-orbisphere is $(S^2,A, \ord)$ with $\#A=4$ and
$\ord(a)=2$ for all $a\in A$.

\begin{lem}\label{lem:Gr:2222}
  If $(S^2,A, \ord)$ is a $(2,2,2,2)$-orbisphere, then $\pi_1(S^2,A,
  \ord) \cong \Z^2\rtimes\{\pm1\}$. There are exactly four
  order-$2$ conjugacy classes in $ \pi_1(S^2,A, \ord) \cong
  \Z^2\rtimes\{\pm1\}$; these classes are identified with $A$ and
  are of the form
\[(n,1)^{ \Z^2\rtimes\{\pm1\}}=\{(n+2m,1)\mid m\in \Z^2\} \text{ for all }n\in\{0,1\}^2.\]
\end{lem}

Let us fix a $(2,2,2,2)$-orbisphere $(S^2,A, \ord)$ and let us set
$G\coloneqq \pi_1(S^2,A, \ord)$. Thanks to Lemma~\ref{lem:Gr:2222} we
identify $A$ with the set of all order-$2$ conjugacy classes of
$G$. By Euler characteristic, every branched covering $f\colon (S^2,A,
\ord) \selfmap$ is a self-covering. Therefore, the biset of $f$ is
right principal.
 
We denote by $\Mat_2^+(\Z)$ the set of $2\times2$ integer matrices $M$
with $\det(M)>0$. For a matrix $M\in \Mat^+_2(\Z)$ and a vector
$v\in \Z^2$ there is an injective endomorphism
$M^v\colon \Z^2\rtimes\{\pm1\} \selfmap$ given by the following
``crossed product'' structure:
\begin{equation}
\label{eq:InjEndOfK}
M^{v}(n,0)=(Mn,0) \text{ and }M^{v}(n,1)=(Mn+v,1).
\end{equation}  
Furthermore, $M^v$ induces a map $(M^v)_*\colon A \selfmap$ on
conjugacy classes. We write $B_{M^v}=B_M\rtimes\{\pm1\}$ the crossed
product decomposition of the biset of $M_v$.
 
\begin{prop}\label{prop:bis:2222}
  Every $(2,2,2,2)$-orbisphere biset $B$ is of the form $B=B_{M^v}$
  for some $M^v$ as in~\eqref{eq:InjEndOfK}, and $M^v$ is computable
  from the biset $B$.

  Conversely, $B_{M^v}$ is an orbisphere biset for every $M^v$ as
  in~\eqref{eq:InjEndOfK}. Two bisets $B_{M^v}$ and $B_{N^w}$ are
  isomorphic if and only if $M=\pm N$ and $(M^v)_*=(N^w)_*$ as maps on
  $A$.

  The biset $B_{M^{v}}$ is geometric if and only if both eigenvalues
  of $M$ are different from $\pm 1$. If $B_{M^{v}}$ is geometric, then
  for every $(B_{M^{v}})_*\colon A\sqcup E\selfmap$ the number of
  conjugacy classes of portraits of bisets in $B_{M^{v}}$
  parameterized by $(B_{M^{v}})_*$ is finite, and it is
  algorithmically decidable whether for a given
  $B_*\colon A\sqcup E\selfmap$ two portraits of bisets within
  $B_{M^v}$ are conjugate.
\end{prop} 

\noindent We need the following fact.
\begin{thm}[Corollary of~\cite{grunewald:conjugacyarithmetic}]\label{thm:MatrConj}
  There is an algorithm deciding whether two $M,N\in \Mat_2^+(\Z)$ are
  conjugate by an element $X\in \SL_2(\Z)$, and produces such an $X$
  if it exists.
 
  There is an algorithm computing, as a finitely generated subgroup of
  $\SL_2(\Z)$, the centralizer of $M\in \Mat_2^+(\Z)$.
\end{thm}

\begin{algo}\label{algo:torbisets}
  \textsc{Given} $\subscript{\widetilde G}{\widetilde B}_{\widetilde G}$, $\subscript{\widetilde G}{\widetilde C}_{\widetilde C}$ two \Tor\ bisets\\
  \textsc{Compute} the centralizer $Z(\widetilde B)$, and \textsc{decide} whether $\widetilde B$ and $\widetilde C$ are conjugate by an element of $\Mod(\widetilde G)$, and if so \textsc{construct} a conjugator \textsc{as follows:}\\\upshape
  \begin{enumerate}
  \item If $\widetilde B_*\neq \widetilde C_*$ as maps on peripheral
    conjugacy classes, then return \texttt{fail}.
  \item Let $\subscript{\widetilde G}{\widetilde B}_{\widetilde G}\to
    \subscript G B_G$ and $\subscript{\widetilde G}{\widetilde C}_{\widetilde G}\to
    \subscript{G'}C_{G'}$ be two maximal forgetful morphisms. If $G\neq G'$,
    then return \texttt{fail}. Otherwise by Lemma~\ref{lem:Gr:2222} identify
    $G=G'\cong \Z^2\rtimes\{\pm1\}$ and by
    Proposition~\ref{prop:bis:2222} present $B$ and $C$ as $B_{M^v}$
    and $B_{N^w}$ respectively.
  \item Using Theorem~\ref{thm:MatrConj} check whether $M$ and $N$ are
    conjugate. If not, return \texttt{fail}; otherwise find a conjugator $X$
    and compute the centralizer subgroup $K$ of $M$.
  \item Check whether there is a $Y\in K$ such that $(Y X)^0$ is a
    conjugator between $M^v$ and $N^w$. If there is none, return
    \texttt{fail}; otherwise set $X\coloneqq Y X$ and replace $K$ by
    $\{Y\in K\mid Y^0\text{ centralizes }M^v\}$, a subgroup of finite
    index in $K$.
  \item Let $(G_a,B_a)_{a\in A\sqcup E}$ and $(G_a,C_a)_{a\in A\sqcup
      E}$ be induced by $\widetilde B$ and $\widetilde C$ portrait of
    bisets in $B\cong B_{M^v}$ and in $C\cong B_{N^w}$. Using
    Proposition~\ref{prop:bis:2222} check whether the is an $Y\in K$
    such that
    \[\left(G^{(Y X)^0}_a,B^{(Y X)^0}_a\right)_{a\in A\sqcup E}\sim (G_a,C_a)_{a\in A\sqcup E}.\] 
    If not, return \texttt{fail}. Otherwise use
    Theorem~\ref{thm:RedConjCentrProb} to promote $(Y X)^0$ into a
    conjugacy between $\widetilde B$ and $\widetilde C$.
  \item The centralizer of $\widetilde B$ is computed using
    Theorem~\ref{thm:RedConjCentrProb}.
  \end{enumerate}
\end{algo}

\begin{cor}\label{cor:tordecidable}
  There is an algorithm that, given two \Tor\ bisets $\subscript GB_G$ and
  $\subscript H C_H$, decides whether $B$ and $C$ are conjugate, and computes
  the centralizer of $B$.
\end{cor}

\section{Expanding maps and the Levy decomposition~\cite{bartholdi-dudko:bc4}}\label{ss:expanding}
Consider a Thurston map $f\colon(S^2, A)\selfmap$. We give a criterion
for the existence of a Riemannian metric on $(S^2,A)$ such that $f$ is
isotopic to an expanding map. This criterion is in terms of
multicurves on $S^2\setminus A$. We then give an application to the
study of \emph{matings} of polynomials.

\subsection{Levy, anti-Levy, Cantor, and anti-Cantor multicurves}
Recall that a multicurve $\CC$ is invariant if $f^{-1}(\CC)=\CC$, up
to isotopy and removing peripheral and trivial components.  If $C$ is
a multicurve and $C\subseteq f^{-1}(C)$, then there is a unique
invariant multicurve $\CC$ \emph{generated} by $C$, namely
$\CC=\bigcup_{n\ge 0} f^{-n}(C)$.

Consider the following graph called the \emph{curve graph} of
$(S^2,A)$. Its vertex set is the set of isotopy classes of essential
curves on $S^2\setminus A$. For every simple closed curve $\gamma$ and
for every component $\delta$ of $f^{-1}(\gamma)$ there is an edge from
$\gamma$ to $\delta$ labeled $\deg(f\restrict\delta)$.

Let $\CC$ be an invariant multicurve, and consider the directed
subgraph of the curve graph that it spans. A \emph{strongly connected
  component} is a maximal subgraph spanned by a subset $C\subseteq\CC$
such that, for every $\gamma,\delta\in C$, there exists a non-trivial
path from $\gamma$ to $\delta$ in $C$. A strongly connected component
$C$ is \emph{primitive} if all its incoming edges come from $C$
itself. We call $C$ a \emph{bicycle} if for every $\gamma,\delta\in C$
there exists $n\in\N$ such that at least two paths of length $n$ join
$\gamma$ to $\delta$ in $C$, and a \emph{unicycle} otherwise; see
Figure~\ref{Fig:ExampleCantorMult} for an illustration.  Clearly,
every invariant multicurve is generated by its primitive strongly
connected components.

We remark that bicycles contain at least two cycles, so that the
number of paths of length $n$ grows exponentially in $n$. On the other
hand, every unicycle is an actual \emph{periodic cycle}: it can be
written as $C=(\gamma_0,\gamma_1,\dots,\gamma_n=\gamma_0)$ in such a
manner that $\gamma_{i+1}$ has an $f$-preimage $\gamma'_i$ isotopic to
$\gamma_i$. If furthermore the $\gamma'_i$ may be chosen so that $f$
maps each $\gamma'_i$ to $\gamma_{i+1}$ by degree $1$, then $C$ is
called a \emph{Levy cycle}.

\[\begin{tikzpicture}[->,auto,vx/.style={circle,minimum size=1ex,inner sep=0pt,outer sep=2pt,draw,fill=gray!50}]
  \node[vx] (1) at (-1,0) {};
  \node[vx] (2) at (1,0) {};
  \draw (1)  to [bend left=60] (2);
  \draw (1)  to [bend left=10] (2);
  \draw (2)  to [bend left] (1);
  \node at (0,-0.6) {bicycle};
\end{tikzpicture}
\qquad
\begin{tikzpicture}[->,auto,vx/.style={circle,minimum size=1ex,inner sep=0pt,outer sep=2pt,draw,fill=gray!50}]
  \node[vx] (1) at (-1,0) {};
  \node[vx] (2) at (0,1.5) {};
  \node[vx] (3) at (1,0) {};
  \draw (1)  to [bend left=10] node {$1:1$} (2);
  \draw (2)  to [bend left=10] node {$1:1$} (3);
  \draw (3)  to [bend left=10] node [above] {$1:1$} (1);
  \node at (0,-0.6) {Levy cycle};
\end{tikzpicture}
\qquad
\begin{tikzpicture}[->,auto,vx/.style={circle,minimum size=1ex,inner sep=0pt,outer sep=2pt,draw,fill=gray!50}]
  \node[vx] (1) at (-1,0) {};
  \node[vx] (2) at (-1,1) {};
  \node[vx] (3) at (1,0) {};
  \node[vx] (4) at (1,1) {};
  \draw[dashed,thick] (-1,0.5) ellipse (6mm and 9mm) node[above right=7mm] {$C$};
  \draw (1)  to [bend left] (2);
  \draw (2)  to [bend left] (1);
  \draw (1)  to (3);
  \draw (3)  to [bend left] (4);
  \draw (4)  to [bend left] (3);
  \node at (0,-0.6) {Primitive s.c.c.};
\end{tikzpicture}
\]

\begin{defn}[Types of invariant multicurves]
  Let $\CC$ be an invariant multicurve. Then $\CC$ is
  \renewcommand{\descriptionlabel}[1]{\hspace{\labelsep}\textbf{\emph{#1}}}
  \begin{description}
  \item[Cantor] if $\CC$ is generated by its bicycles;
  \item[anti-Cantor] if $\CC$ does not contain any bicycle;
  \item[Levy] if $\CC$ is generated by its Levy cycles;
  \item[anti-Levy] if $\CC$ does not contain any Levy cycle.\qedhere
  \end{description}
\end{defn}

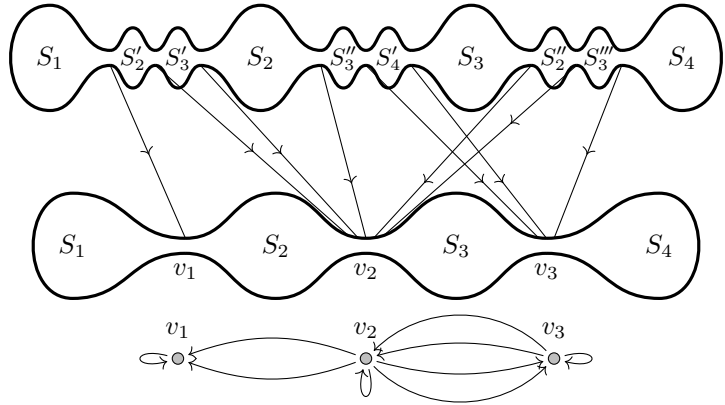
\begin{figure}
  \begin{tikzpicture}[auto,vx/.style={circle,minimum size=1ex,inner sep=0pt,outer sep=2pt,draw,fill=gray!50}]
    \def\leftsphere#1#2{+(0,0.1) .. controls +(180:#1/2) and +(0:#1/2) .. +(-#1,#2)
      .. controls +(180:#2) and +(180:#2) .. +(-#1,-#2)
      .. controls +(0:#1/2) and +(180:#1/2) .. +(0,-0.1) ++(0,0)}
    \def\rightsphere#1#2{+(0,0.1) .. controls +(0:#1/2) and +(180:#1/2) .. +(#1,#2)
      .. controls +(0:#2) and +(0:#2) .. +(#1,-#2)
      .. controls +(180:#1/2) and +(0:#1/2) .. +(0,-0.1) ++(0,0)}
    \def\midsphere#1#2{+(0,0.1) .. controls +(0:#1/2) and +(180:#1/2) .. +(#1,#2)
      .. controls +(0:#1/2) and +(180:#1/2) .. +(#1*2,0.1)
      +(0,-0.1) .. controls +(0:#1/2) and +(180:#1/2) .. +(#1,-#2)
      .. controls +(0:#1/2) and +(180:#1/2) .. +(#1*2,-0.1) ++(#1*2,0)
    }
    
    \draw[very thick] (-2.4,0) node [xshift=-15mm] {$S_1$}
    \leftsphere{1.5}{0.7} node [below=1mm] {$v_1$} node[xshift=12mm] {$S_2$}
    \midsphere{1.2}{0.7} node[below=1mm] {$v_2$} node[xshift=12mm] {$S_3$}
    \midsphere{1.2}{0.7} node[below=1mm] {$v_3$} node[xshift=15mm] {$S_4$}
    \rightsphere{1.5}{0.7};

    \draw[very thick] (-3.4,2.5) node [xshift=-8mm] {$S_1$}
    \leftsphere{0.8}{0.7} node [xshift=3mm] {\small $S'_2$}
    \midsphere{0.3}{0.4} node [xshift=3mm] {\small $S'_3$}
    \midsphere{0.3}{0.4} node [xshift=8mm] {$S_2$} \midsphere{0.8}{0.7}
    node [xshift=3mm] {\small $S''_3$} \midsphere{0.3}{0.4} node
    [xshift=3mm] {\small $S'_4$} \midsphere{0.3}{0.4} node [xshift=8mm]
    {$S_3$} \midsphere{0.8}{0.7} node [xshift=3mm] {\small $S''_2$}
    \midsphere{0.3}{0.4} node [xshift=3mm] {\small $S'''_3$}
    \midsphere{0.3}{0.4} node [xshift=8mm] {$S_4$} \rightsphere{0.8}{0.7};

    \draw (-3.4,2.4) -- (-2.4,0.1) \fwdarrowonline{0.5};
    \draw (-2.8,2.4) -- (-0.25,0.15) \fwdarrowonline{0.5};
    \draw (-2.2,2.4) -- (-0.1,0.1) \fwdarrowonline{0.5};
    \draw (-0.6,2.4) -- (0.0,0.1) \fwdarrowonline{0.667};
    \draw (0.0,2.4) -- (2.3,0.1) \fwdarrowonline{0.667};
    \draw (0.6,2.4) -- (2.4,0.1) \fwdarrowonline{0.667};
    \draw (2.2,2.4) -- (0.1,0.1) \bckarrowonline{0.667};
    \draw (2.8,2.4) -- (0.25,0.15) \bckarrowonline{0.333};
    \draw (3.4,2.4) -- (2.5,0.1) \bckarrowonline{0.5};

    \node[vx,label={$v_1$}] (v1) at (-2.5,-1.5) {};
    \node[vx,label={$v_2$}] (v2) at (0,-1.5) {};
    \node[vx,label={$v_3$}] (v3) at (2.5,-1.5) {};
    \draw[<-] (v1)  to [loop left] (v1);
    \draw[<-] (v1)  to [bend left=20] (v2);
    \draw[<-] (v1)  to [bend right=20] (v2);
    \draw[->] (v2)  to [loop below] (v2);
    \draw[->] (v2)  to [bend right=15] (v3);
    \draw[->] (v2)  to [bend right=45] (v3);
    \draw[->] (v3)  to [loop right] (v3);
    \draw[->] (v3)  to [bend right=15] (v2);
    \draw[->] (v3)  to [bend right=45] (v2);
  \end{tikzpicture}
  \caption{A bicycle $\{v_2,v_3\}$ generates a Cantor multicurve
    $\{v_1,v_2,v_3\}$. The action of the map $f$ is indicated on the
    preimages of $\{v_1,v_2,v_3\}$. If annuli are mapped by degree
    $1$, then it is also a Levy cycle. The graph below is the
    corresponding portion of the curve
    graph.} \label{Fig:ExampleCantorMult}
\end{figure}

Note that $\CC$ being Cantor / anti-Cantor depends on the mappings of
curves of $\CC$ to themselves, and being Levy / anti-Levy depends on
the degrees under which the curves map.

Suppose $f\colon(S^2, A)\selfmap$ is a Thurston map with an invariant
multicurve $\CC$. Recall from~\S\ref{ss:multicurves} that by
$R(f,\CC)$ we denote the return maps induced by $f$ on
$S^2\setminus\CC$.

\begin{prop}\label{Prop:CantorMultCurv}
  Suppose $f\colon(S^2, A)\selfmap$ is a Thurston map with an
  invariant multicurve $\CC$. Then
  \begin{enumerate}
  \item there is a unique maximal invariant Cantor sub-multicurve
    $\CC_{\text{Cantor}}\subseteq\CC$ such that the multicurves
    induced by $\CC$ on pieces in $R(f,\CC_{\text{Cantor}})$ are
    anti-Cantor invariant multicurves;
  \item there is a unique maximal invariant Levy sub-multicurve
    $\CC_{\text{Levy}}\subseteq\CC$ such that multicurves induced by
    $\CC$ on pieces in $R(f,\CC_{\text{Levy}})$ are anti-Levy
    invariant multicurves.
\end{enumerate}
\end{prop}

\begin{defn}[Levy-free]\label{defn:levyfree}
  Let $f\colon(S^2,A)\selfmap$ be a Thurston map. It is
  \textit{Levy-free} if $\deg(f)>1$ and $f$ does not admit a Levy
  cycle.
\end{defn}

We say that an invariant Levy multicurve $\CC$ is \textit{complete} if
every piece in $R(f,\CC)$ is either Levy-free or has degree $1$. We
show that for a Thurston map $f\colon(S^2,A)\selfmap$ there is a
unique minimal invariant complete Levy multicurve, which we denote by
$\CC_{f,\text{Levy}}$ and call \emph{the canonical Levy obstruction}.
Any other invariant complete Levy multicurve $\CC$ contains
$\CC_{f,\text{Levy}}$ as a sub-multicurve. For $B$ a biset, we define
in a similar manner $\CC_{B,\text{Levy}}$.

\begin{defn}[Levy decomposition]\label{def:HypDecomp}
  The Levy decomposition of $f\colon(S^2,A)\selfmap$ is the
  decomposition along the canonical Levy obstruction
  $\CC_{f,\text{Levy}}$.

  The Levy decomposition of a biset $B$ is the sphere tree of bisets
  decomposition of $B$ along $\CC_{B,\text{Levy}}$.
\end{defn}

\subsection{Expanding maps}
A \emph{length metric with singularities} on $S^2$ is a length
orbifold metric that is allowed to have finitely many points, called
\emph{singularities}, at infinite distance such that points
topologically close to singularities are far away from usual points. A
basic example is the hyperbolic metric on the Riemann sphere with
finitely many removed points.  We will also refer to singularities as
\emph{points at infinity}.

A Thurston map $f\colon(S^2,A)\selfmap$ is \emph{metrically expanding}
if there is a length metric $\mu$ on $S^2$, with singularities, such
that $f$ is expanding with respect to $\mu$ and all points
sufficiently close topologically to infinity escape to infinity. A
combinatorial equivalence class of Thurston maps is called
\emph{expanding} if it contains an expanding map.

It follows from the definition that the set of singularities of
$\mu$ is forward invariant, and a periodic point is singular if
and only if it is topologically attracting.

We say that an expanding map is \emph{B\"ottcher} if the first
return map near every critical periodic point is conjugate to
$z\to z^d$, where $d>1$ is the degree of the first return map.  Two
B\"ottcher expanding maps are conjugate if and only if they are
combinatorially equivalent. (This is an application of the ``pullback
argument'': if we have $ f\circ \phi =g\colon (S^2,A)\selfmap$ with
$\phi$ isotopic rel $A$ to the identity, then $\phi$ can be normalized
to be the identity near every periodic critical cycle; by expansion
the lifts $\phi_n$ of $\phi$ through $f^n$ converge exponentially fast
to the identity, so their product
$\psi\coloneqq \cdots\circ\phi_1\circ\phi_0$ conjugates $f$ into $g$.)

\begin{defn}[\cite{nekrashevych:ssg}*{Definition~2.11.8}]\label{defn:contracting}
  Let $B$ be a $G$-$G$-biset. It is called \emph{contracting} if for
  every finite $S\subseteq B$ there
  exists a finite subset $N\subseteq G$ with the following property:
  for every $g\in G$ and every $n>\!\!>0$ we have
  $\{h\in G\mid h S^n\cap S^n g\neq\emptyset\}\subseteq N$.
\end{defn}
For a finitely generated group $G$, the biset $B$ is contracting if
there is a proper metric $|\cdot|$ on $G$ and constants $\lambda<1,C$
such that $|h|\le\lambda|g|+C$ whenever $h S\cap S g\neq\emptyset$.

If $B$ is left-free and we chose a basis $S\subseteq B$ for the left
action, we obtain a wreath recursion
$\psi\colon G\to G^S\rtimes S\perm$, see~\S\ref{ss:coverings}; and $B$
is contracting if for every $g\in G$ one obtains only elements of $N$
as co\"ordinates when one iterates $\psi$ long enough.

The set $N$ in Definition~\ref{defn:contracting} is not unique; but
for every $S\subseteq B$ there exists a minimal such $N$, written
$N(S)$ and called the \emph{nucleus} of $(B,S)$. It gives rise to a
labeled graph, called the \emph{nucleus machine} of $(B,S)$: its
vertex set is $N(S)$, and there is an edge from $g$ to $h$ with
\emph{input} and \emph{output} labels $s\in S$ and $t\in S$
respectively whenever $s g=ht$ holds in $B$.

We slightly modify the definition of ``contracting'' for sphere bisets,
because of the orbisphere structures. Let $\subscript GB_G$ be a sphere biset with
$G=\pi_1(S^2,A)$. Recall from~\eqref{eq:ordB} that there is a minimal
orbisphere structure $\ord_B$ given by $B$. We call an orbisphere structure
$\ord\colon A\to\{2,3,\dots,\infty\}$ \emph{bounded} if
$\ord(a)=\infty\Leftrightarrow\ord_B(a)=\infty$ and
$\ord(a)\deg_a(B)\mid\ord(B_*(a))$ for all $a\in A$. Let $\overline G$
denote the quotient orbisphere group $G/\langle\gamma_a^{\ord(a)}:a\in
A\rangle$. Then we call $B$ an \emph{orbisphere contracting} biset if
$\overline G\otimes_G B\otimes_G\overline G$ is contracting for some
bounded orbisphere structure on $(S^2,A)$.

\noindent The main result of this part is the following criterion:
\begin{mainthm}\label{thm:ExpCr} 
  Let $f\colon(S^2,A)\selfmap$ be a Thurston map, not doubly covered
  by a torus endomorphism. The following are equivalent:
  \begin{enumerate}
  \item $f$ is isotopic to a B\"ottcher expanding map;
  \item $f$ is Levy-free;
  \item $B(f)$ is an orbisphere contracting biset.
  \end{enumerate}

  Furthermore, if these properties hold, the metric $\mu$ on $S^2$
  that is expanded by $f$ may be assumed to be Riemannian of pinched
  negative curvature.
\end{mainthm}

Ha\"\i ssinsky and Pilgrim ask in~\cite{haissinsky-pilgrim:algebraic}
whether every everywhere-expanding map is isotopic to a smooth map. By
Theorem~\ref{thm:ExpCr}, a combinatorial equivalence class contains a
B\"ottcher smooth expanding map if and only if it is Levy free.

It was already proven in~\cite{selinger-yampolsky:geometrization} that
every Levy-free map that is doubly covered by a torus endomorphism is
in \Tor. Combining this with the
results~\cites{nielsen:surfaces,thurston:surfaces} on geometrization of
surface self-homeomorphisms we obtain the
\begin{cor}
  Let $f\colon(S^2,A)\selfmap$ be a Thurston map. Then every map in
  $R(f,\CC_{\text{Levy}})$ has a geometric structure: it is either
  expanding (i.e.\ in \Exp), of degree $1$, or a non-expanding
  irrational map doubly covered by a torus endomorphism (i.e.\ in
  $\Tor\setminus\Exp$).
\end{cor}
Furthermore, the property of being geometric is algorithmically
recognizable: bisets double covered by a torus endomorphism are of a
very particular nature~\eqref{eq:InjEndOfK}; contracting bisets are recognized by their
nucleus, and non-contracting bisets are recognized by their Levy
obstruction.

\begin{algo}\label{algo:levy}
  \textsc{Given} a Thurston map $f\colon(S^2,A)\selfmap$ by its biset,\\
  \textsc{Compute} the Levy decomposition of $f$ \textsc{as follows:}\\\upshape
  \begin{enumerate}
  \item For an enumeration of all multicurves $\CC$ on $(S^2,A)$, that
    never reaches a multicurve before reaching its proper
    submulticurves, do the following steps:
  \item If the multicurve $\CC$ is not invariant, or is not Levy,
    continue in~(1) with the next multicurve.
  \item Compute the decomposition of $f$ using
    Algorithm~\ref{algo:decompose}.
  \item If not all pieces are expanding or \Tor\ or degree-$1$ maps,
    continue in~(1) with the next multicurve;
  \item Return $\CC$.
  \end{enumerate}
\end{algo}

\subsection{Matings of higher degree polynomials}\label{ss:matings}
We turn to an application to matings of polynomials. Let
$p,q\colon\C\selfmap$ be two monic polynomials of same degree
$d\ge2$. Identify $\C$ with the disk $\mathbb D$ by the map
$\nu_+\colon z\mapsto z/\sqrt{1+|z|^2}$ and with its complement
$\hC\setminus\overline{\mathbb D}$ by the map
$\nu_-(z)=1/\nu_+(z)$. Consider the continuous map
\[p\FM q\colon\begin{cases}\hC & \to\hC,\\
  z\text{ with }|z|<1 & \mapsto \nu_+(p(\nu_+^{-1}(z))),\\
  z\text{ with }|z|=1 & \mapsto z^d,\\
  z\text{ with }|z|>1 & \mapsto \nu_-(q(\nu_-^{-1}(z))).
\end{cases}\]
It is called the \emph{formal mating} of $p$ and $q$, and is a
degree-$d$ branched covering of $S^2$.

Recall also that there are \emph{external rays} associated to the
polynomials $p,q$. First, the \emph{filled-in Julia set} $K_p$ of $p$ is
\[K_p=\{z\in\C\mid f^n(z)\not\to\infty\text{ as }n\to\infty\}.\]
Assume that $K_p$ is connected. There exists then a unique holomorphic
isomorphism
$\phi_p\colon\hC\setminus K_p\to\hC\setminus\overline{\mathbb D}$
satisfying $\phi_p(p(z))=\phi_p(z)^d$ and $\phi_p(\infty)=\infty$ and
$\phi_p'(\infty)=1$. It is called a \emph{B\"ottcher co\"ordinate}, and
conjugates $p$ to $z^d$ in a neighbourhood of $\infty$. For
$\theta\in\R/\Z$, the associated \emph{external ray} is
\[R_p(\theta)=\{\phi_p^{-1}(re^{2i\pi\theta})\mid r\ge1\}.\] Let $X_{p,q}$
be the quotient of $\hC$ in which each $\overline{\nu_+(R_p(\theta))}$ has
been identified to one point for each $\theta\in\R/\Z$, and similarly each
$\overline{\nu_-(R_q(\theta))}$ has been identified to one point. Note that
$X_{p,q}$ is a quotient of $K_p\sqcup K_q$, and need not be a Hausdorff
space. A classical criterion (due to Moore) determines when $X_{p,q}$ is
homeomorphic to $S^2$. If this occurs, $p$ and $q$ are said to be
\emph{topologically mateable}, and the map induced by $p\FM q$ on $X_{p,q}$
is called the \emph{topological mating} of $p$ and $q$ and denoted $p\GM
q\colon X_{p,q}\selfmap$.

Furthermore, if there exists a homeomorphism
$\phi\colon X_{p,q}\to\hC$ that is conformal on
$\nu_+(K_p^o)\cup\nu_-(K_q^o)$ and such that
$f\coloneqq\phi\circ(p\GM q)\circ\phi^{-1}$ is a rational map, then
$p,q$ are said to be \emph{geometrically mateable}, and any such $f$
is called a \emph{geometric mating} of $p$ and $q$.

Mary Rees and Tan Lei~\cite{tan:matings} proved that two
post-critically finite quadratic polynomials $p,q$ are geometrically
mateable if and only if $p$ and $q$ do not belong to conjugate limbs
of the Mandelbrot set;
see~\cite{buff+:questionsaboutmatings}*{Theorem~2.1}. To be more
precise, let us assume that $p$ and $q$ are hyperbolic. (The
sub-hyperbolic case requires a slight clarification because points in
the post-critical set may get glued during the geometric mating.) Then
the following are equivalent for $d=2$:
\begin{enumerate}
\item $p \FM q$ is isotopic to a rational map $p \square q$;
\item $p \GM q$ is a sphere map (necessarily conjugate to $p \square q$);
\item $p,q$ are not in conjugate primary limbs of the Mandelbrot
  set.
\end{enumerate}
This theorem relies on the fact that any annular obstruction in
degree $2$ is Levy; indeed in degree $\ge3$ there are topological
matings that are not conjugate to rational maps,
see~\cite{shishikura-t:matings} and~\S\ref{ex:tls}. Since the
obstruction to be an expanding map is also Levy, we can generalize the
degree $2$ criterion in the class of expanding maps as follows.

\begin{defn}
  Let $p,q$ be two hyperbolic post-critically finite polynomials of
  same degree $d$. We say that $p,q$ have a \emph{pinching cycle of
    periodic angles} if there are angles
  $\phi_0,\phi_1, \dots, \phi_{2n-1}\in \R/\Z$, indices treated modulo
  $2n$, such that for all $i$,
  \begin{itemize}
  \item the angle $\phi_i$ is periodic under multiplication by $d$;
  \item the rays $R_p(\phi_{2i})$ and $R_p(\phi_{2i+1})$ land
    together;
  \item the rays $R_q(-\phi_{2i-1})$ and $R_q(-\phi_{2i})$ land
    together.\qedhere
  \end{itemize}
\end{defn}
In degree $2$, the pair $p,q$ has a pinching cycle of periodic angles if
and only if $p$ and $q$ are in the conjugate primary limbs of the
Mandelbrot set.  In degree $\ge3$, there is still no well-accepted
notion of limbs, and they should be rather defined as sets of
parameters in which certain periodic rays land together.

If $p,q$ admit a pinching cycle, then a fibre of the
map $\hC\to X_{p,q}$ is a cycle, so $X_{p,q}$ cannot be homeomorphic
to $S^2$. We obtain the following criterion:
\begin{thm}
  Let $p,q$ be two monic hyperbolic post-critically finite polynomials. Then
  the following are equivalent:
  \begin{enumerate}
  \item $p \FM q$ is isotopic to an expanding map $p \square q$;
  \item $p \GM q$ is a sphere map (necessarily conjugate to $p \square q$);
  \item\label{thm:mating:3} $p, q$ do not have a pinching cycle of
    periodic angles.
  \end{enumerate}
\end{thm}

Furthermore, pinching cycles of periodic angles are effectively
enumerable using the nuclei of the bisets $B(p)$ and $B(q)$.

\section{Symbolic and floating-point algorithms~\cite{bartholdi-dudko:bc5}}
In this part we describe in more details the computational techniques
available to manipulate Thurston maps. In particular, we give a range
of symbolic algorithms, converters from one formalism into another,
and describe hybrid algorithms that bridge between the group theory
side and the complex analytic sides of Thurston maps. Many of these
algorithms are already implemented in a GAP package
\textsc{Img}~\cite{img:manual} available from the GAP website.

The main objects that the algorithms manipulate are bisets. They may
be constructed using classical data such as external angles,
floating-point approximations of maps, Hubbard trees, subdivision
rules, etc. Many operations such as matings, tunings and composition
with twists may be naturally implemented on them.

On the one hand, bisets are often constructed by choosing a basis and
expressing its wreath recursion in that basis. Bisets, however, are
fundamentally basis-free objects, and isomorphism, congruence etc.\ of
bisets is discovered by finding bases in which their recursions
coincide. Classical notions such as supporting rays, Hubbard trees
etc.\ are viewed as invariants of the biset. Twisting of bisets, or
more generally compositions with homeomorphisms, are implemented by
changing the bases.

\subsection{Converters}
In order to compute with Thurston maps, it is important to represent
them, and the objects they are related to, in an efficient manner.

Consider a Thurston map $f\colon(S^2,A)\selfmap$. In case $f$ is a
topological polynomial, namely there exists a point $\infty\in A$ with
$f^{-1}(\infty)=\{\infty\}$, then, on the one hand, the biset of $f$
can be expressed in an adapted basis, in which its structure is
essentially equivalent to a description by external angles. In case
$\#A=4$ and all critical points of $f$ are simple, it is possible to
express $f$ as a ``near-Euclidean map'' (NET),
see~\cite{cannon-floyd-parry-pilgrim:net}. We consider here the
general case.

Sphere groups themselves are represented by their number of generators
and cyclic ordering of peripheral generators, as
in~\eqref{eq:spheregroup}. For each sphere group, an isomorphism with
an underlying free group is chosen so as to allow fast comparison
between elements. Elements in sphere groups are represented as words
in the peripheral generators and their inverses.

Mapping class groups are described as outer automorphisms of sphere
groups. They are represented by keeping track of the images of
generators. This list of images is reduced to a lexicographically
minimal one by conjugating it diagonally by the sphere
group. Stallings' graphs are used to manipulate them, and in
particular compute their inverse.

All our bisets are left-free, so are represented by their associated
wreath recursions: if the right-acting sphere group has $n$ generators
and the biset has degree $d$, the bisets are represented as a list of
$n$ permutations in $d\perm$ and an $n\times d$ table of group
elements.

\begin{algo}\label{algo:angle2biset}
  \textsc{Given} a degree $d$ and a list of tuples of rational numbers representing supporting rays \`a la Poirier~\cite{poirier:portraits},\\
  \textsc{Compute} the biset of the polynomial that it describes.
\end{algo}

\begin{algo}\label{algo:biset2angle}
  \textsc{Given} a sphere biset $B$,\\
  \textsc{Compute} the supporting rays \`a la Poirier~\cite{poirier:portraits}, if $B$ is combinatorially equivalent to a polynomial, and \texttt{fail} else.
\end{algo}

\begin{algo}\label{algo:hubbardtree2biset}
  \textsc{Given} an angled Hubbard tree endowed with a self-map,\\
  \textsc{Compute} the biset of the polynomial that it describes.
\end{algo}

Let $f\colon (S^2,A)\selfmap$ be a Thurston map. It may be
conveniently encoded combinatorially as follows:
\begin{enumerate}
\item A triangulation $\mathcal T_0$ of $S^2$ with vertex set $A$ is chosen;
\item The lift $\mathcal T_1$ of $\mathcal T_0$ under $f$ is computed;
  it is a triangulation of $S^2$ with vertex set $f^{-1}(A)$, and $f$
  is expressed as a simplicial map $\mathcal T_1\to\mathcal T_0$;
\item A refinement $\mathcal T'_1$ of $\mathcal T_0$, also with vertex
  set $f^{-1}(A)$, is chosen, by subdividing in $\mathcal T_0$ all
  triangles containing elements of $f^{-1}(A)\setminus A$;
\item The relation between $\mathcal T_1$ and $\mathcal T'_1$ is
  encoded by computing for each edge of $\mathcal T_1$ the sequence of
  edges of $\mathcal T'_1$ that it crosses.\label{enum:retriangulate}
\end{enumerate}

The triangulation $\mathcal T_0$ expresses in a combinatorial manner
the fundamental group of $S^2\setminus A$: there is a retraction from
$S^2\setminus A$ to the dual graph $\mathcal T_0^\perp$ of
$\mathcal T_0$, so
$\pi_1(S^2\setminus A)\cong\pi_1(\mathcal
T_0^\perp)$. In~\eqref{enum:retriangulate} we give a homotopy
equivalence between $\mathcal T_1^\perp$ and ${\mathcal T'_1}^\perp$,
called a \emph{retriangulation}.

Conversely, a Thurston map is specified uniquely by the above data: a
triangulation $\mathcal T_0$, a simplicial map
$f\colon\mathcal T_1\to\mathcal T_0$, a refinement
$\mathcal T'_1\supseteq\mathcal T_0$, and a retriangulation of
$\mathcal T'_1$ into $\mathcal T_1$ as specified
in~\eqref{enum:retriangulate}.

\emph{De facto}, we are viewing $f$ as a covering pair
$f,i\colon(S^2,f^{-1}(A))\rightrightarrows(S^2,A)$; we consider a
triangulation $\mathcal T_0$ of the range, compute its preimages
$\mathcal T_1$ and $\mathcal T'_1$ under $f$ and $i$ respectively, and
encode combinatorially the relation between these triangulations at
the source.

In fact, in some situations we already start with a combinatorial
description of a map, rather than an actual Thurston map on a
topological sphere. Most classical combinatorial descriptions are very
close to the one given above. We argue in this text that a description
in terms of bisets is algorithmically the most useful; it can be
easily obtained from the combinatorial description sketched above:
\begin{algo}\label{algo:tri2biset}
  \textsc{Given} triangulations $\mathcal T_1,\mathcal T'_1,\mathcal T_0$ of $S^2$ such that $\mathcal T'_1$ is a refinement of $\mathcal T_0$, and given a simplicial map $f\colon\mathcal T_1\rightarrow\mathcal T_0$ as well as a retriangulation $\mathcal T_1\leftrightarrow \mathcal T'_1$,\\
  \textsc{Compute} the biset of $f$.
\end{algo}

\subsection{Floating-point algorithms}
When we write that something is computable over $\C$, we mean, in the
following strong sense, that all elements of $\C$ that we refer to are
algebraic numbers $\xi$, that a minimal polynomial of $\xi$ over $\Z$
is computable, and that a rectangle with rational corners is
computable, that separates $\xi$ from its Galois conjugates.

On the other hand, a \emph{floating-point approximation} of $z\in\C$
is an algorithm producing an infinite data stream
$(x_n+y_n i,\epsilon_n)$, with $x_n,y_n,\epsilon_n\in\Q$ and
$|x_n+y_n i-z|<\epsilon_n$ and $\epsilon_n\to 0$.

Every computable number over $\C$ admits a floating-point
approximation. Conversely, given a floating-point approximation
$(x_n+y_n i,\epsilon_n)$ and the knowledge that its limit $\xi$ is
algebraic with given degree and bound on the coefficients of its
minimal polynomial, the number $\xi$ is computable over $\C$. This
only uses standard algorithms for root finding, root isolation and
lattice reduction (LLL).

A floating-point approximation of a rational map
$f\colon (S^2,A)\selfmap$ is an algorithm producing an infinite data
stream $(f_n,A_n,\epsilon_n)$ with $f_n$ a rational map with rational
coefficients, $A_n$ a set in fixed bijection $a_n\leftrightarrow a$
with $A$, and $\epsilon_n\to0$ rational numbers, such that the
coefficients of $f_n$ are $\epsilon_n$-close to those of $f$, and the
spherical distances between pairs $a,a_n$ and between pairs
$f(x),f_n(x)$ are less than $\epsilon_n$ for all $a\in A$ and all
$x\in \widehat\C$.

Consider a portrait $(f\colon A\cup C\to A,\deg\colon A\cup C\to\N)$
of rational maps with hyperbolic orbifold; $C$ is the set of critical
points, and $A$ is a forward-invariant set. The set of rational
Thurston maps agreeing with $f$ on $A\cup C$ and having required
degree at $A\cup C$ is finite, and explicitly describable by a set of
equations with integer coefficients --- for example, with variables
the coefficients of rational map and the position of $A\cup C$ on
$\widehat\C$. It follows that, for every rational Thurston map with
hyperbolic orbifold, its coefficients are computable over $\C$ in the
sense above.

\begin{algo}\label{algo:rat2biset}
  \textsc{Given} a floating-point approximation of a rational Thurston map $f$, as well as its portrait,
  \textsc{Compute} the biset $B(f)$ \textsc{as follows:}\\\upshape
  \begin{enumerate}
  \item Find the post-critical set of $f$ on a sufficiently close approximation of $f$, using the given portrait.
  \item Compute the Delaunay triangulation on the post-critical set.
  \item Add sufficiently many points to the triangulation, keeping the
    Delaunay condition, so that $f$-lifts of triangles are small
    enough; e.g.\ don't touch more than one post-critical point.
  \item Compute the $f$-preimage of this triangulation.
  \item Apply Algorithm~\ref{algo:tri2biset}.
  \end{enumerate}
\end{algo}

\begin{algo}[\cite{bartholdi-buff-bothmer-kroeker:hurwitz}]\label{algo:hurwitz}
  \textsc{Given} an admissible $n$-tuple of permutations $\pi_1,\dots,\pi_n\in d\perm$ and an $n$-tuple of distinct points $z_1,\dots,z_n\in\hC$,\\
  \textsc{Compute} a floating-point approximation of a rational map
  whose critical values $z_i$ have monodromy $\pi_i$.
\end{algo}

\noindent We recall Thurston's fundamental ``annular obstruction''
theorem:
\begin{thm}[Thurston~\cite{douady-h:thurston}]\label{thm:thurston}
  Let $f\colon(S^2,P_f)\selfmap$ be a Thurston map with hyperbolic
  orbifold. Then $f$ is combinatorially equivalent to a rational map
  if and only if $f$ admits no annular obstruction, namely no
  invariant multicurve whose Thurston matrix has spectral radius
  $\ge1$, see~\eqref{eq:thurston matrix}. Furthermore, in that case
  the rational map is unique up to conjugation by M\"obius
  transformations.
\end{thm}

Given a sphere biset $\subscript G B_G$ with hyperbolic orbifold,
either an annular obstruction for $B$ or a rational map $f$ with
algebraic coefficients and $B(f)\sim B$ may be computed: as we noted
above, there are finitely many rational maps $f_1,\dots,f_N$ with
given degree and portrait, and they can be computed, e.g.\ by solving
algebraic equations over $\Z$. Their bisets may be computed using
Algorithm~\ref{algo:rat2biset}. We may then in parallel search through
the countably many multicurves in $G$, seeking an annular obstruction
for $B$, and the countably many bijections between $B$ and $B(f_i)$
for all $i\in\{1,\dots,N\}$, seeking a conjugacy. By
Theorem~\ref{thm:thurston}, one of these searches will eventually
succeed.

If one knows beforehand that at least one critical point of $f$ is
periodic, and one is only interested in knowing whether $f$ admits an
annular obstruction, then a more straightforward algorithm is
available: one may run in parallel a search for the annular
obstruction and for a realization of $f$ as a self-map on an ``elastic
graph'' with appropriate looseness; such a self-map is a certificate
for $f$ being isotopic to a rational map, by a recent result of Dylan
Thurston~\cite{thurston:rubber}.

The degree of the equations over $\Z$ to consider is so high that this
approach fails except in the most trivial cases. The following
algorithm does work in practice:
\begin{algo}\label{algo:thurston}
  \textsc{Given} a sphere biset $\subscript GB_G$ with hyperbolic orbifold,\\
  \textsc{Compute} either a rational map $f\in\C(z)$, with computable
  algebraic coefficients, such that $B(f)\sim B$, or a $B$-invariant
  annular obstruction \textsc{as follows:}\\\upshape
  \begin{enumerate}
  \item We consider \emph{Teichm\"uller space modeled on $G$} as the
    space of triangulations of $\hC$ with a distinguished $n$-tuple of
    vertices, as well as markings of its dual edges by elements of
    $G$. Two marked triangulations represent the same point in
    Teichm\"uller space if their distinguished vertices are images of
    each other under a M\"obius transformation, and after their
    distinguished vertices are identified they admit a common
    refinement, up to isotopy, compatible with the group markings.
  \item Start by an arbitrary point $S_0$ in Teichm\"uller space. We
    iterate the following steps, starting with $i=0$ and increasing it
    each time by $1$.
  \item Apply Algorithm~\ref{algo:hurwitz} to the permutations in $B$
    and the marked points in $S_i$, find a rational map $f_i$.\label{algo:thurston:3}
  \item Compute the biset $B(f_i)$ by
    Algorithm~\ref{algo:rat2biset}. Its right-acting group is
    $G$ and its left-acting group is $H_i$.\label{algo:thurston:4}
  \item Match the bisets $B$ and $B(f_i)$ by finding a group
    homomorphism $\phi_i\colon H_i\to G$ such that
    $B_{\phi_i}^\vee\otimes B(f_i)\cong B$. Note that $\phi_i$ is
    unique up to conjugation by inner automorphisms. The group $H_i$
    marks a triangulation of $\hC$ with distinguished vertices
    $\widetilde V_i$ in bijection with peripheral classes in
    $H_i$. Those peripheral classes in $H_i$ that map under $\phi_i$
    to non-trivial peripheral classes in $G$ determine a subset
    $V_i\subseteq \widetilde V_i$, and $\phi_i$ may be applied to the
    markings of the triangulation to produce a new triangulation,
    marked by $G$ and with distinguished vertices $V_i$. This defines
    a new element $S_{i+1}$ in Teichm\"uller space.\label{algo:thurston:5}
  \item Normalize the marking of $S_{i+1}$, viewed as
    $\{(x,y,z)\in \R^3\colon x^2+y^2+z^2=1\}$, so that the barycentre
    of its marked points is $(0,0,0)\in\R^3$, the first of them (in
    some predetermined order) is at $(0,0,1)$, and the second of them
    lies in $\{0\}\times\R_+\times\R$.
  \item By Thurston's Theorem~\ref{thm:thurston}, either the maps
    $f_i$ converge, and give approximations of the desired $f$, or
    points in $S_i$ cluster.

    To detect this, compute all cross-ratios of points in
    $S_i$. Partition the points in $S_i$ by clustering together points
    with degenerating cross-ratios. Produce conjugacy classes in $G$
    by following labels on dual edges in the triangulation around
    clusters. Complete these conjugacy classes by adding their
    $B_*$-preimages till the collection becomes either $B$-invariant
    or non-(non-crossing). If it became invariant, compute its
    transition matrix to check whether it is an annular
    obstruction. In that case, return the annular obstruction.

  \item Simultaneously, seek algebraic numbers of low degree, small
    height (= maximal absolute value of coefficients over $\Z$), and
    close to the estimated position of the set $S_i$ (now normalizing
    it in $\widehat\C$ by putting three of its points at $0,1,\infty$
    respectively), and obtain in this manner a probable position
    $\tilde S$ of $S_i$ consisting of algebraic numbers. Find the
    corresponding rational map $\tilde f$ with algebraic
    coefficients. Run
    Steps~\eqref{algo:thurston:4}--\eqref{algo:thurston:5} with
    $\tilde f,\tilde S$ \emph{en lieu} of $f_i,S_i$, checking at the
    same time that the permutations of $B(\tilde f)$ match those of
    $B$; this produces a lift $S_{i+1}$ of $\tilde S$. If $S_{i+1}$,
    suitably normalized as above, coincides with $\tilde S$, then
    return $\tilde f$, while if not increase $i$ and return to
    step~\eqref{algo:thurston:3}.\label{algo:thurston:AlgNumb}
  \end{enumerate}
\end{algo}
We note that the resort to algebraic numbers in
Step~\eqref{algo:thurston:AlgNumb} of Algorithm~\ref{algo:thurston} is
only necessary to prove its validity; in practice it may be replaced
by careful interval arithmetic; details are postponed
to~\cite{bartholdi-dudko:bc5}. If we are interested in good
floating-point approximations of the position of $\tilde S$, then we
may iterate Steps~\eqref{algo:thurston:4}--\eqref{algo:thurston:5},
without changing the combinatorics of the marked spheres and only
improving the position of the marked set $S_i$.

\begin{cor}\label{cor:conj:RatBis}
  There is an algorithm that, given two sphere bisets
  $\subscript G B_G$ and $\subscript H C_H$ of rational non-\Tor\
  maps, decides whether $B$ and $C$ are conjugate, and if so produces
  a conjugator.

  The centralizers of $\subscript G B_G$ and $\subscript H C_H$
  are trivial.
\end{cor}
\begin{proof}
  The rational maps may be constructed using
  Algorithm~\ref{algo:thurston}, and their coefficients compared once
  three points in each post-critical set are fixed.
\end{proof}

\subsection{The canonical decomposition of Levy free maps}
Recall the canonical decomposition of a Thurston map, introduced
in~\S\ref{ss:canonical decomposition}. We show that Levy-free maps
have quite restricted canonical decompositions:

\begin{lem}
  Let $f$ be a Levy free map. Either small maps in the canonical
  decomposition of $f$ are rational or $f$ is a \Tor\ map with trivial
  canonical obstruction.
\end{lem}
As a consequence, the canonical obstruction is the union of the
\emph{canonical Levy obstruction} (the minimal multicurve whose return maps
are Levy-free) and the \emph{rational obstruction}: for an expanding map
not doubly covered by a torus endomorphism, the minimal multicurve whose
return maps are rational (for \Tor\ maps, that multicurve is characterized
in~\cite{selinger:canonical}).

\begin{algo}\label{algo:canonical}
  \textsc{Given} a Levy-free non-\Tor\ sphere biset $\subscript GB_G$,\\
  \textsc{Compute} the canonical obstruction of $B$ \textsc{ as follows:}\\
  \begin{enumerate}
  \item Enumerate in increasing order all multicurves $\CC$ on the
    sphere marked by $G$;
  \item If $\CC$ is not fully $B$-invariant, discard it;
  \item Run Algorithm~\ref{algo:thurston} on the small maps
    $R(B,\CC)$. If all small maps are rational, then return $\CC$,
    otherwise discard it.
  \end{enumerate}
\end{algo}
\noindent Selinger and Yampolsky already showed
in~\cite{selinger-yampolsky:geometrization} that the canonical
obstruction of a Thurston map is computable.
\begin{cor}\label{cor:CanDecExp}
  There is an algorithm that, given a sphere biset $\subscript GB_G$ of type
  $\Exp\setminus\Tor$, computes the canonical decomposition
  $\subscript \gf\gfB_\gf$. All bisets in $R(\gfB)$ are the
  bisets of rational maps.
\end{cor}

\begin{algo}\label{algo:DecidExp}
  \textsc{Given} two contracting sphere bisets $\subscript {G}B_{G}$ and  $\subscript {H}C_{H}$ that are not \Tor,\\
  \textsc{Decide} whether $B$ and $C$ are conjugate, and \textsc{compute} the centralizer $Z(B)$ \textsc{as
    follows:}\\\upshape
 \begin{enumerate}
 \item If there are more peripheral conjugacy classes in $B$,
   respectively $C$, than in their post-critical set, then these
   conjugacy classes can be expressed into a portrait of bisets, as
   in~\S\ref{ss:portraits}. Conjugacy of portraits of bisets is
   decidable, since $B$ and $C$ are contracting, see
   Proposition~\ref{prop:PortrBis:ContrCase}. We therefore assume that
   $B,C$ are marked by their post-critical conjugacy classes.
 \item Using Corollary~\ref{cor:CanDecExp}, compute the canonical 
   decompositions $\subscript \gf\gfB_\gf$ and
   $\subscript \gfY\gfC_\gfY$ of $B$ and $C$
   respectively.
 \item Let $X$ and $Y$ be the set of distinguished conjugacy classes
   of $\gf$ and $\gfY$ respectively, see~\S\ref{ss:distinguished
     cc}. Enumerate all possible bijections $h\colon X\to Y$.
 \item For every $h\colon X\to Y$ try to do the following
   steps. Return \texttt{fail} if there is no success.
 \item Using Algorithm~\ref{algo:MarkCl:Prom}, try to promote
   $h\colon X\to Y$ into a biprincipal sphere $\gf$-$\gfY$-tree of
   bisets $\gfI$. Discard $h$ if there is no promotion.
 \item Check by Algorithm~\ref{algo:chech:InMCB} whether
   $\gfC^\gfI \in M(\gfB)$; if not, discard
   $h$.
 \item For every $\gfI$ check, using
   Corollary~\ref{cor:conj:RatBis}, whether bisets in
   $R(\gfB)$ and $R(\gfC^\gfI)$ are conjugate;
   if not discard $h$.
 \item For every $\gfI$ check, using Theorem~\ref{thm:extension},
   whether the conjugacies between $R(\gfB)$ and $R(\gfC^\gfI)$
   promote into a conjugacy between $\gfB$ and $\gfC$.
 \item Using Theorem~\ref{thm:extension}, compute the centralizer
   $Z(B)$ of $B$, which is computable because by
   Corollary~\ref{cor:conj:RatBis} all bisets in $R(\gfB)$
   have trivial centralizers.
 \end{enumerate}
\end{algo}

\begin{cor}\label{cor:expdecidable}
  Let $B,C$ be contracting sphere bisets that are not \Tor. Then it is
  decidable whether $B,C$ are conjugate, and the centralizer of $B$ is
  computable.\qed
\end{cor}

\renewcommand\thesection{\arabic{section}}
\section{Decidability of combinatorial equivalence}\label{ss:decidability}
In this brief section we put together results from~\S\S I--V to prove
Theorem~\ref{thm:decidable}: ``it is decidable whether two Thurston
maps are combinatorially equivalent''.

\begin{thm}\label{thm:conj:LevyFree}
  Let $B,C$ be two sphere bisets. Assume that $B$ and $C$ are either
  degree-$1$ or Levy-free.

  Then it is decidable whether $B$ and $C$ are conjugate, and the
  centralizer of $B$ is computable as a finitely generated subgroup of
  a product of pure mapping class groups.
\end{thm}
\begin{proof}
  It is decidable whether $B$ and $C$ have the same degree, and have
  isomorphic acting groups; if not, they are not conjugate.

  If $B,C$ have degree $1$, then they may be written as
  $B_\phi,B_\psi$ respectively; then $B\sim C$ if and only if $\phi$
  and $\psi$ are conjugate in a pure mapping class group; this is
  decidable by~\cite{hemion:homeos}. The centralizer of a mapping
  class is also computable; for example, train
  tracks~\cite{bestvina-h:tt} can be used to compute the centralizer
  of a pseudo-Anosov mapping class.

  If $B,C$ have degree $>1$ and their orbisphere quotient groups are
  both $(2,2,2,2)$-orbisphere groups, then their conjugacy and
  centralizer problems are solvable by Corollary~\ref{cor:tordecidable}.

  In the remaining case, $B,C$ are contracting by
  Theorem~\ref{thm:ExpCr}, so their conjugacy and centralizer problems
  are solvable by Corollary~\ref{cor:expdecidable}.
\end{proof}
We note that the condition that the bisets be Levy-free is essential:
in general, the centralizer need not be finitely generated, see
Example~\ref{ex:infinitely generated}.

\noindent We are ready to prove Theorem~\ref{thm:decidable}. Its proof
consists of the following
\begin{algo}\label{algo:Decid}
  \textsc{Given} two sphere bisets $\subscript G B_G$ and  $\subscript H C_H$,\\
  \textsc{Decide} whether $B$ and $C$ are conjugate \textsc{as follows:}\\\upshape
  \begin{enumerate}
  \item Using Algorithm~\ref{algo:levy}, compute the Levy
    decompositions $\subscript \gf\gfB_\gf$ and
    $\subscript \gfY\gfC_\gfY$ of $B$ and $C$
    respectively.
  \item Let $X$ and $Y$ be the set of distinguished conjugacy classes
    of $\gf$ and $\gfY$ respectively, see~\S\ref{ss:distinguished
      cc}. Enumerate all possible bijections $h\colon X\to Y$.
  \item For every $h\colon X\to Y$ try to do the following
    steps. Return \texttt{fail} if there is no success.
  \item Using Algorithm~\ref{algo:MarkCl:Prom}, try to promote
    $h\colon X\to Y$ into a biprincipal sphere $\gf$-$\gfY$-tree of
    bisets $\gfI$. Discard $h$ if there is no promotion.
  \item Using Algorithm~\ref{algo:chech:InMCB}, check whether
    $\gfC^\gfI \in M(\gfB)$; if not, discard
    $\gfI$.
  \item For every $\gfI$ check, using
    Theorem~\ref{thm:conj:LevyFree}, whether bisets in $R(\gfB)$
    and $R(\gfC^\gfI)$ are conjugate; if not discard
    $\gfI$.
  \item Using Theorem~\ref{thm:conj:LevyFree}, compute the centralizer
    of bisets in $R(\gfB)$.
  \item For every $\gfI$ check, using Theorem~\ref{thm:extension},
    whether the conjugacies between $R(\gfB)$ and $R(\gfC^\gfI)$
    promote into a conjugacy between $\gfB$ and $\gfC$.
  \item Using Theorem~\ref{thm:extension}, compute the centralizer
    $Z(B)$ of $B$.
  \end{enumerate}
\end{algo}

\section{Algebraic realizations}\label{ss:examples}
The previous sections explain how a Thurston map can canonically, and
computably, be decomposed into pseudo-Anosov and finite-order
homeomorphisms, rational maps of degree $\ge2$ and maps doubly covered
by torus endomorphisms.

Finite-order homeomorphisms are isotopic to M\"obius transformations,
which are rational maps. If $f\colon(S^2,A)\selfmap$ has only
rational maps as pieces, then each of its small spheres may be given a
complex structure so as to make all pieces rational. The structure
encoding $f$ is then a \emph{map of noded spheres}.

A \emph{complex stable curve} is an algebraic variety with the
topology of a noded sphere. It may be given as
$X=(\hC_1,P_1)\sqcup\cdots\sqcup(\hC_n,P_n)/{\sim}$ with $\#P_i\ge3$ for
all $i$ and an equivalence relation $\sim$ on $P_1\sqcup P_2\sqcup \dots \sqcup P_k$ with
classes of size $\le2$, such that the space obtained by replacing each
$\hC\cong S^2$ by a closed ball is contractible. It is thus a
collection of $\hC=\mathbb P^1(\C)$'s glued to each other at single
points, in a tree-like manner. The points in non-trivial
$\sim$-classes are called \emph{nodes}. Listing the non-trivial
equivalence classes as $p_k\sim q_k$ for $i=1,\dots,\ell$, with
$p_k\in P_{i(k)}$ and $q_k\in P_{j(k)}$, one may describe $X$ as an
algebraic singular curve
\[\mathbb X=\{(x_1,\dots,x_n)\in\hC^n\mid (x_{i(k)}-p_k)(x_{j(k)}-q_k)=0\text{
  for }k=1,\dots,\ell\}.
\]

If $\CC$ is a multicurve on a topological sphere $(S^2,A)$, then
shrinking all components of $\CC$ to points gives a noded topological
sphere. Conversely, given a complex stable curve, each node may be
``opened'', namely replaced by a small cylinder, so as to give a
topological sphere and a multicurve.

Let $f\colon(S^2,A)\selfmap$ be a Thurston map, and let $\CC$
be an $f$-invariant multicurve. Let $X$ be the corresponding noded
topological sphere, and let $Y$ be the noded topological sphere
corresponding to $(S^2,f^{-1}(A),f^{-1}(\CC))$. Then $f$ induces a
branched covering still written $f\colon Y\to X$, while the inclusion
$f^{-1}(A)\supseteq A$ induces a continuous map $i\colon Y\to X$. We
thus have a topological self-correspondence
$f,i\colon Y\rightrightarrows X$. Topologically, $i$ is a
\emph{blow-down}: it shrinks some spheres to points, and erases some
marked points.

Note that the correspondence $f,i\colon Y\rightrightarrows X$ does not
quite determine $f\colon(S^2,A)\selfmap$: once nodes of $X$
are opened to a multicurve $\CC$ on $(S^2,A)$, the map $f$ is defined
on $(S^2,A)$ only up to an integer number of full twists along $\CC$.

If $X,Y$ may be given structures of complex stable curves
$\mathbb X,\mathbb Y$ so that $f,i$ become rational maps
$\mathbb Y\to\mathbb X$, we say $f$ is \emph{realized} by
$f,i\colon\mathbb Y\to\mathbb X$.

A periodic cycle of curves is called an \emph{annular obstruction} if
the spectral radius of its Thurston matrix~\eqref{eq:thurston matrix}
is $\ge1$; and a Thurston map is called \emph{obstructed} if it admits
an annular obstruction.  We show that many examples of Thurston maps,
even if they are obstructed and therefore not combinatorially
equivalent to a rational map (see Theorem~\ref{thm:thurston}), may be
described algebraically, using the following result:
\begin{mainthm}\label{thm:G}
  Let $f\colon(S^2,A)\selfmap$ be a Thurston map. Then $f$ may be
  realized on a complex stable curve by pinching along a multicurve
  generated by annular obstructions if and only if all pieces of
  $f$'s canonical decomposition are rational, non-irrational \Tor\
  maps and homeomorphisms not containing pseudo-Anosov maps.

  Furthermore, it is computable whether all pieces of $f$ are this
  form, and in that case what the correspondence of complex stable
  curves is.
\end{mainthm}
In case all return maps are rational with hyperbolic orbifold, it
follows from~\cite{selinger:augts}*{Theorem~10.4} that such a
realization exists, given by a fixed point in augmented Teichm\"uller
space.

\begin{proof}[Proof of Theorem~\ref{thm:G}]
  If there are irrational maps doubly covered by torus endomorphisms,
  or homeomorphisms containing pseudo-Anosov maps, in the canonical
  decomposition, then these maps appear in every decomposition along
  an annular obstruction. They preserve transverse measured foliations
  with different stretch factors; this prevents them from being
  rational. Furthermore, since the eigenvalues are irrational, these
  pieces cannot be further decomposed.

  On the other hand, given a decomposition in which the return maps
  are rational, non-irrational \Tor\ maps and homeomorphisms not
  containing pseudo-Anosov maps, it is straightforward to assemble
  them together in a complex stable curve. Maps doubly covered by
  torus endomorphisms with integer eigenvalues decompose into
  Chebyshev polynomials, see~\S\ref{ex:endomn}. Homeomorphisms with no
  pseudo-Anosov piece decompose along an annular obstruction into
  finite-order homeomorphisms, which are isotopic to M\"obius
  transformations.
\end{proof}

We conclude this section, and this article, with some fundamental
examples illustrating decompositions, and of possible realizations on
complex stable curves:
\begin{description}
\item[\S\ref{ex:z2pi}] various maps obtained from composing the
  polynomial $z^2+i$ with Dehn twists;
\item[\S\ref{ex:z2m1}] the mating of $z^2-1$ with itself: it is the
  simplest example of an obstructed map;
\item[\S\ref{ex:dh}] Douady-Hubbard's mating of the quadratic
  lamination $5/12$ with itself. This map has intersecting annular
  obstructions, and we show which maps are obtained by cutting along
  the different multicurves;
\item[\S\ref{ex:endomn}] maps doubly covered by diagonalizable torus
  endomorphisms;
\item[\S\ref{ex:pilgrim}] A degree-$5$ rational map considered by
  Pilgrim in~\cite{pilgrim:combinations}*{\S1.3.4}. It is obtained by
  blowing up an arc on a torus endomorphism. (We solve the hitherto-open
  problem of determining whether this map is obstructed; it is so);
\item[\S\ref{ex:tls}] a degree-$3$ mating studied by Tan Lei and
  Shishikura~\cite{shishikura-t:matings}, which is obstructed but does
  not admit Levy cycles. (We describe the mating both as a biset and
  as an algebraic self-map of a doubly-noded sphere, and
  answer~\cite{cheritat:shishikura}*{Question~2} by Ch\'eritat);
\item[\S\ref{ex:infinitely generated}] a degree-$4$ Thurston map whose
  centralizer is infinitely generated.
\end{description}

\subsection{\myboldmath $(2,3,6)$-maps}
We give a brief description of rational maps multiply covered by a torus
endomorphism in terms of their orbispace. Orbispaces of Euler
characteristic $0$ have marked points of respective orders
$(2,2,2,2)$, $(3,3,3)$, $(2,4,4)$ or $(2,3,6)$. They may be treated
uniformly as follows: let $\zeta$ be a $k$th root of unity, for
$k\in\{2,3,4,6\}$. If $k\ge3$, let $\Lambda\coloneqq\Z[\zeta]$ be the
lattice spanned by $\zeta$ in $\C$; while if $k=2$ choose at will
$\Lambda=\Z[i]$ or $\Lambda=\Z[(-1+\sqrt{-3})/2]$. Set then
\[G\coloneqq\Lambda\rtimes\langle\zeta\rangle.\] The orbispace is
$\C/G$, with marked points of orders $(2,2,2,2)$, $(3,3,3)$,
$(2,4,4)$, $(2,3,6)$ for $k=2,3,4,6$ respectively. Graphically,
\begin{center}
  \begin{tikzpicture}
    \def\drawedge#1{\draw[decoration={markings, mark=at position 0.5 with {\arrow{#1}}},postaction={decorate}]}
    \begin{scope}[xshift=0cm,yshift=0cm,scale=1.5]
      \node at (-0.5,0) {$k=2$};
      \drawedge{>} (-1,-0.5) node [below left] {\small $2$} -- +(1,0) node [below] {\small $2$};;
      \drawedge{>} (1,-0.5) -- +(-1,0);
      \drawedge{>>} (-1,0.5) node [above left] {\small $2$} -- +(1,0) node [above] {\small $2$};
      \drawedge{>>} (1,0.5) -- +(-1,0);
      \drawedge{>>>} (-1,-0.5) -- +(0,1);
      \drawedge{>>>} (1,-0.5) -- +(0,1);
      \draw[dotted] (0,-0.5) -- +(0,1);
    \end{scope}
    \begin{scope}[xshift=3.5cm,yshift=-1.8cm,scale=1.5]
      \node at (0,0.666) {$k=2$};
      \drawedge{>} (0,0) node [below] {\small $2$} -- +(-1,0) node [below] {\small $2$};
      \drawedge{>} (0,0) -- +(1,0);
      \drawedge{>>} (-1,0) ++(60:1) node [left] {\small $2$} -- +(240:1);
      \drawedge{>>} (-1,0) ++(60:1) -- +(60:1);
      \drawedge{>>>} (1,0) ++(120:1) node [right] {\small $2$} -- +(300:1);
      \drawedge{>>>} (1,0) ++(120:1) -- +(120:1);
      \draw[dotted] (-1,0) ++(60:1) -- ++(1,0) -- ++(240:1) -- cycle;
    \end{scope}
    \begin{scope}[xshift=-1cm,yshift=-3cm,scale=1.5]
      \node at (0.35,0.8) {$k=4$};
      \drawedge{>} (0,0) node [below left] {\small $4$} -- +(1,0) node [below right] {\small $2$};;
      \drawedge{>} (0,0) -- +(0,1);
      \drawedge{>>} (1,0) -- +(0,1) node [above] {\small $4$};
      \drawedge{>>} (0,1) -- +(1,0);
      \draw[dotted] (0,0) -- +(1,1);
    \end{scope}
    \begin{scope}[xshift=6cm,yshift=-3cm,scale=1.5]
      \node at (-0.1,0.2) {$k=6$};
      \drawedge{>} (0,0) node [below] {\small $2$} -- +(-1,0) node [below] {\small $6$};
      \drawedge{>} (0,0) -- +(1,0);
      \drawedge{>>} (0,0.666) node [above] {\small $3$} -- (-1,0);
      \drawedge{>>} (0,0.666) -- (1,0);
      \draw[dotted] (0,0) -- (0,0.666);
    \end{scope}
    \begin{scope}[xshift=7cm,yshift=-0.4cm,scale=1.5]
      \node at (-0.1,0.2) {$k=3$};
      \drawedge{>} (0,-0.666) node [below] {\small $3$} -- (-1,0) node [left] {\small $3$};
      \drawedge{>} (0,-0.666) -- (1,0);
      \drawedge{>>} (0,0.666) node [above] {\small $3$} -- (-1,0);
      \drawedge{>>} (0,0.666) -- (1,0);
      \draw[dotted] (0,-0.666) -- (0,0.666);
    \end{scope}
  \end{tikzpicture}
\end{center}

Endomorphisms of the orbisphere $\C/G$ are quotients $\overline f$ of
maps $f\colon\C\selfmap$ with $f\circ G\subseteq G\circ f$. It
immediately follows that $f$ is isotopic to an affine map, say
$f(z)=a z+b$ with $a,b\in\C$. Then
$f\circ(z\mapsto z+1)=(z\mapsto z+a)\circ f$, so $a\in\Lambda$; and
$f\circ(z\mapsto\zeta z)=(z\mapsto\zeta z+(1-\zeta)b)\circ f$, so
$b\in(1-\zeta)^{-1}\Lambda$. Since $z\mapsto z+1$ and
$z\mapsto\zeta z$ generate $G$, these conditions are also
sufficient. Therefore, every rational map on $\C/G$ is determined by
parameters $a\in\Lambda$ and $b\in\frac1{1-\zeta}\Lambda/\Lambda$.

The biset of $\overline f$ is then easy to compute: since $f$ is
invertible, there exists, for all $g\in G$, a unique element, written
$g^f\in G$, such that $g f=f g^f$. We obtain:
\begin{prop}
  The biset $B(\overline f)$ of $f(z)=a z+b$ is $G$ as a right $G$-set,
  with left action given by $g\cdot b=g^f b$.
\end{prop}

Note also $g^f\in\Lambda$ for all $g\in\Lambda$, so
$\Lambda\subseteq B(\overline f)$ is a $\Lambda$-subbiset. Its
structure depends only on the linear part $a$ of $f$. We may then
write $B(\overline f)$ as a crossed product biset
$\Lambda\rtimes\langle\zeta\rangle$ of this $\Lambda$-subbiset with
the group $\langle\zeta\rangle$ acting by automorphisms on
$\Lambda$. This extends the discussion in~\S\ref{ss:lattes}, and in
particular~\eqref{eq:InjEndOfK}.

\subsection{A Dehn twist}\label{ss:ex:DehnTwist}
Before studying the twisted cousins of the polynomial $z^2+i$
in~\S\ref{ex:z2pi}, we consider the simple example of a Dehn twist on
a sphere with four punctures. Consider a basepoint $*$, and arrange
the punctures in counterclockwise order around $*$. Let $a,b,c,d$
denote the ``lollipop'' generators about the punctures, giving rise to
the sphere group
\[G=\langle a,b,c,d\mid d c b a\rangle.
\]
Set $r=c b$, and let $T$ denote the Dehn twist about the curve
$r^G$. Note that $T$ fixes $*$.  The action of $T$ on $G$ is given by
\[a\mapsto a,\quad b\mapsto b^r,\quad c\mapsto c^r,\quad d\mapsto d,\]
so the biset of $T$ has in basis $\{*\}$ the wreath recursion
\[a=\pair{a},\quad b=\pair{b^r},\quad c=\pair{c^r},\quad d=\pair{d}.\]

We have $G=G_1*_\Z G_2$, with $G_1=\langle a,r,d\mid d r a\rangle$ and
$G_2=\langle b,c,r^{-1}\mid b c r^{-1}\rangle$, glued along
$\langle r\rangle\cong\Z$. The tree of groups decomposition of $G$
therefore consists of a single segment. The decomposition $\mathfrak T$ of $B(T)$ as
a tree of bisets is also a single segment, with $\rho=\lambda=\one$
and vertex and edge bisets $B_1,B_2,E_1$ above $G_1,G_2,\langle r\rangle$
respectively:
\[\begin{tikzpicture}[baseline=0.75cm]
    \node[inner sep=0pt] (B1) at (0,1.5) [label={left:$B_1$}] {$\bullet$};
    \node[inner sep=0pt] (B2) at (4,1.5) [label={right:$B_2$}] {$\bullet$};
    \draw[thick]  (B1.center) -- node[below] {$E_1$} (B2.center);

    \node[inner sep=0pt] (G1) at (0,0) [label={left:$G_1$}] {$\bullet$};
    \node[inner sep=0pt] (G2) at (4,0) [label={right:$G_2$.}] {$\bullet$};
    \draw[thick] (G1.center) -- node[above] {$\langle r\rangle$} (G2.center);

    \draw (B1) edge[rho] node[right] {$\rho$} (G1);
    \draw (B1) edge[lambda,bend right] node[left] {$\lambda$} (G1);
    \draw (B2) edge[rho] node[right,pos=0.75] {$\rho$} (G2);
    \draw (B2) edge[lambda,bend right] node[left] {$\lambda$} (G2);
  \end{tikzpicture}
\]

These bisets are obtained as subbisets of $B(T)$ by restricting the
wreath recursion of $B(T)$ to the subgroups $G_1$, $G_2$ and
$\langle r\rangle$ respectively. If we use for them the same basis
$\{*\}$, they are given as follows:
\begin{itemize}
\item the $G_1$-$G_1$-biset $B_1$ is the biset of the identity
  \[a=\pair{a},\quad r=\pair{r},\quad d=\pair{d},
  \]
  because the basepoint $*$ belongs to the sphere marked by $a,d,r$;
\item the $G_2$-$G_2$-biset $B_2$ has wreath recursion
  \[b=\pair{b^r},\quad c=\pair{c^r},\quad r=\pair r;\]
\item the biset $E_1$ is the identity $\Z$-$\Z$-biset. It embeds
  naturally in its respective source and target vertex bisets under
  $g\cdot*\mapsto g\cdot *$ for all $g\in\langle r\rangle$.
\end{itemize}
Note that, if one changes $B_2$'s basis to $\{*'\coloneqq r\cdot *\}$,
one obtains for $B_2$ the wreath recursion of the identity map; but
then the embeddings of the edge $\Z$-$\Z$-biset $E_1$ in $B_1$ and
$B_2$ change: the basis element $*$ of $E_1$ maps to the basis element
$*$ of $B_1$, but to $r^{-1}\times$ the basis element $*'$ of $B_2$.

Note also that the biset of $T$ is biprincipal and that $\mathfrak T$
is a biprincipal tree of bisets. The biset of $T^n$ is
$B(T)^{\otimes n}$, and its decomposition $\mathfrak T^{\otimes n}$
also consists of a single edge, with identity bisets at vertices and
the edge, and embeddings $b\mapsto b$ and $b\mapsto r^{-n}\cdot b$ of
the edge biset into the vertex bisets $B_1$ and $B_2$ respectively.

This example illustrates the importance of the edge biset
embeddings. For all $\mathfrak T^{\otimes n}$ the trees of groups are
the same, the trees of bisets have the same underlying tree and the
same vertex and edge bisets; the algebraic realizations are the same
(a noded sphere with identity self-map), but the embeddings of the edge
biset in the vertex bisets depend on the twist parameter $n$.

\subsection{\myboldmath All the twisted cousins of $z^2+i$}\label{ex:z2pi}
The polynomial $z^2+i$ has the following critical graph:
$0\Rightarrow i\to i-1\leftrightarrow -i$. There exist maps with the
same post-critical graph as $z^2+i$ that are obstructed, and one such
map, $f$ may be constructed as
follows~\cite{bartholdi-n:thurston}*{\S6.1}:
\begin{center}
  \begin{tikzpicture}
  \begin{scope}[xshift=-5cm]
    \node (t0) at (-10:2) {$\circ$};
	\node (t1) at (170:2) {$\circ$};
	\node[black!30] (b0) at (-0.9,0.6) {$\bullet$};
	\node (c0) at (0.9,0.6) {$\bullet$};
	\node (b1) at (0.9,-0.6) {$\bullet$};
	\node (c1) at (-0.9,-0.6) {$\bullet$};
	\node (z) at (0,0) {$\bullet$};
	\node[anchor=west] at ($(0.1,-0.05)$) {\scriptsize $f^{-1}(x)$};
	\node[black!30,anchor=west] at ($(b0)+(0.0,-0.1)$) {\scriptsize $f^{-1}(z)$};
	\node[anchor=east] at ($(c0)+(-0.05,0.1)$) {\scriptsize $f^{-1}(y)$};	
	\node[anchor=east] at ($(b1)+(-0.05,0.1)$) {\scriptsize $f^{-1}(z)$};
	\node[anchor=west] at ($(c1)+(0.0,-0.1)$) {\scriptsize $f^{-1}(y)$};
	\draw (0:2) .. controls +(90:1) and ($(180:2)+(270:1)$) .. (180:2) node [pos=0.7,above] {$a$};
	\draw[black!30] (t0) arc [start angle=-10,end angle=110,radius=2] .. controls +(200:1) and (b0) .. ($(b0)+(130:0.15)$) arc [start angle=130,end angle=490,radius=0.15] \fwdarrowonline{0.5} node [above=3pt] {$c$};
	\draw (t0) arc [start angle=-10,end angle=20,radius=2] (20:2) .. controls +(110:1) and (c0) .. ($(c0)+(60:0.15)$) arc [start angle=60,end angle=420,radius=0.15] \fwdarrowonline{0.5} node [above=1pt] {$b$};
	\draw (200:2) .. controls +(290:1) and (c1) .. ($(c1)+(240:0.15)$) arc [start angle=-120,end angle=240,radius=0.15] \bckarrowonline{0.5} node [below=1pt] {$b$};
    \draw (t1) arc [start angle=170,end angle=290,radius=2] .. controls +(20:1) and (b1) .. ($(b1)+(-50:0.15)$) arc [start angle=-50,end angle=310,radius=0.15] \bckarrowonline{0.5} node [below=3pt] {$c$};
	\draw[dotted,thick] (0,-0.6) ellipse (15mm and 3.2mm);
	\draw[dotted,thick] (0,0.6) ellipse (15mm and 3.2mm);
	\node[inner sep=-1pt,fill=white] at (t0) {$\circ$};
	\node[inner sep=-1pt,fill=white] at (t1) {$\circ$};
	\draw[fill=white] (z) circle (1.5mm) node[black!30] {$\bullet$};
	\end{scope}
	
	\begin{scope}
    \node (t) at (-10:2) {$\circ$};
	\node (a) at (0.9,0.6) {$\bullet$};
	\node (b) at (0.9,-0.6) {$\bullet$};
	\node (c) at (-0.9,-0.6) {$\bullet$};
	\node at ($(a)+(-0.35,0)$) {$x$};
	\node at ($(b)+(0,0.4)$) {$y$};
	\node at ($(c)+(-0.35,0)$) {$z$};	
	\draw (40:2) .. controls +(130:1) and (a) .. (a.north) arc [start angle=90,end angle=450,radius=0.2] \fwdarrowonline{0.5} node [above right=-1pt] {$a$};
	\draw (190:2) .. controls +(280:1) and (-1.5,0.6) .. ($(b)+(160:0.2)$) arc [start angle=-200,end angle=160,radius=0.2] \fwdarrowonline{0.5} node [below left=-1pt] {$b$};
    \draw (t) arc [start angle=-10,end angle=290,radius=2] .. controls +(20:1) and (c) .. ($(c)+(-20:0.2)$) arc [start angle=-20,end angle=340,radius=0.2] \bckarrowonline{0.5} node [below=1pt] {$c$};
	\draw[dotted,thick] (0,-0.6) ellipse (15mm and 8mm);
    \node at (1.4,-1.3) {$\Gamma$};
	\node[inner sep=-1pt,fill=white] at (t) {$\circ$};
	\end{scope}
	\draw[->,thick] (-4,2) .. controls (-3,2.5) and (-2,2.5) .. (-1,2) node [pos=0.5,below] {$f$};
\end{tikzpicture}
\end{center}
The post-critical set of $f$ is $P_f\coloneqq\{\infty,x,y,z\}$ with
post-critical graph $\cdot\Rightarrow x\to y\leftrightarrow z$. Let
$\Gamma$ denote the simple closed curve $\Gamma$ encircling $y$ and
$z$.

We first describe the biset $B(f)$. For this, we choose a basepoint
close to $\infty$ on the positive real axis (indicated by a white dot on
the figure above), and we consider the simple (``lollipop'') loops
$a,b,c,d$ around $x,y,z,\infty$ respectively. This gives the sphere
group (see~\S\ref{ss:spheregroups})
\[G=\langle a,b,c,d\mid d c b a\rangle.\] The wreath recursion
(see~\S\ref{ss:coverings}) of $B(f)$ may be computed as follows. By
our choice of basepoint $*$, one preimage $*_1$ is close to $+\infty$
and the other one $*_2$ is close to $-\infty$. We choose as basis of
$B(f)$ the set $Q=\{\ell_1,\ell_2\}$ with $\ell_1$ a very short path
from $*_1$ to $*$ and $\ell_2$ a half-turn in the upper half-plane
from $*_2$ to $*$. Then, tracing $f$-lifts of the lollipop generators,
we obtain a presentation of $B(f)$ as
\[a=\pair{a^{-1},a}(1,2),\quad b=\pair{a,c},\quad c=\pair{1,c b c^{-1}},\quad d=\pair{d,1}(1,2).
\]

One checks easily that $r=c b$, representing the conjugacy class
$\Gamma$, is an annular obstruction, and furthermore is a Levy cycle:
indeed $r=\pair{a,r}$, so its transition matrix is $(1)$.

Let $T$ denote the Dehn twist about $\Gamma$ as
in~\ref{ss:ex:DehnTwist}. We also consider all the maps
$f_n=T^n\circ f$; they are also obstructed, and are all
combinatorially inequivalent
(see~\cite{pilgrim:combinations}*{Theorem~8.2}
or~\cite{bartholdi-n:thurston}*{Proposition~6.10}). The action of $T$
on $G$ is given by
\[a\mapsto a,\quad b\mapsto b^r,\quad c\mapsto c^r,\quad d\mapsto d\]
so the biset of $f_n$ has wreath recursion
\begin{equation}\label{eq:ftwistn}
  a=\pair{a^{-1},a}(1,2),\quad b=\pair{a,c^{r^n}},\quad c=\pair{1,b^{r^{n-1}}},\quad d=\pair{d,1}(1,2).
\end{equation}

Another construction of $f_n$ may be given as follows. Start with the
biset $B(z^2+i)$; it can be computed by drawing paths in
$\C\setminus\{i,i-1,-i\}$ and lifting them by $\sqrt{z-i}$, but can
also be obtained by Algorithm~\ref{algo:angle2biset} starting from the
external angle $1/6$ of the map $z^2+i$. The wreath recursion of
$B(z^2+i)$ is
\[a=\pair{dc,b a}(1,2),\quad b=\pair{a,c},\quad c=\pair{b,1},\quad
d=\pair{d,1}(1,2).\]
Consider next the Dehn twist $U$ about the simple closed curve
encircling $i$ and $i-1$; its action on $G$ is
\[a\mapsto a^{b a},\quad b\mapsto b^a,\quad c\mapsto c,\quad d\mapsto d.
\]
Then $f_1$ is combinatorially equivalent (see~\S\ref{ss:spheregroups}) to
$(z^2+i)\circ U^{-1}$, so $B(f_1)\cong B(U)^\vee\otimes B(z^2+i)$, and
indeed the wreath recursion of $B(U)^\vee\otimes B(z^2+i)$ is
\[a=\pair{c b,ad}(1,2),\quad b=\pair{a^d,c^b},\quad c=\pair{1,b},\quad
d=\pair{1,d}(1,2),\]
which in basis $\{\ell_2,d^{-1}\ell_1\}$ coincides
with~\eqref{eq:ftwistn} for $n=1$.

We have $G=G_1*_\Z G_2$, with $G_1=\langle a,r,d\mid d r a\rangle$ and
$G_2=\langle b,c,r^{-1}\mid b c r^{-1}\rangle$, glued along $\langle
r\rangle\cong\Z$. The tree of groups therefore consists of a single
segment.

Since no small sphere in $(S^2,f_n^{-1}(P_f),f_n^{-1}(\Gamma))$ maps
to an annulus, we do not need to subdivide the tree of groups
barycentrically. Thus the decomposition $\gfB(f_n)$ of $B(f_n)$ as a
tree of bisets has three vertex bisets, corresponding to the three
components $B_1,B_2,B_3$ of $S^2\setminus f^{-1}(\Gamma)$. They are
arranged as follows:
\begin{equation}\label{eq:z2pi:gog}
  \begin{tikzpicture}[baseline=0.75cm]
    \node[inner sep=0pt] (B3) at (4,1.9) [label={right:$B_3$}] {$\bullet$};
    \node[inner sep=0pt] (B1) at (0,1.5) [label={left:$B_1$}] {$\bullet$};
    \node[inner sep=0pt] (B2) at (4,1.1) [label={right:$B_2$}] {$\bullet$};
    \draw[thick] (B3.center) -- node[above] {$E_3$} (B1.center) -- node[below] {$E_2$} (B2.center);

    \node[inner sep=0pt] (G1) at (0,0) [label={left:$G_1$}] {$\bullet$};
    \node[inner sep=0pt] (G2) at (4,0) [label={right:$G_2$.}] {$\bullet$};
    \draw[thick] (G1.center) -- node[above] {$\langle r\rangle$} (G2.center);

    \draw (B1) edge[rho] node[right] {$\rho$} (G1);
    \draw (B1) edge[lambda,bend right] node[left] {$\lambda$} (G1);
    \draw (B3) edge[rho] node[right,pos=0.75] {$\rho$} (G2);
    \draw (B2) edge[lambda,bend right] node[left] {$\lambda$} (G2);
    \draw (B3) edge[lambda,bend right=13] node[above,pos=0.75] {$\lambda$} (G1);
  \end{tikzpicture}
\end{equation}
By convention, the covering map $\rho$ is given by vertical projection
and drawn in plain lines, while the map $\lambda$ is drawn in squiggly
lines; it sends $B_1$, $E_3$ and $B_3$ to $G_1$ while sending $E_2$ to
$\langle r\rangle$ and $B_2$ to $G_2$.

The vertex bisets are obtained as subbisets of $B(f_n)$ by
restricting the wreath recursion of $B(f_n)$ to the subgroups $G_1$
and $G_2$, using subsets of the basis $Q$, and are given as follows:
\begin{itemize}
\item the $G_1$-$G_1$-biset $B_1$ has in the basis $Q$ the wreath
  recursion
  \[a=\pair{a^{-1},a}(1,2),\quad r=\pair{a,r},\quad d=\pair{d,1}(1,2),\]
  and is isomorphic to the biset of the rational map $z^2-2$;
\item the $G_2$-$G_2$-biset $B_2$ has in the basis $\{\ell_2\}$ the
  wreath recursion
  \[b=\pair{c^{r^n}},\quad c=\pair{b^{r^{n-1}}},\quad r=\pair r,\] and
  is isomorphic to the biset of the rational map $z^{-1}$ marked at
  $\{0,1,\infty\}$;
\item the $G_1$-$G_2$-biset $B_3$ has in the basis $\{\ell_1\}$ the
  wreath recursion
  \[b=\pair a,\quad c=\pair1,\quad r=\pair a;\]
\item the bisets $E_2$ is the identity $\Z$-$\Z$-biset in the basis
  $\{\ell_1\}$, and the biset $E_3$ is the $G_1$-$\Z$-biset given in the
  basis $\{\ell_2\}$ by the wreath recursion $r=\pair a$; these edge
  bisets embed naturally in their respective source and target vertex
  bisets.
\end{itemize}
Recalling from~\S\ref{ss:ex:DehnTwist} the notation $\mathfrak T$ for
the biprincipal biset of $T$, we have
$\gfB(f_n)\cong\mathfrak T^{\otimes n}\otimes\gfB(f_0)$. Note that, if
one changes $B_2$'s basis to $\{\ell'_2\coloneqq r^n\cdot \ell_2\}$,
one obtains a simpler wreath recursion
\[b=\pair c ,\quad c=\pair{c b c^{-1}},\quad r=\pair r
\]
that does not depend on $n$, but then one has to specify the
embeddings of the edge $\Z$-$\Z$-biset $E_2$ in $B_1$ and $B_2$
respectively: the basis element $ \ell_2$ of $E_2$ maps to the basis
element $\ell_2$ in $B_1$, but to $r^{-n}\times$ the basis element
$\ell'_2$ in $B_2$.

The maps $f_n$ all admit the same algebraic realization on a singly
noded complex stable curve, as
$f,i\colon\mathbb Y\rightrightarrows \mathbb X$. By convention, we
identify the post-critical points with the elementary loops in $G$
representing them, so that the post-critical set is now $\{a,b,c,d\}$.
We only give the map $f$, in red; we chose the co\"ordinates on the
spheres in $\mathbb Y$ so that $i$ is either the identity map or the
constant map on each component, so that it suffices to label, on
$\mathbb Y$, the $i$-preimages of the post-critical set. The mapping
$i$ on spheres is, anyways, the same as the map $\lambda$
in~\eqref{eq:z2pi:gog}. We indicate inside the spheres the
co\"ordinates that we chose so as to make the maps rational; they are
usually unimportant in those spheres of $\mathbb Y$ which get blown
down, and which are drawn shaded. We also attempt to give the
correct geometry to the spheres by indicating the angles at the
post-critical points and their $f$-preimages:
\begin{center}
\begin{tikzpicture}
  \coordinate (p) at (0,4);
  \coordinate (q) at (0,2.5);
  \draw (p) node[above left=0mm and -1mm] {$a$} node[below left] {\scriptsize $-2$}
  .. controls +(-1,0) and +(0.2,-0.1) .. +(-3,0.5)
  node[left] {$d$} node[above right=-1mm and 0mm] {\scriptsize $\infty$}
  .. controls +(2.0,-1.0) and +(0,1.0) .. +(-3,-1.5)
  node[right] {\scriptsize $0$}
  .. controls +(0,-0.5) and +(-1,0) .. (q) node[above left] {\scriptsize $2$}
  .. controls +(-0.1,0.2) and +(-0.1,-0.2) .. cycle;
  \draw[shrunksphere] (p) node[right] {\scriptsize $1$} .. controls +(0.1,0.5) and +(-0.5,0.5) .. +(3,0.5) node [below left] {\scriptsize $0$}
  .. controls +(0.1,-0.4) and +(0.1,0.4) .. +(3,-0.5) node[above left] {\scriptsize $\infty$}
  .. controls +(-0.5,-0.5) and +(0.1,-0.5) .. cycle;
  \draw (q) node[right] {\scriptsize $1$} .. controls +(0.1,0.5) and +(-0.5,0.5) .. +(3,0.5) node[right] {$b$} node [below left] {\scriptsize $0$}
  .. controls +(0.1,-0.4) and +(0.1,0.4) .. +(3,-0.5) node[right] {$c$} node[above left] {\scriptsize $\infty$}
  .. controls +(-0.5,-0.5) and +(0.1,-0.5) .. cycle;
	
  \coordinate (r) at (0,0);
  \draw (r) node[left=1pt] {\scriptsize $2$}
    .. controls +(-1,1) and +(0.8,-0.4) .. +(-3,0.7)
	node[left] {$d$} node[above right=-1mm and 0mm] {\scriptsize $\infty$}
	.. controls +(1.0,-0.5) and +(1.0,1.0) .. +(-3,-0.7) node[left] {$a$}
 	node[right=1pt] {\scriptsize $-2$}
	.. controls +(0.5,-0.5) and +(-1,-1) .. cycle;
	\draw (r) node[right] {\scriptsize $1$} .. controls +(0.2,1) and +(-0.5,0.5) .. +(3,0.7) node[right] {$b$} node [below left] {\scriptsize $0$}
	.. controls +(0.1,-0.4) and +(0.1,0.4) .. +(3,-0.7) node[right] {$c$} node[above left] {\scriptsize $\infty$}
	.. controls +(-0.5,-0.5) and +(0.2,-1) .. cycle;

        \node at (-4,3.5) {$\mathbb Y$};
        \node at (-4,0) {$\mathbb X$};
	\draw[->,red,thick] (-2,3) -- node[left] {$z^2-2$} (-2,0);
	\draw[->,red,thick] (2.2,2.5) -- node[right,pos=0.45] {$z^{-1}$} (2.2,0);
	\draw[->,red,thick] (1.2,4) -- node[left,pos=0.65] {$z$} (1.2,0);
\end{tikzpicture}
\end{center}

\subsection{\myboldmath The mating of $z^2-1$ with itself}\label{ex:z2m1}
This example appears
in~\cite{pilgrim:combinations}*{\S1.3.2}. Consider first the
polynomial $g(z)=z^2-1$, with post-critical set
$P_g=\{0,-1,\infty\}$. Choose as in~\S\ref{ex:z2pi} a basepoint
$*\in\R_+$ close to $\infty$, and a basis $\{\ell_1,\ell_2\}$
consisting of a short path $\ell_1$ from $\sqrt{*+1}$ to $*$ and an
upper half-circle from $-\sqrt{*+1}$ to $*$. Write
$H=\langle a,b,t\mid t b a\rangle$ for the fundamental group of
$\hC\setminus P_g$, with $a,b,t$ elementary loops around $-1,0,\infty$
respectively (i.e.\ $a,b,t$ follow the straight lines from $*$ to the
respective point). The presentation of the biset
$\subscript H{B(g)}_H$ is
\[a=\pair{a^{-1},b a}(1,2),\quad b=\pair{a,1},\quad t=\pair{t,1}(1,2).
\]
It may also be obtained by Algorithm~\ref{algo:angle2biset} starting
from the external angle $1/3$ of the map $z^2-1$.

Let now the branched covering $f$ be the \emph{mating} of $g$ with
itself. Topologically, it is the following map. Take two copies of
$\C$, and compactify each with a circle
$\{\infty e^{i\theta}\mid\theta\in\R/2\pi\Z\}$. Glue these two closed
disks by identifying $\infty e^{2i\pi\theta}$ on the first with
$\infty e^{-i\theta}$ on the second. Let $g$ act on each copy, and
note that they agree with the map
$\infty e^{i\theta}\mapsto\infty e^{2i\theta}$ on the circle at
infinity. In effect, we have decomposed $S^2$ in its upper and lower
hemispheres and let $g$ act independently on both.

A presentation of $B(f)$ may easily be computed. Consider a copy
$\overline H=\langle \bar a,\bar b,\bar t\mid\bar t\bar b\bar
a\rangle$, and the group
\[G=H*_{\langle t\rangle=\langle\bar t^{-1}\rangle}\bar H=\langle
a,b,\bar a,\bar b\mid b a\bar b\bar a\rangle.\]
This is a sphere tree of groups with two vertices and a single edge
corresponding to the fundamental group of the circle at infinity. The
biset $B(f)$ is the fundamental biset of the following tree of
bisets: it has two vertices each carrying the biset $B(g)$; we write
$\bar B(g)$ for the biset of the second vertex to distinguish it from
the first. An edge connects these vertices, carrying the biset
$B(z^2)$ which is $\langle t\rangle$ as a set with actions
$t^i\cdot t^j\cdot t^k=t^{2i+j+k}$. The inclusions of $B(z^2)$ in
$B(g)$ are as follows: choosing for each of the bisets
$B(z^2),B(g),\bar B(g)$ the same basis $\{\ell_1,\ell_2\}$, the maps are
\[\begin{cases} B(z^2)&\to B(g)\\
  \ell_1&\mapsto\ell_1\\
  \ell_2&\mapsto\ell_2
\end{cases}\qquad\text{and}\qquad
\begin{cases} B(z^2)&\to B(g)\\
  \ell_1&\mapsto\ell_1\\
  \ell_2&\mapsto t\cdot\ell_2.
\end{cases}\]

\noindent We obtain in this manner the following presentation for $\subscript G{B(f)}_G$:
\[a=\pair{a^{-1},b a}(1,2),\quad b=\pair{a,1},\quad
\bar a=\pair{\bar b\bar a,\bar a^{-1}}(1,2),\quad \bar b=\pair{1,\bar a}.\]

We naturally have an invariant multicurve $\{(b a)^G\}$, since
$b a=\pair{1,b a}(1,2)$; and its Thurston matrix is $(1/2)$. There is,
however, another invariant multicurve
\[\Gamma=\{x^G\}\text{ with }x=\bar a a,\]
since $x=\pair{1,x^{-\bar a}}$; and its Thurston matrix is $(1)$, so it is
a Levy obstruction.

We have $G=G_1*_\Z G_2$, with
$G_1=\langle a,\bar a,x^{-1}\mid \bar a a x^{-1}\rangle$ and
$G_2=\langle b^a,\bar b,x\mid b^a\bar bx\rangle$, glued along
$\langle x\rangle=\Z$. The sphere tree of groups therefore consists of a
single segment.

We change the basis of $B(f)$ to $Q=\{\ell_1,\bar a\ell_2\}$ so as to
make more visible the decomposition of $B(f)$ as a tree of bisets;
indeed in that basis $x=\pair{1,x^{-1}}$. The presentation of $B(f)$
becomes, on the generating set $\{a,\bar a,b^a,\bar b\}$,
\[a=\pair{x^{-1},x b^a}(1,2),\quad\bar a=\pair{\bar b,1}(1,2),\quad
b^a=\pair{1,a^{x^{-1}}},\quad\bar b=\pair{1,\bar a}.\]

Again we do not need to subdivide the tree of groups
barycentrically. We note that, in the new basis of $B(f)$, the wreath
recursion restricts to maps $G_1\to G_2^2\rtimes Q\perm$ and
$G_2\to 1\times G_1$. Therefore, the decomposition of $B(f)$ as a
sphere tree of bisets has two vertices $B_1,B_2$ joined by an edge
$E_2$ and such that $\rho$ sends $B_i$ to $G_i$ while $\lambda$ sends
$B_i$ to $G_{3-i}$, and an additional trivial vertex $B_3$ above $G_2$
joined by an edge $E_3$ to $B_1$:
\begin{equation}\label{eq:z2m1:gog}
  \begin{tikzpicture}[baseline=0.75cm]
    \node[inner sep=0pt] (B3) at (4,1.9) [label={right:$B_3$}] {$\bullet$};
    \node[inner sep=0pt] (B1) at (0,1.5) [label={left:$B_1$}] {$\bullet$};
    \node[inner sep=0pt] (B2) at (4,1.1) [label={right:$B_2$}] {$\bullet$};
    \draw[thick] (B3.center) -- node[above] {$E_3$} (B1.center) -- node[below] {$E_2$} (B2.center);

    \node[inner sep=0pt] (G1) at (0,0) [label={left:$G_1$}] {$\bullet$};
    \node[inner sep=0pt] (G2) at (4,0) [label={right:$G_2$.}] {$\bullet$};
    \draw[thick] (G1.center) -- node[below] {$\langle x\rangle$} (G2.center);

    \draw (B1) edge[rho] node[right] {$\rho$} (G1);
    \draw (B3) edge[rho] node[right,pos=0.75] {$\rho$} (G2);
    \draw (B1) edge[lambda] node[below,pos=0.7] {$\lambda$} (G2);
    \draw (B2) edge[lambda] node[above,pos=0.75] {$\lambda$} (G1);
    \draw (B3) edge[lambda,bend right=20] node[left,pos=0.75] {$\lambda$} (G2);
  \end{tikzpicture}
\end{equation}

The vertex bisets are again obtained by restricting $B(f)$
while using subsets of the basis $Q$, and are given as follows:
\begin{itemize}
\item the $G_2$-$G_1$-biset $B_1$ has in the basis $Q$ the wreath
  recursion
  \[a=\pair{x^{-1},x b^a}(1,2),\quad \bar a=\pair{\bar b,1}(1,2),\quad x^{-1}=\pair{1,x};\]
\item the $G_1$-$G_2$-biset $B_2$ has in the basis $\{\bar a\ell_2\}$ the
  wreath recursion
  \[b^a=\pair{a^{x^{-1}}},\quad \bar b=\pair{\bar a},\quad x=\pair{x^{-1}};\]
\item the $G_2$-$G_2$-biset $B_3$ is trivial in the basis
  $\{\ell_1\}$: it has the wreath recursion $b^a=\bar
  b=x=\pair{1}$. It corresponds to a sphere that gets shrunk under $i$.
\item the bisets $E_2$ and $E_3$ are the identity $\Z$-$\Z$-bisets in
  the bases $\{\bar a\ell_2\}$ and $\{\ell_1\}$ respectively, and
  embed naturally in their respective source and target vertex bisets.
\end{itemize}

To obtain the pieces of the decomposition, we consider the first
return map of $\lambda^{-1}\rho$ to $G_1$ via its biset
$B_2\otimes_{G_2}B_1$. Its wreath recursion is
\[a=\pair{x,ax^{-1}}(1,2),\quad \bar a=\pair{\bar a,1}(1,2),\quad
x^{-1}=\pair{1,x^{-1}}\]
which is isomorphic to $B(z^2)$, with $x^{-1}$ representing a loop
around the fixed point $1$.

The algebraic realization takes place on a singly noded complex stable
curve, as $f,i\colon\mathbb Y\rightrightarrows \mathbb X$. We keep the
convention of identifying the post-critical points with the elementary
loops in $G$ representing them, giving the map $f$ in red and forcing
$i$ to be either the identity or the constant map on each component:
\begin{center}
\begin{tikzpicture}
  \coordinate (p) at (0,4);
  \coordinate (q) at (0,2.5);
  \draw (p) node[left] {\scriptsize $-1$}
  .. controls +(100:1) and +(0.2,-0.1) .. +(-3,0.7) node[left] {$b^a$} node[below right=-0.5mm and 3.5mm] {\scriptsize $\infty$}
  .. controls +(2.0,-1.0) and +(2.0,1.0) .. +(-3,-2.2) node[left] {$\bar b$} node[above right=-0.5mm and 3.5mm] {\scriptsize $0$}
  .. controls +(0.2,0.1) and +(-100:1) .. (q) node[left] {\scriptsize $1$}
  .. controls +(100:1) and +(-100:1) .. cycle;
  \draw[shrunksphere] (p) node[right] {\scriptsize $1$}
  .. controls +(80:1) and +(-0.8,-0.4) .. +(3,0.6) node [above left=-1mm and 0mm] {\scriptsize $0$}
  .. controls +(-1.0,-0.5) and +(-1.0,0.5) .. +(3,-0.6) node[above left=1mm and 0mm] {\scriptsize $\infty$}
  .. controls +(-0.8,0.4) and +(-80:1) .. cycle;
  \draw (q) node[right] {\scriptsize $1$}
  .. controls +(80:1) and +(-0.8,-0.4) .. +(3,0.6) node[right] {$\bar a$} node [above left=-1mm and 0mm] {\scriptsize $0$}
  .. controls +(-1.0,-0.5) and +(-1.0,0.5) .. +(3,-0.6) node[right] {$a$} node[below left=-1mm and 0mm] {\scriptsize $\infty$}
  .. controls +(-0.8,0.4) and +(-80:1) .. cycle;
	
  \coordinate (r) at (0,0);
  \draw (r) node[left=1pt] {\scriptsize $1$}
  .. controls +(100:1) and +(0.8,-0.4) .. +(-3,0.8) node[left] {$a$} node[above right=-1mm and 0mm] {\scriptsize $\infty$}
  .. controls +(1.0,-0.5) and +(1.0,0.5) .. +(-3,-0.8) node[left] {$\bar a$} node[below right=-1mm and 0mm] {\scriptsize $0$}
  .. controls +(0.8,0.4) and +(-100:1) .. cycle;
  \draw (r) node[right] {\scriptsize $1$}
  .. controls +(80:1) and +(-0.8,-0.4) .. +(3,0.8) node[right] {$\bar b$} node [above left=-1mm and 0mm] {\scriptsize $0$}
  .. controls +(-1.0,-0.5) and +(-1.0,0.5) .. +(3,-0.8) node[right] {$b^a$} node[below left=-1mm and 0mm] {\scriptsize $\infty$}
  .. controls +(-0.8,0.4) and +(-80:1) .. cycle;

  \node at (-4,3.5) {$\mathbb Y$};
  \node at (-4,0) {$\mathbb X$};
  \draw[<-,red,thick] (-1,0) -- node[left,pos=0.4] {$z^2$} +(0,3.25);
  \draw[<-,red,thick] (0.8,0) -- node[left,pos=0.52] {$z$} +(0,2.5);
  \draw[<-,red,thick] (1.6,0) -- node[right,pos=0.325] {$z$} +(0,4);
\end{tikzpicture}
\end{center}

\subsection{Maps doubly covered by diagonal torus endomorphisms}\label{ex:endomn}
We consider next the endomorphism
$(\begin{smallmatrix}m&0\\0&n\end{smallmatrix})$ of the torus
$\R^2/\Z^2$, of degree $m n$, and its projection to a map
$f_{m,n}\colon S^2\selfmap$ via the Weierstra\ss\ function
$\wp$ of the square lattice.  The example $m=3,n=2$ is treated
in~\cite{pilgrim:combinations}*{\S1.3.3}. Without loss of generality,
we restrict ourselves to $m\ge n$.

The critical points of $\wp$ are at $\frac12\Z^2$, so $f_{m,n}$ has
$\{\wp(0),\wp(\frac12),\wp(\frac i2),\wp(\frac{1+i}2)\}$ as
post-critical set. If $m=n$, then $f$ is a rational map --- a
\emph{flexible Latt\`es map}, see~\cite{milnor:lattes}
and~\S\ref{ss:lattes} --- and its pullback map $\sigma_{f_{m,n}}$ on
Teichm\"uller space is the identity. In the general case, the pullback
map associated to the map
$f=\wp\circ(\begin{smallmatrix}p&q\\r&s\end{smallmatrix})\circ\wp^{-1}$
is $\sigma_f(z)=(p z+r)/(q z+s)$, if one identifies Teichm\"uller space
$\mathscr T_{\{a,b,c,d\}}$ with the upper half plane.

Let us write
$a=\wp(0),b=\wp(\frac12),c=\wp(\frac i2),d=\wp(\frac{1+i}2)$; the
post-critical graph depends on the parity of $m$ and $n$, but in all
cases $a$ is a fixed point, $b$ maps to $a$ or $b$, etc. The map
$f_{m,n}$ may be given by considering the rectangle
$X=[0,2]\times[0,1]\subset\R^2$, with sides identified under
$(0,y)\sim(2,y)$ and $(1-x,0)\sim(1+x,0)$ and
$(1-x,2)\sim(1+x,2)$. This is topologically a sphere, and metrically a
``pillowcase''. Consider next the rectangle $Y=[0,2m]\times[0,n]$ and
the maps $f,i\rightrightarrows Y\to X$ given by $i(x,y)=(x/m,y/n)$ and
$f(x,y)=(x,y)$ on $[0,2]\times[0,1]$, extended by reflections in the
lines $y\in\Z$ and $x\in2\Z$. The picture for $m=3,n=2$ is given in
Figure~\ref{fig:m=3,n=2}.
\begin{figure}
\begin{center}
\begin{tikzpicture}
  \begin{scope}[xshift=6cm,yshift=-3cm,scale=0.4]
    \fill[gray!20] (0,0) rectangle +(6,6);
    \draw[gray,very thick] (0,0) rectangle +(12,6);
    \draw[gray,thin] (6,0) -- (6,6);
    \node (A) at (9,3) {A};
    \node[red] at (6,6) [above] {$a$};
    \draw[red,thick] (6,6) ++(-45:0.5) -- (A);
    \draw[red,thick,->] (5.5,6) arc (180:360:0.5);
    \node[green] at (12,6) [above right] {$b$};
    \draw[green,thick] (12,6) ++(225:0.5) -- (A);
    \draw[green,thick,->] (11.5,6) arc (180:270:0.5);
    \draw[green,thick,->] (0,5.5) arc (-90:0:0.5);
    \node[blue] at (12,0) [below right] {$c$};
    \draw[blue,thick] (12,0) ++ (135:0.5) -- (A);
    \draw[blue,thick,->] (12,0.5) arc (90:180:0.5);
    \draw[blue,thick,->] (0.5,0) arc (0:90:0.5);
    \node[black] at (6,0) [below] {$d$};
    \draw[black,thick] (6,0) ++ (45:0.5) -- (A);
    \draw[black,thick,->] (6.5,0) arc (0:180:0.5);

    \draw[dashed,thick] (4.2,0) -- +(0,6);
    \draw[dashed,thick] (7.8,0) -- +(0,6);
  \end{scope}

  \begin{scope}[inner sep=1pt]
    \foreach\x/\y in {0/0,2/2,4/0,6/2,8/0,10/2} \fill[gray!20] (\x,\y) rectangle +(2,2);
    \draw[gray,very thick] (0,0) rectangle +(12,4);
    \node at (6,4) [label={above:$a$}] {};
    \node at (12,4) [label=above right:$b$] {};
    \node at (12,0) [label=below right:$c$] {};
    \node at (6,0) [label={below:$d$}] {};
    \foreach\x in {2,4,...,10} \draw[gray,thin] (\x,0) -- +(0,4);
    \draw[gray,thin] (0,2) -- (12,2);
    \foreach\x/\y/\A in {0/2/2,2/0/1,4/2/4,6/0/3,8/2/6,10/0/5} \node (A\A) at (\x+1,\y+1) {\small A\A};
    \foreach\x/\A/\B in {2/1/2,6/3/4,10/5/6} {\draw[red] (A\A) -- ($(\x,2)+(-45:0.2)$) arc (-45:135:0.2) \bckarrowonline{0.5} arc (135:315:0.2) \fwdarrowonline{0.5} +(135:0.4) -- (A\B);}
    
    \foreach\x/\A/\B in {4/1/4,8/3/6} {\draw[green] (A\A) -- ($(\x,2)+(-135:0.2)$) arc (-135:45:0.2) \fwdarrowonline{0.5} arc (45:225:0.2) \fwdarrowonline{0.5} +(45:0.4) -- (A\B);}
    \draw[green] (0,2) ++ (45:0.2) -- (A2);
    \draw[green,->] (0,1.8) arc (-90:90:0.2);
    \draw[green] (12,2) ++ (-135:0.2) -- (A5);
    \draw[green,->] (12,2.2) arc (90:270:0.2);

    \foreach\x/\theta/\A/\c in {2/45/1/black,4/135/1/blue,6/45/3/black,8/135/3/blue,10/45/5/black} {\draw[\c,->] (\x,0) ++ (0.2,0) arc (0:180:0.2); \draw[\c] (\x,0) ++ (\theta:0.2) -- (A\A);}

    \foreach\x/\theta/\A/\c in {2/-135/2/black,4/-45/4/blue,6/-135/4/black,8/-45/6/blue,10/-135/6/black} {\draw[\c,->] (\x,4) ++ (-0.2,0) arc (180:360:0.2); \draw[\c] (\x,4) ++ (\theta:0.2) -- (A\A);}

    \draw[blue] (12,0) ++ (135:0.2) -- (A5);
    \draw[blue,->] (12,0.2) arc (90:180:0.2);
    \draw[blue,->] (0.2,0) arc (0:90:0.2);
    \draw[blue] (0,4) ++ (-45:0.2) -- (A2);
    \draw[blue,->] (0,3.8) arc (-90:0:0.2);
    \draw[blue,->] (11.8,4) arc (180:270:0.2);

    \foreach\x in {1.4,2.6,5.4,6.6,9.4,10.6} \draw[dashed,thick] (\x,0) -- +(0,4);
    \end{scope}
    \draw[thick,->] (3,-0.5) -- node[below] {$i$} (5.5,-2.5);
    \draw[thick,->] (11.5,-0.5) -- node[right] {$f$} (11,-2.5);
  \end{tikzpicture}
\end{center}
\caption{The map doubly covered by the torus endomorphism $m=3,n=2$}\label{fig:m=3,n=2}
\end{figure}
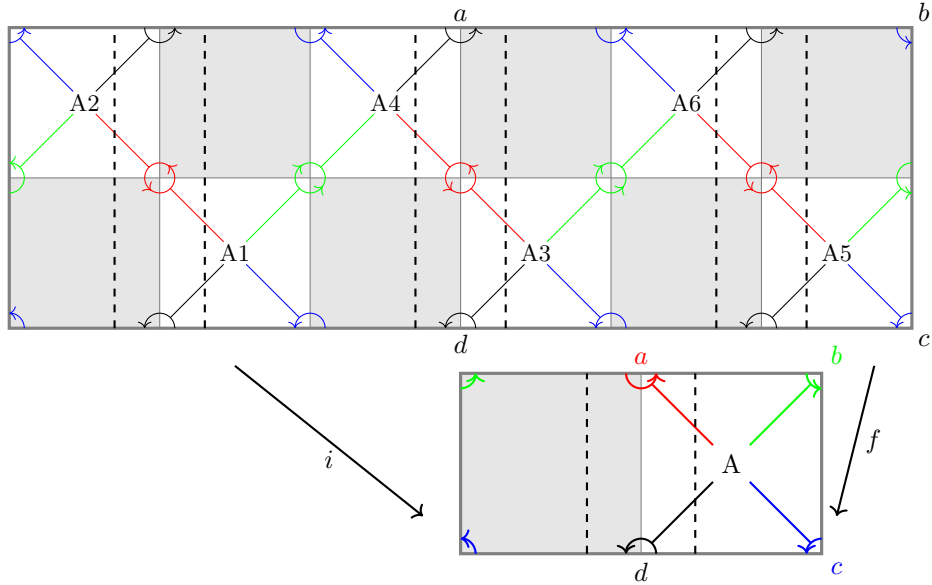

Since all post-critical points are orbispace points of order $2$, the
sphere group of $f_{m,n}$ is
\[G=\langle a,b,c,d\mid d c b a,a^2,b^2,c^2,d^2\rangle;\] its subgroup
$H=\langle b a,ad\rangle$ has index $2$ and $H\cong\Z^2$.  The biset
$B(f_{m,n})$ can easily be computed from this picture, but the answer
is not very illuminating. Let $\tilde f_{m,n}$ be the self-map of
$\R^2/\Z^2$ defined by $(x,y)\mapsto(m x,n y)$; then the biset
$B(\tilde f_{m,n})$ admits the following simple description as an
$H$-$H$-biset: identify $H$ with $\Z^2$. As a set,
$B(\tilde f_{m,n})=\Z^2$.  The left and right actions of $\Z^2$ are
given by
$v\cdot \beta\cdot
w=(\begin{smallmatrix}m&0\\0&n\end{smallmatrix})v+\beta+w$.  From the
degree-$2$ branched cover $\R^2/\Z^2\to S^2$ we deduce that
$B(\tilde f_{m,n})$ is a subbiset of $B(f_{m,n})$ of index $2$, namely
$B(f_{m,n})$ is, qua left $H$-set, the disjoint union of two copies of
$B(\tilde f_{m,n})$.

We turn to the decomposition of $f_{m,n}$. Set $x=ad$; then the
multicurve $\{x^G\}$ is invariant. It has $m$ preimages mapping by
degree $n$ to itself, so its Thurston matrix is $(m/n)$ and it is an
obstruction.

The sphere tree of groups decomposition of $G$ associated with $\{x^G\}$ is
$G=G_1*_{\langle x\rangle}G_2$ with
$G_1=\langle a,d,x^{-1}\mid a d x^{-1}\rangle$ and
$G_2=\langle b,c,x\mid b c x\rangle$, and the tree of bisets
decomposition of $B(f_{m,n})$ has $m+1$ vertices arranged in a
chain. Note that we needed, in this case, to consider the barycentric
subdivision of the tree of groups with two vertices and one edge,
because the map $\lambda$ sends some vertices to annuli. Here is the
graph for $m=5$; for even $m$, both endpoints of the tree of bisets
would map to $G_1$ by $\rho$:
\begin{equation}\label{eq:torusmn:gog}
  \begin{tikzpicture}[baseline=1.2cm]
    \node[inner sep=0pt] (B1) at (0,2.7) [label={left:$B_1$}] {$\bullet$};
    \node[inner sep=0pt] (B2) at (4,2.6) [label={right:$B_2$}] {$\bullet$};
    \node[inner sep=0pt] (B3) at (0,1.9) [label={left:$B_3$}] {$\bullet$};
    \node[inner sep=0pt] (B4) at (4,1.8) [label={right:$B_4$}] {$\bullet$};
    \node[inner sep=0pt] (B5) at (0,1.1) [label={left:$B_5$}] {$\bullet$};
    \node[inner sep=0pt] (B6) at (4,1.0) [label={right:$B_6$}] {$\bullet$};
    \draw[thick] (B1.center) -- node[fill=white,inner sep=0] {$\circ$} (B2.center) -- node[fill=white,inner sep=0]
    {$\circ$} (B3.center) -- node[fill=white,inner sep=0] {$\circ$} (B4.center) -- node[fill=white,inner sep=0]
    {$\circ$} (B5.center) -- node[fill=white,inner sep=0] {$\circ$} (B6.center);

    \node[inner sep=0pt] (G1) at (0,0) [label={left:$G_1$}] {$\bullet$};
    \node[inner sep=0pt] (G2) at (4,0) [label={right:$G_2$.}] {$\bullet$};
    \draw[thick] (G1.center) -- node[fill=white,inner sep=0] (Ge) {$\circ$} node[below] {$\langle x\rangle$} (G2.center);

    \draw (B1) edge[rho] node[right,pos=0.8] {$\rho$} (G1);
    \draw (B1) edge[lambda,bend right=58] node[left] {$\lambda$} (G1);
    \draw (B2) edge[rho] node[right,pos=0.85] {$\rho$} (G2);
    \draw (B2) edge[lambda,bend right=10] node[left=1pt,pos=0.55] {$\lambda$} (Ge);
    \draw (B3) edge[lambda,bend left=10] node[right=2pt,pos=0.7] {$\lambda$} (Ge);
    \draw (B4) edge[lambda] node[below] {$\lambda$} (Ge);
    \draw (B5) edge[lambda] node[below,pos=0.4] {$\lambda$} (Ge);
    \draw (B6) edge[lambda,bend right] node[left] {$\lambda$} (G2);
  \end{tikzpicture}
\end{equation}

We give directly the complex stable curve, which is singly
noded. Denote by $T_n$ the Chebyshev polynomial
$T_n(z)=\cos(n\arccos z)$. We keep the convention of identifying the
post-critical points with the elementary loops in $G$ representing
them, giving the map $f$ in red and forcing $i$ to be either the
identity or the constant map on each component --- and, in the latter
case, indicating the blown-down spheres in shade. We consider
$m$ odd; if $m$ is even, then the first and last spheres both map to
the left sphere by $i$:
\begin{center}
\begin{tikzpicture}
  \coordinate (p0) at (0,4);
  \coordinate (p1) at (0,3);
  \coordinate (p2) at (0,2);
  
  \draw (p0) node[left=2mm] {\scriptsize $\infty$}
  .. controls +(170:1) and +(30:1) .. +(-3,0.6) node[left] {$a$} node[below right=-1mm and 0mm] {\scriptsize $-1$}
  .. controls +(-60:0.5) and +(60:0.5) .. +(-3,-0.6) node[left] {$d$} node[above right=-1mm and 0mm] {\scriptsize $1$}
  .. controls +(-30:1) and +(-170:1) .. cycle;
  \draw[shrunksphere] (p0) node[below right=-1.5mm and 1.5mm] {\scriptsize $\infty$}
  .. controls +(10:1) and +(160:1) .. +(3,0.5) node[below left=-1mm and 0mm] {\scriptsize $-1$}
  .. controls +(-110:0.5) and +(110:0.5) .. +(3,-0.8) node[above left=-1mm and -0mm] {\scriptsize $1$}
  .. controls +(-160:1) and +(-10:1) .. (p1) node[above right=-1.5mm and 1.5mm] {\scriptsize $0$}
  .. controls +(10:0.6) and +(-10:0.6) .. cycle;
  \draw[shrunksphere] (p1) node[below left=-1.5mm and 1.5mm] {\scriptsize $0$}
  .. controls +(170:1) and +(20:1) .. +(-3,-0.2) node[below right=-1mm and 0mm] {\scriptsize $-1$}
  .. controls +(-70:0.5) and +(70:0.5) .. +(-3,-1.5) node[above right=-1mm and -0mm] {\scriptsize $1$}
  .. controls +(-20:1) and +(-170:1) .. (p2) node[above left=-1.5mm and 1.5mm] {\scriptsize $\infty$}
  .. controls +(170:0.6) and +(-170:0.6) .. cycle;
  \draw[loosely dotted,thick] (p0) -- (p2);
  
  \draw (p2) node[right=2mm] {\scriptsize $\infty$}
  .. controls +(10:1) and +(150:1) .. +(3,0.6) node[right] {$b$} node[below left=-1mm and 0mm] {\scriptsize $1$}
  .. controls +(-120:0.5) and +(120:0.5) .. +(3,-0.6) node[right] {$c$} node[above left=-1mm and 0mm] {\scriptsize $-1$}
  .. controls +(-150:1) and +(-10:1) .. cycle;

  \coordinate (r) at (0,0);
  \draw (r) node[left=2mm] {\scriptsize $\infty$}
  .. controls +(170:1) and +(30:1) .. +(-3,0.7) node[left] {$a$} node[below right=-1mm and 0mm] {\scriptsize $-1$}
  .. controls +(-60:0.5) and +(60:0.5) .. +(-3,-0.7) node[left] {$d$} node[above right=-1mm and 0mm] {\scriptsize $1$}
  .. controls +(-30:1) and +(-170:1) .. cycle;
  \draw (r) node[right=2mm] {\scriptsize $\infty$}
  .. controls +(10:1) and +(150:1) .. +(3,0.7) node[right] {$b$} node[below left=-1mm and 0mm] {\scriptsize $1$}
  .. controls +(-120:0.5) and +(120:0.5) .. +(3,-0.7) node[right] {$c$} node[above left=-1mm and 0mm] {\scriptsize $-1$}
  .. controls +(-150:1) and +(-10:1) .. cycle;

  \node at (-4,3.5) {$\mathbb Y$};
  \node at (-4,0) {$\mathbb X$};
  \draw[<-,red,thick] (-2.2,0) -- node[left,pos=0.27] {$T_n$} +(0,4);
  \draw[<-,red,thick] (-1.5,0) -- +(0,2.35);
  \node[red] at (0,1) {$\cdots\frac12(z^n+z^{-n})\cdots$};
  \draw[<-,red,thick] (1.5,0) -- +(0,3.65);
  \draw[<-,red,thick] (2.2,0) -- node[right,pos=0.53] {$T_n$} +(0,2); 
\end{tikzpicture}
\end{center}

\subsection{The formal mating \myboldmath $5/12\FM 5/12$}\label{ex:dh}
Douady and Hubbard, in their article~\cite{douady-h:thurston},
consider the formal mating with itself of the (obstructed) topological
polynomial with lamination angle $5/12$. This example is important
because it is a Thurston map with six curves, four of which forming a
chain, such that various subsets of these six curves define an annular
obstruction. Because curves in an obstruction cannot intersect, this
means that there is, in general, no such thing as a ``maximal
annular obstruction''.

We start by describing the presentation of the map $f_{5/12}$, using
this time the language of laminations;
see~\cite{bartholdi-n:mandelbrot2}. An identical recursion arises if
one starts with the angle $7/12$. In that picture, the generators lie
at angles $5/12$, $5/6$, $2/3$ and $1/3$. With the generators
$t,g_1,g_2,g_3,g_4$ corresponding to the angles $0,1/3,5/12,2/3,5/6$
in increasing order, the biset $B(f_{5/12})$ has presentation
\begin{xalignat*}{3}
  t&=\pair{t,1}(1,2), & g_1&=\pair{1,g_3}, &
  g_2&=\pair{g_1^{-1}g_2^{-1}g_3^{-1},g_3g_2g_1}(1,2),\\
  && g_3&=\pair{g_1,g_4}, & g_4&=\pair{g_2,1}.
\end{xalignat*}
The fundamental group is
$H=\langle t,g_1,g_2,g_3,g_4\mid g_4g_3g_2g_1t\rangle$.  The biset
$B(f_{5/12})$ admits a Levy obstruction $\{g_0^H\}$ with
$g_0=g_3^{g_2}g_1$, since $g_0=\pair{g_4^{t^{-1}},g_0^{g_2^{-1}}}$.
It is canonical. This Levy cycle comes from the external rays with
angles $1/3$ and $2/3$ landing together. If one considers the subgroup
$H_0=\langle t,g_0,g_2,g_4\mid g_4g_2g_0t\rangle$ of $H$, consisting
of paths that do not cross a fixed arc between the punctures $g_1$ and
$g_3$, one obtains a realizable biset with presentation
\[t=\pair{t,1}(1,2),\quad g_0=\pair{g_4^{t^{-1}},g_0^{g_2^{-1}}},\quad
  g_2=\pair{g_0^{-1}g_2^{-1},g_2g_0}(1,2),\quad g_4=\pair{g_2,1},
\]
which is the biset of a polynomial $\cong z^2-1.54369$. In effect,
passing to the subgroup $H_0$ amounts to considering only curves in
$(S^2,P_{f_{5/12}})$ that do not cross the curve $g_0^H$.

To compute the mating $f$ of $f_{5/12}$ with itself, we consider as
in~\S\ref{ex:z2m1} the group
\[G=G_{5/12}*_{\langle t\rangle}G_{5/12}=\langle g_1,\dots,g_4,h_1,\dots,h_4\mid g_4g_3g_2g_1h_4h_3h_2h_1\rangle.
\]
Its generators are lollipops around two copies
$\cdot\Rightarrow x_1\to x_2\to x_3\leftrightarrow x_4$ and
$\cdot\Rightarrow y_1\to y_2\to y_3\leftrightarrow y_4$ of the
post-critical set of $f_{5/12}$. Let us write
$t=h_4h_3h_2h_1=(g_4g_3g_2g_1)^{-1}$. The presentation of
$B(f)$ is then
\begin{xalignat*}{4}
  g_1&=\pair{1,g_3}, & g_2&=\pair{t g_4,g_4^{-1}t^{-1}}(1,2), & g_3&=\pair{g_1,g_4}, & g_4&=\pair{g_2,1},\\
  h_1&=\pair{h_3,1}, & h_2&=\pair{h_4^{-1}t^{-1},th_4}(1,2), & h_3&=\pair{h_4,h_1}, & h_4&=\pair{1,h_2}.
\end{xalignat*}

There are now many annular obstructions: setting as before $g_0=g_3^{g_2}g_1$
and $h_0=h_3^{h_2}h_1$, the multicurves $\{g_0^G\}$ and $\{h_0^G\}$
are both invariant with matrix $(1)$; however, setting $r=g_3h_1$ and
$s=g_1h_3$, we also have $r=\pair{s,g_4}$ and $s=\pair{h_4,r}$ so
$\{r^G,s^G\}$ is a Levy multicurve. The four curves $g_0,r,h_0,s$
intersect each other cyclically, so none of these are part of the
canonical obstruction, see~\cite{douady-h:thurston}*{Page~34}. In
fact, the canonical obstruction is
\[\{u^G,v^G\}\text{ with }u=h_2g_2,v=g_4^{t^{-1}}h_4,\]
since $u=\pair{v^{-1},v^{t g_4}}$ and
$v=\pair{1,u^{h_2}}$.  Its Thurston matrix is
$(\begin{smallmatrix}0&1\\2&0\end{smallmatrix})$.  We get a sphere
decomposition
\[G=G_1*_{\langle u\rangle}G_2*_{\langle v\rangle}G_3\text{ with }
\begin{cases}G_1&=\langle g_2,h_2,u^{-1}\mid h_2g_2u^{-1}\rangle,\\
  G_2&=\langle g_1,g_3^{g_2},h_1^{g_2},h_3,u,v\mid v h_3uh_1^{g_2}g_3^{g_2}g_1\rangle,\\
  G_3&=\langle g_4^{t^{-1}},h_4,v^{-1}\mid g_4^{t^{-1}}h_4v^{-1}\rangle.
\end{cases}
\]

\noindent The corresponding sphere tree of bisets decomposition is
\begin{equation}\label{eq:5/12:gog}
  \begin{tikzpicture}[baseline=0.75cm]
    \node[inner sep=0pt] (B1) at (0,1.5) [label={left:$B_1$}] {$\bullet$};
    \node[inner sep=0pt] (B2) at (3,1.9) [label={above:$B_2$}] {$\bullet$};
    \node[inner sep=0pt] (B3) at (6,2.0) [label={right:$B_3$}] {$\bullet$};
    \node[inner sep=0pt] (B4) at (3,1.1) [label={below:$B_4$}] {$\bullet$};
    \node[inner sep=0pt] (B5) at (6,1.0) [label={right:$B_5$}] {$\bullet$};
    \draw[thick] (B3.center) -- node[fill=white,inner sep=0] {$\circ$} (B2.center)
    -- node[fill=white,inner sep=0] {$\circ$} (B1.center)
    -- node[fill=white,inner sep=0] {$\circ$} (B4.center)
    -- node[fill=white,inner sep=0] {$\circ$} (B5.center);

    \node[inner sep=0pt] (G1) at (0,-0.2) [label={left:$G_1$}] {$\bullet$};
    \node[inner sep=0pt] (G2) at (3,-0.2) [label={below:$G_2$}] {$\bullet$};
    \node[inner sep=0pt] (G3) at (6,-0.2) [label={right:$G_3$,}] {$\bullet$};
    \draw[thick] (G1.center) -- node[fill=white,inner sep=0] (Gu) {$\circ$} node[below] {$\langle u\rangle$} (G2.center) -- node[fill=white,inner sep=0] (Gv) {$\circ$} node[below] {$\langle v\rangle$} (G3.center);

    \draw (B1) edge[rho] node[right] {$\rho$} (G1);
    \draw (B2) edge[rho] node[right,pos=0.92] {$\rho$} (G2);
    \draw (B3) edge[rho] node[right,pos=0.8] {$\rho$} (G3);
    \draw (B1) edge[lambda] node[below] {$\lambda$} (Gv);
    \draw (B2) edge[lambda,bend right=30] node[left=1pt,pos=0.54] {$\lambda$} (G2);
    \draw (B3) edge[lambda,bend right=10] node[right=2pt,pos=0.85] {$\lambda$} (G1);
    \draw (B4) edge[lambda] node[above] {$\lambda$} (G3);
    \draw (B5) edge[lambda,bend right=30] node[left] {$\lambda$} (G3);
  \end{tikzpicture}
\end{equation}
with the following vertex bisets:
\begin{itemize}
\item the $\langle v\rangle$-$G_1$-biset $B_1$ has, in basis
  $\{\ell_1,g_1^{-1}g_2^{-1}g_3^{-1}\ell_2\}$, the wreath recursion
  \[g_2=\pair{1,1}(1,2),\qquad h_2=\pair{v,v^{-1}}(1,2),\qquad u^{-1}=\pair{v,v^{-1}};\]
\item the $G_2$-$G_2$-biset $B_2$ has, in basis
  $\{g_1^{-1}g_2^{-1}g_3^{-1}\ell_2\}$, the wreath recursion
  \begin{xalignat*}{3}
    g_1&=\pair{(g_3^{g_2})^{g_1}}, & g_3^{g_2}&=\pair{g_1}, & u&=\pair{v},\\
    h_1^{g_2}&=\pair{h_3}, & h_3&=\pair{(h_1^{g_2})^{g_3^{g_2}g_1}}, & v&=\pair{u^{g_3^{g_2}g_1}};
  \end{xalignat*}
\item the $G_1$-$G_3$-biset $B_3$ has, in basis $\{\ell_2\}$, the
  wreath recursion
  \[g_4^{t^{-1}}=\pair{g_2},\qquad h_4=\pair{h_2},\qquad v=\pair{u^{g_2}};\]
\item the $G_3$-$G_2$-biset $B_4$ has, in basis $\{\ell_1\}$, the
  wreath recursion
  \begin{xalignat*}{3}
    g_1&=\pair{1}, & g_3^{g_2}&=\pair{g_4^{t^{-1}}}, & u&=\pair{v^{-1}},\\
    h_1^{g_2}&=\pair{1}, & h_3&=\pair{h_4}, & v&=\pair{1};
  \end{xalignat*}
\item the $G_3$-$G_3$-biset $B_5$ has, in basis $\{\ell_1\}$, the wreath
  recursion $g_4^{t^{-1}}=h_4=v=\pair{1}$. It corresponds to a sphere
  being contracted to a point.
\end{itemize}

The only return maps of $B(f)$ are the trivial biset $B_5$ and the
degree-$1$ biset $B_2$. This last biset is the biset of an outer
automorphism $\varphi$ of order two: the wreath recursion of
$B_2\otimes B_2$ is conjugation by $g_0=g_3^{g_2}g_1$. Indeed
$\varphi(g_0)=g_0$ and the action of $\varphi$ is
\[g_3^{g_2}\mapsto g_1\mapsto (g_3^{g_2})^{g_0},\quad u\mapsto
v\mapsto u^{g_0},\quad h_1^{g_2}\mapsto h_3\mapsto(h_1^{g_2})^{g_0}.\]
On top of $\varphi(g_0)=g_0$, we also have
$\varphi(h_0^{G_2})=h_0^{G_2}$ and $\varphi(r^{g_2G_2})=s^{G_2}$ and
$\varphi(s^{G_2})=r^{g_2G_2}$; so the simple closed curves
$g_0,r,h_0,s$ can be homotoped into periodic curves in the sphere
$G_2$.

\noindent We are ready to give the complex stable curve on which $f$
may be realized:
\begin{center}
\begin{tikzpicture}
  \coordinate (p0) at (-0.8,4.5);
  \coordinate (p1) at (0.8,4.5);
  \coordinate (q0) at (-0.8,2.4);
  \coordinate (q1) at (0.8,2.4);

  \draw[shrunksphere] (p0) node [below left=-4.5pt and 14pt] {\scriptsize $0$} node[right] {\scriptsize $-1$}
  .. controls +(170:3.5) and +(-170:3.5) .. (q0)
  node [above left=-3.5pt and 14pt] {\scriptsize $\infty$} node[right] {\scriptsize $1$}
  .. controls +(170:2.2) and +(-170:2.2) .. cycle;

  \begin{scope}[yshift=4.5cm]
  \draw (-0.8,-0.8) rectangle (0.8,0.8);
  \node[anchor=east] at (-0.8,-0.8) {$y_4$};
  \node[anchor=south west] at (-0.8,-0.8) {\scriptsize $-i$};
  \node[anchor=west] at (0.8,-0.8) {$y_3$};
  \node[anchor=south east] at (0.8,-0.8) {\scriptsize $\infty$};
  \node[anchor=east] at (-0.8,0.8) {$x_4$};
  \node[anchor=north west] at (-0.8,0.8) {\scriptsize $i$};
  \node[anchor=west] at (0.8,0.8) {$x_3$};
  \node[anchor=north east] at (0.8,0.8) {\scriptsize $0$};
  \end{scope}
  \draw (p1) node[right=3mm] {\scriptsize $1$} node[left] {\scriptsize $1$}
  .. controls +(10:1) and +(150:1) .. +(3,0.6) node[right] {$x_1$} node[below left=-1mm and 0mm] {\scriptsize $0$}
  .. controls +(-120:0.5) and +(120:0.5) .. +(3,-0.6) node[right] {$y_1$} node[above left=-1mm and 0mm] {\scriptsize $\infty$}
  .. controls +(-150:1) and +(-10:1) .. cycle;

  \begin{scope}[yshift=2.4cm]
  \draw (-0.8,-0.8) rectangle (0.8,0.8);
  \node[anchor=east] at (-0.8,-0.8) {$y_2$};
  \node[anchor=south west] at (-0.8,-0.8) {\scriptsize $\infty$};
  \node[anchor=south east] at (0.8,-0.8) {\scriptsize $-i$};
  \node[anchor=east] at (-0.8,0.8) {$x_2$};
  \node[anchor=north west] at (-0.8,0.8) {\scriptsize $0$};
  \node[anchor=north east] at (0.8,0.8) {\scriptsize $i$};
  \end{scope}
  \draw[shrunksphere] (q1) node[right=3mm] {\scriptsize $1$} node[left] {\scriptsize $-1$}
  .. controls +(10:1) and +(150:1) .. +(3,0.6) node[below left=-1mm and 0mm] {\scriptsize $0$}
  .. controls +(-120:0.5) and +(120:0.5) .. +(3,-0.6) node[above left=-1mm and 0mm] {\scriptsize $\infty$}
  .. controls +(-150:1) and +(-10:1) .. cycle;

  \coordinate (r0) at (-0.8,0);
  \coordinate (r1) at (0.8,0);
  \draw (r0) node[left=3mm] {\scriptsize $1$} node[right] {\scriptsize $1$}
  .. controls +(170:1) and +(30:1) .. +(-3,0.7) node[left] {$x_1$} node[below right=-1mm and 0mm] {\scriptsize $0$}
  .. controls +(-60:0.5) and +(60:0.5) .. +(-3,-0.7) node[left] {$y_1$} node[above right=-1mm and 0mm] {\scriptsize $\infty$}
  .. controls +(-30:1) and +(-170:1) .. cycle;
  \draw (-0.8,-0.8) rectangle (0.8,0.8);
  \node[anchor=east] at (-0.8,-0.8) {$y_3$};
  \node[anchor=south west] at (-0.8,-0.8) {\scriptsize $\infty$};
  \node[anchor=west] at (0.8,-0.8) {$y_4$};
  \node[anchor=south east] at (0.8,-0.8) {\scriptsize $-i$};
  \node[anchor=east] at (-0.8,0.8) {$x_3$};
  \node[anchor=north west] at (-0.8,0.8) {\scriptsize $0$};
  \node[anchor=west] at (0.8,0.8) {$x_4$};
  \node[anchor=north east] at (0.8,0.8) {\scriptsize $i$};
  \draw (r1) node[right=3mm] {\scriptsize $1$} node[left] {\scriptsize $-1$}
  .. controls +(10:1) and +(150:1) .. +(3,0.7) node[right] {$x_2$} node[below left=-1mm and 0mm] {\scriptsize $0$}
  .. controls +(-120:0.5) and +(120:0.5) .. +(3,-0.7) node[right] {$y_2$} node[above left=-1mm and 0mm] {\scriptsize $\infty$}
  .. controls +(-150:1) and +(-10:1) .. cycle;

  \node at (-4,3.5) {$\mathbb Y$};
  \node at (-4,0) {$\mathbb X$};
  \draw[<-,red,thick] (-2.8,0) -- node[left,pos=0.383] {$z^2$} +(0,3.45);
  \draw[<-,red,thick] (-0.3,0) -- node[left,pos=0.5] {$z$} +(0,2.4);
  \draw[<-,red,thick] (0.3,0) --  node[right,pos=0.266] {$\frac{z-i}{iz-1}$} +(0,4.5);
  \draw[<-,red,thick] (2.2,0) -- node[left,pos=0.5] {$z$} +(0,2.4); 
  \draw[<-,red,thick] (3,0) -- node[right,pos=0.266] {$z$} +(0,4.5); 
\end{tikzpicture}
\end{center}
Note that the co\"ordinates on the central small sphere in $\mathbb X$
are not uniquely determined; rather, choose any M\"obius
transformation $\mu$ that is an involution and that does not fix $0$,
$1$ or $\infty$. Then, once the left and right small spheres are
labelled by $0,1,\infty$ as above, the central small sphere is
labelled by $0,1,\infty,\mu(0),\mu(1),\mu(\infty)$ and the covering
$f$ is given, on the top central small sphere of $\mathbb Y$, by
$\mu$.

\subsection{Blowing up an arc}\label{ex:pilgrim}
In~\cite{pilgrim:combinations}*{\S1.3.4}, Kevin Pilgrim describes a
self-covering of the sphere, obtained from the $z\mapsto 2z$ map on
the torus by rotating and blowing up an edge. It is a degree-$5$ map
$f$, and Pilgrim asks whether it can be realized as a complex map
(from the construction, it is clear that $f$ is expanding, so it
cannot have any Levy obstruction, see Theorem~\ref{thm:ExpCr}). We
start by describing the map similarly to the examples
in~\S\ref{ex:endomn}:

\begin{figure}
\begin{center}\begin{tikzpicture}
    \begin{scope}
      \fill[gray!20] (0,3) rectangle +(3,3);
      \fill[gray!20] (3,0) rectangle +(3,3);
      \def\bendota{(6,3) .. controls (8,4) and (10,5) .. (12,3)}
      \def\bendotb{(9,0) .. controls (10,2) and (11,4) .. (9,6)}
      \def\bendtta{(6,3) .. controls (9,3) and (9,3) .. (9,0)}
      \def\bendttb{(9,0) .. controls (7.5,1) and (7,1.5) .. (6,3)}
      \fill[gray!20] \bendtta \bendttb;
      \begin{scope}
        \clip \bendotb -- (6,6) -- (6,0) -- cycle;
        \fill[gray!20] \bendota -- (12,6) -- (6,6) -- cycle;
      \end{scope}
      \begin{scope}
        \clip \bendotb -- (12,6) -- (12,0) -- cycle;
        \fill[gray!20] \bendota -- (12,0) -- (6,0) -- cycle;
      \end{scope}

      \draw[gray,very thick] (0,0) -- (0,6) -- (12,6) -- (12,0) -- cycle;
      \draw[gray,thin] (6,0) -- (6,6);
      \draw[gray,thin] (3,0) -- (3,6);
      \draw[gray,thin] (0,3) -- (6,3);
      \draw[gray,thin,name path=p12a] \bendota;
      \draw[gray,thin,name path=p12b] \bendotb;
      \draw[gray,thin,name path=p23a] \bendtta;
      \draw[gray,thin,name path=p23b] \bendttb;
      \node (A1) at (11,5) {A1};
      \node (A2) at (9.3,3.3) {A2};
      \node (A3) at (6.7,0.7) {A3};
      \node (A4) at (4.5,4.5) {A4};
      \node (A5) at (1.5,1.5) {A5};
      \path[name path=diag] (A2) -- (A3);
      \draw[red] (12,3) ++(110:0.3) -- (A1);
      \draw[red,->] (12,3.3) arc (90:270:0.3);
      \draw[red] (6,3) ++(15:0.3) .. controls (8,3.5) .. (A2);
      \draw[red] (6,3) ++ (-80:0.3) .. controls (6.3,1.5) .. (A3);
      \draw[red,->] (6.3,3) arc (0:360:0.3);
      \draw[red] (6,3) ++ (135:0.3) -- (A4);
      \draw[red] (0,3) ++ (-45:0.3) -- (A5);
      \draw[red,->] (0,2.7) arc (-90:90:0.3);
      \draw[green,name intersections={of=p12a and p12b,by=x12}] (x12) ++(45:0.3) -- (A1);
      \draw[green] (x12) ++(225:0.3) -- (A2);
      \draw[green,->] (x12) ++(0.3,0) arc (0:360:0.3);
      \draw[green,name intersections={of=p23b and diag,by=y23}] (y23) ++(225:0.3) -- (A3);
      \draw[green,->] (y23) ++(0.3,0) arc (0:360:0.3);
      \draw[green] (3,3) ++(45:0.3) -- (A4);
      \draw[green] (3,3) ++(225:0.3) -- (A5);
      \draw[green,->] (3.3,3) arc (0:360:0.3);
      \draw[blue] (9,6) ++(-30:0.3) -- (A1);
      \draw[blue,->] (8.7,6) arc (180:360:0.3);
      \draw[blue] (3,6) ++(-45:0.3) -- (A4);
      \draw[blue] (9,0) ++(75:0.3) .. controls (9.5,2) .. (A2);
      \draw[blue] (9,0) ++ (170:0.3) .. controls (7.5,0.3) .. (A3);
      \draw[blue,>-] (9.3,0) arc (0:180:0.3);
      \draw[blue] (3,0) ++(135:0.3) -- (A5);
      \draw[blue,->] (3.3,0) arc (0:180:0.3);
      \draw[blue,->] (2.7,6) arc (180:360:0.3);
      \draw[black] (12,6) ++(225:0.3) -- (A1);
      \draw[black,->] (11.7,6) arc (180:270:0.3);
      \draw[black,->] (0,5.7) arc (-90:0:0.3);
      \draw[black,->] (12,0.3) arc (90:180:0.3);
      \draw[black,name intersections={of=p23a and diag,by=x23}] (x23) ++(45:0.3) -- (A2);
      \draw[black,->] (x23) ++(0.3,0) arc (0:360:0.3);
      \draw[black] (6,0) ++(45:0.3) -- (A3);
      \draw[black,->] (6.3,0) arc (0:180:0.3);
      \draw[black] (6,6) ++(225:0.3) -- (A4);
      \draw[black,->] (5.7,6) arc (180:360:0.3);
      \draw[black] (0,0) ++(45:0.3) -- (A5);
      \draw[black,->] (0.3,0) arc (0:90:0.3);
    \end{scope}
    \begin{scope}[xshift=6cm,yshift=-3cm,scale=0.4]
      \fill[gray!20] (0,0) rectangle +(6,6);
      \draw[gray,very thick] (0,0) rectangle +(12,6);
      \draw[gray,thin] (6,0) -- (6,6);
      \node (A) at (9,3) {\large A};
      \node[red] at (6,6) [above] {\large $a$};
      \draw[red,thick] (6,6) ++(-45:0.5) -- (A);
      \draw[red,thick,->] (5.5,6) arc (180:360:0.5);
      \node[green] at (12,6) [above right] {\large $b$};
      \draw[green,thick] (12,6) ++(225:0.5) -- (A);
      \draw[green,thick,->] (11.5,6) arc (180:270:0.5);
      \draw[green,thick,->] (0,5.5) arc (-90:0:0.5);
      \node[blue] at (12,0) [below right] {\large $c$};
      \draw[blue,thick] (12,0) ++ (135:0.5) -- (A);
      \draw[blue,thick,->] (12,0.5) arc (90:180:0.5);
      \draw[blue,thick,->] (0.5,0) arc (0:90:0.5);
      \node[black] at (6,0) [below] {\large $d$};
      \draw[black,thick] (6,0) ++ (45:0.5) -- (A);
    \draw[black,thick,->] (6.5,0) arc (0:180:0.5);
    \end{scope}
    \draw[thick,->] (3,-0.5) -- node[below] {$i$} (5.5,-2.5);
    \draw[thick,->] (11.5,-0.5) -- node[right] {$f$} (11,-2.5);
\end{tikzpicture}\end{center}
\caption{Pilgrim's ``blowing up an arc'' subdivision rule}\label{fig:pilgrimmap}
\end{figure}
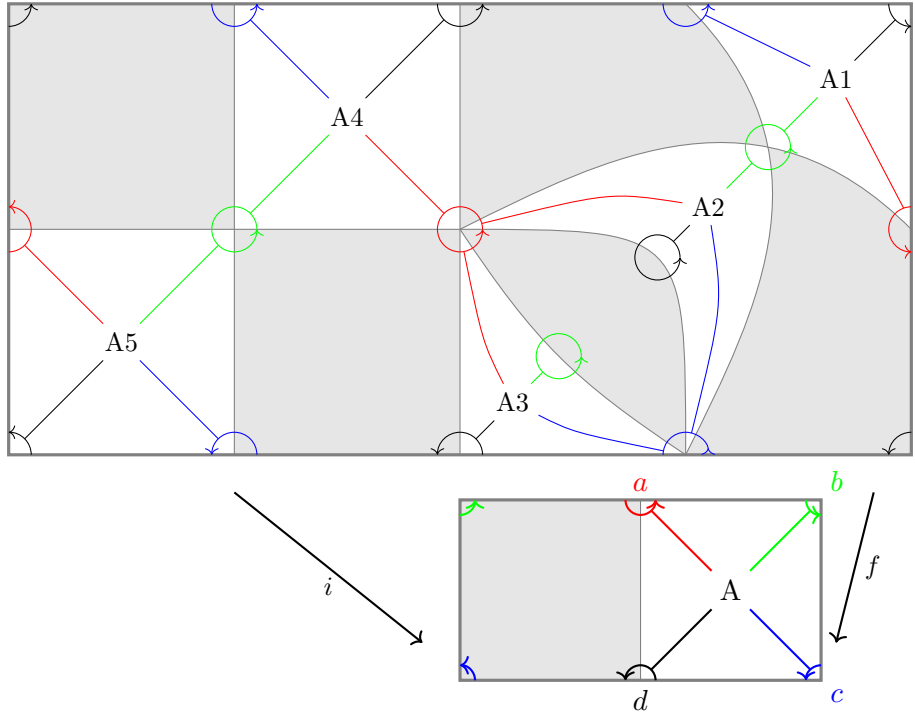

The map may be group-theoretically presented as follows --- see
Figure~\ref{fig:pilgrimmap}. Consider
$G=\langle a,b,c,d\mid d c b a\rangle$ generated by lollipops around the
punctures. In the basis
$\{\ell_1,\dots,\ell_5\}$ consisting of straight paths from
$i(\mathrm{A1}),\dots,i(\mathrm{A5})$ to the basepoint $\mathrm{A}$,
the biset $B(f)$ is presented as
\begin{align*}
  a &= \pair{c^{-1},1,1,1,c}(1,5)(2,4,3),\\
  b &= \pair{1,1,1,d,d^{-1}}(1,2)(4,5),\\
  c &= \pair{a,1,1,a^{-1},1}(1,4)(2,3,5),\\
  d &= \pair{b,1,d,a,c}.
\end{align*}

\noindent Consider $x=ac$. One then computes directly
\[x=\pair{c^{-1},c x^{-1},1,1,x^{c^{-1}}}(1,2)(4,5),\]
so $\{x^G\}$ is an annular obstruction, with Thurston matrix
$(\frac12+\frac12)$. This already answers Pilgrim's question in the
negative. However, let us study this example further, and decompose
$B(f)$ as a sphere tree of bisets.

The tree of groups decomposition is $G=G_1*_\Z G_2$, with
$G_1=\langle a,c,x^{-1}\mid a c x^{-1}\rangle$ and
$G_2=\langle x,b,d^c\mid x d^c b\rangle$.  To compute the sphere tree of
bisets, we first change the basis of $B(f)$ to
$Q\coloneqq\{\ell_1,c\ell_2,c\ell_3,c\ell_4,c\ell_5\}$ so as to simplify the
presentation of $x$, and obtain
\begin{align*}
  a &= \pair{1,1,1,1,1}(1,5)(2,4,3),\\
  b &= \pair{c,c^{-1},1,d^c,d^{-c}}(1,2)(4,5),\\
  c &= \pair{x,1,1,x^{-1},1}(1,4)(2,3,5),\\
  d^c &= \pair{a,c,1,b^x,d^c},\\
  x &= \pair{1,x^{-1},1,1,x}(1,2)(4,5).
\end{align*}
The multicurve $\{x^G\}$ has $3$ preimages, corresponding to the
cycles $(1,2)(3)(4,5)$ of the permutation associated with $x$; so the sphere
tree of bisets into which $B(f)$ decomposes has $4$ vertex
bisets. They are arranged as follows:
\begin{equation}\label{eq:pilgrim:gog}
  \begin{tikzpicture}[baseline=0.9cm]
    \node[inner sep=0pt] (B1) at (0,1.8) [label={left:$B_1$}] {$\bullet$};
    \node[inner sep=0pt] (B2) at (4,2.6) [label={right:$B_2$}] {$\bullet$};
    \node[inner sep=0pt] (B3) at (4,1.8) [label={right:$B_3$}] {$\bullet$};
    \node[inner sep=0pt] (B4) at (4,1) [label={right:$B_4$}] {$\bullet$};
    \draw[thick] (B1.center) -- node[fill=white,inner sep=0] {$\circ$} (B2.center)
    (B1.center) -- node[fill=white,inner sep=0] {$\circ$} (B3.center)
    (B1.center) -- node[fill=white,inner sep=0] {$\circ$} (B4.center);

    \node[inner sep=0pt] (G1) at (0,-0.2) [label={left:$G_1$}] {$\bullet$};
    \node[inner sep=0pt] (G2) at (4,-0.2) [label={right:$G_2$.}] {$\bullet$};
    \draw[thick] (G1.center) -- node[fill=white,inner sep=0] (Ge) {$\circ$} node[below] {$\langle x\rangle$} (G2.center);

    \draw (B1) edge[rho] node[left] {$\rho$} (G1);
    \draw (B2) edge[rho] node[right,pos=0.8] {$\rho$} (G2);
    \draw (B1) edge[lambda] node[below,pos=0.4] {$\lambda$} (Ge);
    \draw (B2) edge[lambda,bend left=5] node[below,pos=0.54] {$\lambda$} (G1);
    \draw (B3) edge[lambda] node[right=2pt,pos=0.75] {$\lambda$} (Ge);
    \draw (B4) edge[lambda,bend right] node[left] {$\lambda$} (G2);
  \end{tikzpicture}
\end{equation}
The vertex bisets are again obtained by restricting $B(f)$ while using
subsets of the basis $Q$, and are given as follows:
\begin{itemize}
\item the $\langle x\rangle$-$G_1$-biset $B_1$ has in the basis $Q$
  the wreath recursion
  \begin{align*}
    a&=\pair{1,1,1,1,1}(1,5)(2,4,3),\\
    c&=\pair{x,1,1,x^{-1},1}(1,4)(2,3,5),\\
    x^{-1}&=\pair{x,1,1,x^{-1},1}(1,2)(4,5);
  \end{align*}
\item the $G_1$-$G_2$-biset $B_2$ has in the subbasis
  $\{\ell_1,c\ell_2\}$ of $Q$ the wreath recursion
  \[x=\pair{1,x^{-1}}(1,2),\quad b=\pair{c,c^{-1}}(1,2),\quad d^c=\pair{a,c};\]
\item the $\langle x\rangle$-$G_2$-biset $B_3$ has in the subbasis
  $\{c\ell_3\}$ of $Q$ the trivial wreath recursion $x=b=d^c=\pair1$;
  it corresponds to a sphere that gets blown down to a point on the
  annulus;
\item the $G_2$-$G_2$-biset $B_4$ has in the subbasis
  $\{c\ell_4,c\ell_5\}$ of $Q$ the wreath recursion
  \[x=\pair{1,x}(1,2),\quad b=\pair{d^c,d^{-c}}(1,2),\quad d^c=\pair{b^x,d^c};\]
  it is isomorphic to the biset of the map $z^2-2$.
\end{itemize}

The biset $B_1$ is the biset of the rational map
$z^3(4z+5)^2/(5z+4)^2$. The algebraic realization of $f$ takes place
on a singly noded complex stable curve. Keeping the same conventions,
it is
\begin{center}
\begin{tikzpicture}
  \coordinate (p0) at (0,4.4);
  \coordinate (p1) at (0,3.2);
  \coordinate (p2) at (0,2);
  
  \draw[shrunksphere] (p0) node[below left=-5pt and 4pt] {\scriptsize $\frac{-7+i\sqrt{15}}8$}
  .. controls +(170:1) and +(90:3) .. node[pos=0.5,below] {\scriptsize $0$} node[pos=0.9,right] {\scriptsize $-\frac54$} +(-3,-1.2)
  .. controls +(-90:3) and +(-170:0.5) .. node[pos=0.5,above] {\scriptsize $\infty$} node[pos=0.1,right] {\scriptsize $-\frac45$} (p2)
  node[below left=-8pt and 4pt] {\scriptsize $\frac{-7-i\sqrt{15}}8$}
  .. controls +(170:0.6) and +(-170:0.6) .. (p1) node[left=1.5mm] {\scriptsize $1$}
  .. controls +(170:0.6) and +(-170:0.6) .. cycle;
  \draw (p0) node[right=3mm] {\scriptsize $1$}
  .. controls +(10:1) and +(170:1) .. node[below=-2pt] {\scriptsize $-1$} +(3,0.5) node[right] {$a$} node[below left=-3pt and 1pt] {\scriptsize $0$}
  .. controls +(-130:0.5) and +(130:0.5) .. +(3,-0.5) node[right] {$c$} node[above left=-2pt and 1pt] {\scriptsize $\infty$}
  .. controls +(-170:1) and +(-10:1) .. cycle;
  \draw[shrunksphere] (p1) node[right=3mm] {\scriptsize $1$}
  .. controls +(10:1) and +(150:1) .. +(3,0.4) node[below left=-4pt and 0pt] {\scriptsize $0$}
  .. controls +(-120:0.5) and +(150:0.5) .. +(3,-0.4) node[above left=-1pt and 4pt] {\scriptsize $\infty$}
  .. controls +(-180:1) and +(-10:1) .. cycle;
  \draw (p2) node[right=3mm] {\scriptsize $1$}
  .. controls +(10:1) and +(170:1) .. node[below=-2pt] {\scriptsize $-1$} +(3,0.5) node[right] {$b$} node[below left=-3pt and 1pt] {\scriptsize $0$}
  .. controls +(-130:0.5) and +(130:0.5) .. +(3,-0.5) node[right] {$d^c$} node[above left=-2pt and 1pt] {\scriptsize $\infty$}
  .. controls +(-170:1) and +(-10:1) .. cycle;

  \coordinate (r) at (0,0);
  \draw (r) node[left=3mm] {\scriptsize $1$}
  .. controls +(170:1) and +(10:1) .. +(-3,0.7) node[left] {$a$} node[below right=-1mm and 0mm] {\scriptsize $0$}
  .. controls +(-50:0.5) and +(50:0.5) .. +(-3,-0.7) node[left] {$c$} node[above right=-1mm and 0mm] {\scriptsize $\infty$}
  .. controls +(-10:1) and +(-170:1) .. cycle;
  \draw (r) node[right=3mm] {\scriptsize $1$}
  .. controls +(10:1) and +(150:1) .. +(3,0.7) node[right] {$b$} node[below left=-4pt and 0mm] {\scriptsize $0$}
  .. controls +(-120:0.5) and +(150:0.5) .. +(3,-0.7) node[right] {$d^c$} node[above left=-1pt and 1mm] {\scriptsize $\infty$}
  .. controls +(-180:1) and +(-10:1) .. cycle;

  \node at (-4,3.2) {$\mathbb Y$};
  \node at (-4,0) {$\mathbb X$};
  \draw[<-,red,thick] (-2.2,0) -- node[right=-2pt,pos=0.3125] {$z^3\big(\frac{4z+5}{5z+4}\big)^2$} +(0,3.2);
  \draw[<-,red,thick] (0.3,0) -- node[left,pos=0.3125] {$z$} +(0,3.2);
  \draw[<-,red,thick] (1.0,0) -- +(0,2);
  \node[red] at (1.6,1) {$\frac{(z+1)^2}{4z}$};
  \draw[<-,red,thick] (2.2,0) -- +(0,4.4);   
\end{tikzpicture}
\end{center}

\subsection{Shishikura and Tan Lei's example}\label{ex:tls}
Tan Lei and Shishikura consider in~\cite{shishikura-t:matings} a
mating of two polynomials of degree $3$, and show that it is
obstructed, but does not admit a Levy cycle (it is known that rational
maps of degree $2$ are obstructed if and only if they admit a Levy
cycle). This example was further studied
in~\cite{cheritat:shishikura}.

Their example may be described as follows. Consider the polynomial
\[f(z)=z^3+c,\text{ with $(c^3+c)^3+c=0$ and }c\approx-0.264425+1.26049i.
\]
Its post-critical orbit is $w\Rrightarrow u\to v\to w$, and the angles
landing at the critical point $u$ are $\{11,24,37\}/39$. Consider then the polynomial
\[g(z)=(a-1)(3z^2-2z^3)+1,\text{ with $g(a)=0$ and }a\approx
-0.42654.\]
Its post-critical orbit is $z\Rightarrow x\Rightarrow y\to z$, the
angles landing at $z$ are $\{21,47\}/78$ and those landing at $x$ are
$\{11,63\}/78$. Consider finally the mating $h$ of $f$ and $g$.

We choose as usual a basepoint $*$ close to the equator of $h$ and write
\[G=\pi_1(S^2\setminus P_h,*)=\langle u,v,w,x,y,z\mid u w v x y z\rangle\]
for standard generators of $G$ consisting of lollipops along external
rays of $f$ and $g$ respectively. For convenience, we write
$t=u w v=(x y z)^{-1}$ for the loop along the equator. In the basis
$\{\ell_1,\ell_2,\ell_3\}$ of $B(h)$ consisting of positively
oriented paths along the equator, we obtain
$t=u w v=(x y z)^{-1}=\pair{1,1,t}(1,2,3)$ and
\begin{xalignat*}{2}
  u&=\pair{v^{-1},u^{-1},t}(1,2,3), & x&=\pair{1,y z,y^{-1}}(2,3),\\
  v&=\pair{1,1,u}, & y&=\pair{t^{-1},1,t x}(1,3),\\
  w&=\pair{1,v,1}, & z&=\pair{1,y,1}.
\end{xalignat*}

Consider now the multicurve $\{r^G,s^G\}$ with $r=v y$ and
$s=u^w x^{v^{-1}}$. We have
\[r=\pair{t^{-1},1,t s^v}(1,3),\qquad s=\pair{s^{-v}t^{-1},r^{-v},t}(1,3),\]
so
$\{r^G,s^G\}$ is an annular obstruction with transition matrix
$(\begin{smallmatrix}0&1\\\frac12&\frac12\end{smallmatrix})$. The
corresponding decomposition of $G$ has three vertices, and is
\[G=G_1*_{\langle s\rangle}G_2*_{\langle r\rangle} G_3\text{ with }\begin{cases}
  G_1 &= \langle u^w,x^{v^{-1}},s^{-1}\mid u^w x^{v^{-1}}s^{-1}\rangle,\\
  G_2 &= \langle w,z,r,s\mid z w s r\rangle,\\
  G_3 &= \langle v,y,r^{-1}\mid v y r^{-1}\rangle.
\end{cases}\]

In basis $Q\coloneqq\{v\ell_1,v\ell_2,(u w)^{-1}\ell_3\}$ the
presentation of $B(h)$ becomes
\begin{xalignat*}{2}
  u^w&=\pair{1,w,1}(1,2,3), & x^{v^{-1}}&=\pair{1,s^{-1},w^{-1}r^{-1}}(2,3),\\
  r&=\pair{1,1,s}(1,3), & s&=\pair{s^{-1},r^{-1},1}(1,3),\\
  w&=\pair{1,v,1}, & z&=\pair{1,y^{v^{-1}},1},\\
  v&=\pair{1,1,u^w}, & y&=\pair{1,1,x^{v^{-1}}}(1,3).
\end{xalignat*}
From this presentation, just by looking at which of $G_1,G_2,G_3$ the
entries belong to, we get the sphere tree of bisets decomposition
\begin{equation}\label{eq:tls:gog}
  \begin{tikzpicture}[baseline=0.9cm]
    \node[inner sep=0pt] (B1) at (0,2) [label={left:$B_1$}] {$\bullet$};
    \node[inner sep=0pt] (B2) at (3,2.5) [label={above:$B_2$}] {$\bullet$};
    \node[inner sep=0pt] (B3) at (6,1.5) [label={right:$B_3$}] {$\bullet$};
    \node[inner sep=0pt] (B4) at (3,1.5) [label={[label distance=-3pt]below right:$B_4$}] {$\bullet$};
    \node[inner sep=0pt] (B5) at (6,2.5) [label={right:$B_5$}] {$\bullet$};
    \draw[thick] (B1.center) -- node[fill=white,inner sep=0] {$\circ$}
    (B2.center) -- node[fill=white,inner sep=0] {$\circ$} (B5.center)
    (B1.center) -- node[fill=white,inner sep=0] {$\circ$} (B4.center)
    -- node[fill=white,inner sep=0] {$\circ$} (B3.center);

    \node[inner sep=0pt] (G1) at (0,-0.0) [label={left:$G_1$}] {$\bullet$};
    \node[inner sep=0pt] (G2) at (3,-0.0) [label={below:$G_2$}] {$\bullet$};
    \node[inner sep=0pt] (G3) at (6,-0.0) [label={right:$G_3$.}] {$\bullet$};
    \draw[thick] (G1.center) -- node[fill=white,inner sep=0] (Gr) {$\circ$} node[below] {$\langle s\rangle$} (G2.center) -- node[fill=white,inner sep=0] (Gs) {$\circ$} node[below] {$\langle r\rangle$} (G3.center);

    \draw (B1) edge[rho] node[left] {$\rho$} (G1);
    \draw (B2) edge[rho] node[right,pos=0.83] {$\rho$} (G2);
    \draw (B5) edge[rho] node[right,pos=0.75] {$\rho$} (G3);
    \draw (B1) edge[lambda] node[below,pos=0.33] {$\lambda$} (G2);
    \draw (B2) edge[lambda,bend right=5] node[below,pos=0.66] {$\lambda$} (G3);
    \draw (B3) edge[lambda] node[below,pos=0.36] {$\lambda$} (G1);
    \draw (B4) edge[lambda,bend right=10] node[left=3pt,pos=0.25] {$\lambda$} (Gr);
  \end{tikzpicture}
\end{equation}
The vertex bisets are as usual obtained by restricting $B(h)$ while using
subsets of the basis $Q$, and are given as follows:
\begin{itemize}
\item the $G_2$-$G_1$-biset $B_1$ has in the basis $Q$
  the wreath recursion
  \[u^w=\pair{1,w,1}(1,2,3), \qquad x^{v^{-1}}=\pair{1,s^{-1},w^{-1}r^{-1}}(2,3),\qquad
  s^{-1}=\pair{1,r,s}(1,3);\]
  in appropriate co\"ordinates, it is the biset of the map
  $3z^2-2z^3$;
\item the $G_3$-$G_2$-biset $B_2$ has in the subbasis
  $\{\ell_2\}$ of $Q$ the wreath recursion
  \[w=\pair{v},\qquad z=\pair{y^{v^{-1}}},\qquad r=\pair{1},\qquad s=\pair{r^{-1}};\]
\item the $G_1$-$G_3$-biset $B_3$ has in the subbasis
  $\{\ell_1,\ell_3\}$ of $Q$ the wreath recursion
  \[v=\pair{1,u^w},\qquad y=\pair{1,x^{v^{-1}}}(1,2),\qquad r^{-1}=\pair{s^{-1},1}(1,2);\]
\item the $\langle s\rangle$-$G_2$-biset $B_4$ has in the subbasis
  $\{\ell_1,\ell_3\}$ of $Q$ the wreath recursion
  \[w=\pair{1,1},\qquad z=\pair{1,1},\qquad r=\pair{1,s}(1,2),\qquad s=\pair{s^{-1},1}(1,2);\]
\item the $1$-$G_3$ biset $B_5$ is trivial on the subbasis
  $\{\ell_2\}$, and corresponds to a sphere that gets blown down to a
  point.
\end{itemize}

The only small cycle in the sphere tree of bisets is the
$G_1$-$G_1$-biset $C\coloneqq B_3\otimes_{G_3}B_2\otimes_{G_2}B_1$ and
its two cyclic permutations. A presentation for $C$, in the basis
$\{\ell_1\ell_2\ell_1,\ell_3\ell_2\ell_1,\ell_1\ell_2\ell_2,\ell_3\ell_2\ell_2,\ell_1\ell_2\ell_3,\ell_3\ell_2\ell_3\}$,
is
\begin{align*}
  u^w&=\pair{1,1,1,u^w,1,1}(1,3,5)(2,4,6),\\
  x^{v^{-1}}&=\pair{1,1,1,u^w x^{v^{-1}},1,u^{-w}}(3,6,4,5),\\
  s^{-1}&=\pair{1,1,1,1,s^{-1},1}(1,5,2,6).
\end{align*}
A direct calculation shows that it is isomorphic to
$B(z^4(-2z^2+6z-3)/(2z-1)^3)$; note that this map factors under
composition as $(3z^2-2z^3)\circ(z^2/(2z-1))$. This will also follow
from the algebraic realization of $h$, which is as follows:
\begin{center}
\begin{tikzpicture}
  \coordinate (p0) at (0,4.9);
  \coordinate (p1) at (0,2.7);
  \coordinate (q0) at (3,4.9);
  \coordinate (q1) at (3,2.7);
  
  \draw (p0) node[below left=-5pt and 7pt] {\scriptsize $-1/2$}
  .. controls +(170:1) and +(-30:1) .. +(-2.5,1.0) node[left] {$w$} node[below right=5pt and 4pt] {\scriptsize $\infty$}
  .. controls +(-50:1) and +(20:0.2) .. +(-2.7,-0.75) node [right] {\scriptsize $3/2$} .. controls +(-20:0.2) and +(50:1) .. +(-2.5,-2.5) node[left] {$z$} node[above right=3pt and 4pt] {\scriptsize $0$}
  .. controls +(30:1) and +(-170:1) .. (p1) node[above left=-4pt and 7pt] {\scriptsize $1$}
  .. controls +(170:1) and +(-170:1) .. cycle;
  \draw[shrunksphere] (p1) node[right=4pt]{\scriptsize $\infty$}
  .. controls +(10:0.5) and +(-100:0.5) .. +(0.8,0.8) node[below right=3pt and -4pt] {\scriptsize $1$}
  .. controls +(-80:0.5) and +(-100:0.5) .. +(2.2,0.8) node[below left=4pt and -6pt] {\scriptsize $-1$}
  .. controls +(-80:0.5) and +(170:0.5) .. (q1) node[left=6pt] {\scriptsize $0$}
  .. controls +(-170:0.5) and +(80:0.5) .. +(-0.8,-0.8) node[above left=6pt and -7pt] {\tiny $\sqrt{-2}$}
  .. controls +(100:0.5) and +(80:0.5) .. +(-2.2,-0.8) node[above right=6pt and -5pt] {\tiny $-\sqrt{-2}$}
  .. controls +(100:0.5) and +(-10:0.5) .. cycle;
 \draw (p0) node[right=8pt]{\scriptsize $1$}
  .. controls +(10:0.5) and +(-100:0.5) .. +(1.5,0.8) node[below=4pt] {\scriptsize $\infty$} node[right] {$v$}
  .. controls +(-80:0.5) and +(170:0.5) .. (q0) node[left=13pt] {\scriptsize $-1/2$}
  .. controls +(-170:0.5) and +(80:0.5) .. +(-1.5,-0.8) node[above=3pt] {\scriptsize $0$} node[right] {$y$}
  .. controls +(100:0.5) and +(-10:0.5) .. cycle;
 \draw[shrunksphere] (q0) node[right=15pt] {\scriptsize $1$}
  .. controls +(10:1) and +(-150:1) .. +(2.5,0.6) node[below left=4pt and 3pt] {\scriptsize $\infty$}
  .. controls +(-130:0.5) and +(130:0.5) .. +(2.5,-0.6) node[above left=3pt and 3pt] {\scriptsize $0$}
  .. controls +(150:1) and +(-10:1) .. cycle;
 \draw (q1) node[above right=-5pt and 9pt] {\scriptsize $1$}
  .. controls +(10:0.7) and +(-110:0.5) .. +(1.5,1.0) node[below=5pt] {\scriptsize $1/2$} .. controls +(-70:0.5) and +(-150:0.7) .. +(2.5,0.8) node[right] {$u^w$} node[below left=4pt and 3pt] {\scriptsize $\infty$}
  .. controls +(-130:0.5) and +(130:0.5) .. +(2.5,-0.8) node[right] {$x^{v^{-1}}$} node[above left=3pt and 3pt] {\scriptsize $0$}
  .. controls +(150:1) and +(-10:1) .. cycle;
 
  \coordinate (r) at (0,0);
  \coordinate (s) at (3,0);
  
  \draw (r) node[left=4mm] {\scriptsize $1$}
  .. controls +(170:1) and +(-30:1) .. +(-2.5,1.0) node[left] {$u^w$} node[below right=3pt and 1pt] {\scriptsize $\infty$}
  .. controls +(-50:0.5) and +(50:0.5) .. +(-2.5,-1.0) node[left] {$x^{v^{-1}}$} node[above right=2pt and 1pt] {\scriptsize $0$}
  .. controls +(30:1) and +(-170:1) .. cycle;
  \draw (r) node[right=6pt]{\scriptsize $1$}
  .. controls +(10:0.5) and +(-100:0.5) .. +(1.5,1) node[below=4pt] {\scriptsize $\infty$} node[right] {$w$}
  .. controls +(-80:0.5) and +(170:0.5) .. (s) node[left=7pt] {\scriptsize $-1/2$}
  .. controls +(-170:0.5) and +(80:0.5) .. +(-1.5,-1) node[above=3pt] {\scriptsize $0$} node[right] {$z$}
  .. controls +(100:0.5) and +(-10:0.5) .. cycle;
  \draw (s) node[right=3mm] {\scriptsize $1$}
  .. controls +(10:1) and +(-150:1) .. +(2.5,1.0) node[right] {$v$} node[below left=4pt and 2pt] {\scriptsize $\infty$}
  .. controls +(-130:0.5) and +(130:0.5) .. +(2.5,-1.0) node[right] {$y$} node[above left=3pt and 2pt] {\scriptsize $0$}
  .. controls +(150:1) and +(-10:1) .. cycle;
  
  \node at (-3.5,3.5) {$\mathbb Y$};
  \node at (-3.5,0) {$\mathbb X$};
  \draw[<-,red,thick] (-1.5,0) -- node[left=-2pt,pos=0.4] {$3z^2-2z^3$} +(0,3.5);
  \draw[<-,red,thick] (1.2,0) -- node[left=-2pt,pos=0.4] {$\frac{z^2-1}{z^2+2}$} +(0,2.8);
  \draw[<-,red,thick] (1.92,0) -- node[right=-2pt,pos=0.286] {$z$} +(0,4.8);
  \draw[<-,red,thick] (3.8,0) -- node[left=-2pt,pos=0.292] {$z$} +(0,4.9);
  \draw[<-,red,thick] (4.9,0) -- node[left=-2pt,pos=0.519] {$\frac{z^2}{2z-1}$} +(0,2.7);
\end{tikzpicture}
\end{center}

We remark that, since $h$ does not admit any Levy obstruction, it is
isotopic to an expanding map for a path metric on $(S^2,P_h)$. It
admits, therefore, a \emph{Julia set}, defined for example as the
accumulation set of iterated preimages of a generic point. On the
other hand, $(3z^2-2z^3)\circ z\circ\frac{z^2}{2z-1}$ is
a rational map, so also admits a Julia set. Ch\'eritat investigates
this example in~\cite{cheritat:shishikura} by comparing these Julia
sets.

On the above noded sphere model, the Julia set of $h$ is, on the first
sphere, the Julia set of $(3z^2-2z^3)\circ z\circ\frac{z^2}{2z-1}$; on
the other two spheres, it is the Julia set of its cyclic permutations.

Ch\'eritat also asks, in~\cite{cheritat:shishikura}*{Question~2}, to
prove that the map of noded spheres obtained from pinching the
canonical obstruction in $h$ is indeed (in a different normalization)
the map above. This follows from~\cite{selinger:augts}, and also
follows immediately from computing the bisets in the decomposition, as
we have done.

\subsection{A Thurston map with infinitely generated centralizer}\label{ex:infinitely generated}
We conclude this survey with an example that shows that centralizers
of Thurston maps can be sometimes quite complicated, and in particular
not finitely generated (whence our notion of ``sub-computable''). The
example has degree $6$ and $7$ marked points. Many generalizations are
possible, but we content ourselves with its direct description, see
Figure~\ref{fig:infinitely generated}.

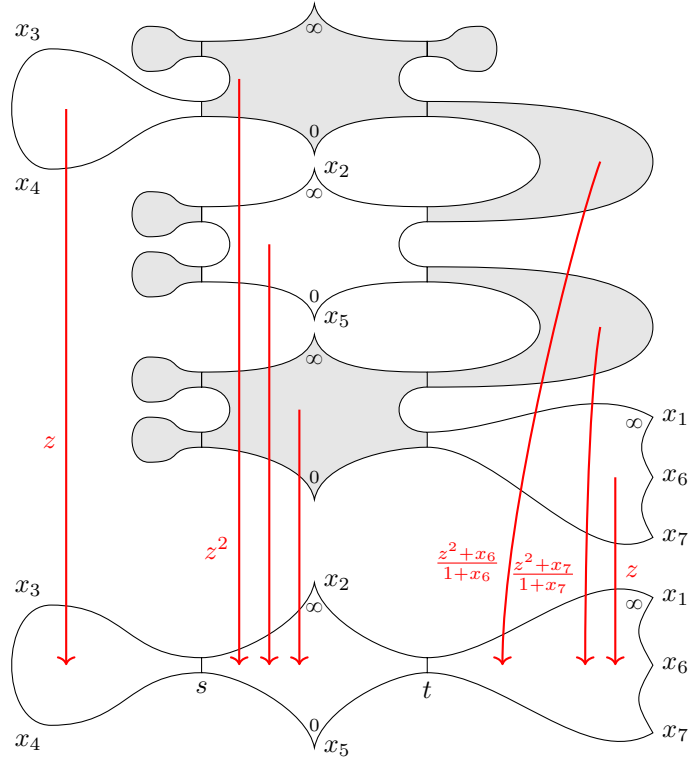
\begin{figure}
  \begin{center}
    \begin{tikzpicture}
      \def\smallsphere{.. controls +(180:0.5) and +(0:0.5) .. +(-0.7,0.2)
        .. controls +(180:0.3) and +(180:0.3) .. +(-0.7,-0.4)
        .. controls +(0:0.5) and +(180:0.5) .. +(0,-0.2) -- +(0,0.0) -- +(0,-0.2)}
      \def\smallrightsphere{-- +(0,0.2) .. controls +(0:0.5) and +(180:0.5) .. +(0.7,0.4)
        .. controls +(0:0.3) and +(0:0.3) .. +(0.7,-0.2)
        .. controls +(180:0.5) and +(0:0.5) .. +(0,0)}
      \draw (-3,7.5) .. controls +(180:1) and +(0:1) .. +(-2,0.7) node[above left] {$x_3$}
      .. controls +(180:0.7) and +(180:0.7) .. +(-2,-0.9) node[below left] {$x_4$}
      .. controls +(0:1) and +(180:1) .. ++(0,-0.2);
      \draw[fill=gray!20] (0,7.5) .. controls +(180:0.5) and +(180:0.5) .. ++(0,0.6) \smallrightsphere --
      ++(0,0.2) .. controls +(180:0.5) and +(-80:0.5) .. +(-1.5,0.5) node[below=4pt] {\scriptsize $\infty$}
      .. controls +(-100:0.5) and +(0:0.5) .. ++(-3,0) \smallsphere -- ++(0,-0.2)
      .. controls +(0:0.5) and +(0:0.5) .. ++(0,-0.6) -- ++(0,-0.2)
      .. controls +(0:0.5) and +(100:0.5) .. +(1.5,-0.5) node[above=3pt] {\scriptsize $0$}
      .. controls +(80:0.5) and +(180:0.5) .. ++(3,0) -- +(0,0.2);
      \draw[fill=gray!20] (0,7.5) .. controls +(0:1) and +(90:0.8) .. +(3,-0.8)
      .. controls +(-90:0.8) and +(0:1) .. ++(0,-1.6) +(0,0.2) .. controls +(0:1) and +(-90:0.3) .. +(1.5,0.8)
      .. controls +(90:0.3) and +(0:1) .. ++(0,1.4) -- ++(0,0.2);	   
      \draw (0,5.3) .. controls +(180:0.5) and +(180:0.5) .. ++(0,0.6) --
      ++(0,0.2) .. controls +(180:0.5) and +(-80:0.5) .. +(-1.5,0.5) node[below=4pt] {\scriptsize $\infty$} node[right] {$x_2$}
      .. controls +(-100:0.5) and +(0:0.5) .. ++(-3,0) -- ++(0,-0.2)
      .. controls +(0:0.5) and +(0:0.5) .. ++(0,-0.6) -- ++(0,-0.2)
      .. controls +(0:0.5) and +(100:0.5) .. +(1.5,-0.5) node[above=3pt] {\scriptsize $0$} node[right] {$x_5$}
      .. controls +(80:0.5) and +(180:0.5) .. ++(3,0) -- +(0,0.2);
      \draw[fill=gray!20] (-3,5.3) \smallsphere ++(0,0.8) \smallsphere
      (0,5.3) .. controls +(0:1) and +(90:0.8) .. +(3,-0.8)
      .. controls +(-90:0.8) and +(0:1) .. ++(0,-1.6) +(0,0.2) .. controls +(0:1) and +(-90:0.3) .. +(1.5,0.8)
      .. controls +(90:0.3) and +(0:1) .. ++(0,1.4) -- ++(0,0.2);
      \draw[fill=gray!20] (0,3.1) .. controls +(180:0.5) and +(180:0.5) .. ++(0,0.6) --
      ++(0,0.2) .. controls +(180:0.5) and +(-80:0.5) .. +(-1.5,0.5) node[below=4pt] {\scriptsize $\infty$}
      .. controls +(-100:0.5) and +(0:0.5) .. ++(-3,0) \smallsphere -- ++(0,-0.2)
      .. controls +(0:0.5) and +(0:0.5) .. ++(0,-0.6) \smallsphere -- ++(0,-0.2)
      .. controls +(0:0.5) and +(100:0.5) .. +(1.5,-0.7) node[above=3pt] {\scriptsize $0$}
      .. controls +(80:0.5) and +(180:0.5) .. ++(3,0) -- ++(0,0.2);
      \draw (0,2.9)
      .. controls +(0:1) and +(-150:1) .. +(3,-1.2) node[right] {$x_7$}
      .. controls +(120:0.4) and +(-120:0.4) .. +(3,-0.4) node[right] {$x_6$}
      .. controls +(120:0.4) and +(-120:0.4) .. +(3,0.4) node[right] {$x_1$} node[below left=-1mm and 0mm] {\scriptsize $\infty$}
      .. controls +(150:1) and +(0:1) .. ++(0,0.2);
      \draw (0,0.1) .. controls +(180:0.5) and +(-80:0.5) .. +(-1.5,1) node[below=4pt] {\scriptsize $\infty$} node[right] {$x_2$}
      .. controls +(-100:0.5) and +(0:0.5) .. ++(-3,0)
      .. controls +(180:1) and +(0:1) .. +(-2,0.7) node[above left] {$x_3$}
      .. controls +(180:0.7) and +(180:0.7) .. +(-2,-0.9) node[below left] {$x_4$}
      .. controls +(0:1) and +(180:1) .. ++(0,-0.2) node[below] {$s$} -- +(0,0.2) ++(0,0)
      .. controls +(0:0.5) and +(100:0.5) .. +(1.5,-1) node[above=3pt] {\scriptsize $0$} node[right] {$x_5$}
      .. controls +(80:0.5) and +(180:0.5) .. ++(3,0) node[below] {$t$} -- +(0,0.2) ++(0,0)
      .. controls +(0:1) and +(-150:1) .. +(3,-0.8) node[right] {$x_7$}
      .. controls +(120:0.5) and +(-120:0.5) .. +(3,0.1) node[right] {$x_6$}
      .. controls +(120:0.5) and +(-120:0.5) .. +(3,1.0) node[right] {$x_1$} node[below left=-1mm and 0mm] {\scriptsize $\infty$}
      .. controls +(150:1) and +(0:1) .. ++(0,0.2);
      \draw[<-,red,thick] (-4.8,0) -- node[left,pos=0.4] {$z$} +(0,7.4);
      \draw[<-,red,thick] (-2.5,0) -- node[left,pos=0.2] {$z^2$} +(0,7.8);
      \draw[<-,red,thick] (-2.1,0) -- +(0,5.6);
      \draw[<-,red,thick] (-1.7,0) -- +(0,3.4);
      \draw[<-,red,thick] (1.0,0) .. controls +(90:2) and +(-110:1) .. node[left,pos=0.2] {$\frac{z^2+x_6}{1+x_6}$} +(1.3,6.7);
      \draw[<-,red,thick] (2.1,0) .. controls +(90:1) and +(-100:1) .. node[left,pos=0.3] {$\frac{z^2+x_7}{1+x_7}$} +(0.2,4.5);
      \draw[<-,red,thick] (2.5,0) -- node[right,pos=0.5] {$z$} +(0,2.5);
    \end{tikzpicture}
  \end{center}
  \caption{A Thurston map with infinitely generated centralizer}\label{fig:infinitely generated}
\end{figure}

The seven marked points $A=\{x_1,\dots,x_7\}$ of the Thurston map $f$
are separated by two curves $s,t$ as follows: $x_3,x_4$ lie on an
$f$-fixed sphere $S_0$ separated by $s$ from a sphere $S_1$ containing
$x_2,x_5$ and on which $f$ acts as $z^2$; and $S_1$ is separated by
$t$ from an $f$-fixed sphere $S_2$ containing $x_1,x_6,x_7$. The
canonical obstruction of $f$ is $\{s,t\}$. There are two other
preimages of $S_1$, mapping by $z^2$ to $S_1$ and embedded in annuli
about $s$ and $t$ respectively, and two other preimages of $S_2$
mapping by degree $2$ to $S_2$, respectively with critical values
$\{x_1,x_6\}$ and embedded in an annulus about $s$, and with critical
values $\{x_1,x_7\}$ and embedded in an annulus about $t$.

The Thurston matrix of the multicurve $\CC\coloneqq\{s,t\}$ is
$(\begin{smallmatrix}1 & 2\\0 & 3\end{smallmatrix})$. Note that every
curve in $S_2$ is a Levy cycle, however the return map on $S_2$ is the
identity and the canonical obstruction does not contain the Levy
cycles in degree-$1$ pieces.

The centralizer of $f$ preserves the canonical obstruction, so is a
subgroup of
$\Mod(S^2,A,\CC)\cong \Mod(S_0)\times \Mod(S_1)\times \Mod(S_2)\times
\Z^{\CC}$.  Furthermore, $\Mod(S_0)=1$ because it contains only three
marked points, and the projection of $Z(f)$ into $\Mod(S_1)$ is
trivial because the restriction of $f$ to $S_1$ is a rational
self-map. Therefore, $Z(f)$ is a subgroup of $\Z^\CC\times\Mod(S_2)$.

To compute it, we write down a presentation of $B(f)$, and compute
some relations in its mapping class biset. We set
\[G=\langle x_1,x_2,x_3,x_4,x_5,x_6,x_7\mid x_1x_2x_3x_4x_5x_6x_7\rangle,\]
write $s=x_3x_4$ and $t=x_2x_3x_4x_5$, and in a basis
$\{\ell_1,\dots,\ell_7\}$ we compute the presentation
\begin{equation}\label{eq:infinite centralizer}
  \begin{aligned}
  x_1 &=\pair{1,s,s^{-1},t,t^{-1},x_1}(2,3)(4,5),\\
  x_2 &=\pair{1,1,s^{-1},x_2s,t^{-1},t}(1,2)(3,4)(5,6),\\
  x_3 &=\pair{x_3,1,1,1,1,1},\\
  x_4 &=\pair{x_4,1,1,1,1,1},\\
  x_5 &=\pair{1,1,x_5,1,1,1}(1,2)(3,4)(5,6),\\
  x_6 &=\pair{1,1,1,1,1,x_6}(2,3),\\
  x_7 &=\pair{1,1,1,1,1,x_7}(4,5),
\end{aligned}
\end{equation}
giving $s=\pair{s,1,1,1,1,1}$ and $t=\pair{1,s,s^{-1},t,t^{-1},t}$.
We write $\sigma,\tau,\alpha,\beta$ for Dehn twists about $s,t,x_1x_6$
and $x_6x_7$ respectively; their actions on $G$ are given respectively
by
\begin{align*}
  \sigma:\qquad & x_3\mapsto x_3^s,\;x_4\mapsto x_4^s,\\
  \tau:\qquad & x_2\mapsto x_2^t,\;x_3\mapsto x_3^t,\;x_4\mapsto x_4^t,\;x_5\mapsto x_5^t,\\
  \alpha:\qquad & x_1\mapsto x_1^{t x_6t^{-1}},\;x_6\mapsto x_6^{t^{-1}x_1t x_6},\\
  \beta:\qquad & x_6\mapsto x_6^{x_6x_7},\;x_7\mapsto x_7^{x_6x_7},
\end{align*}
all other generators being fixed. Naturally
$[\sigma,\alpha]=[\tau,\alpha]=[\sigma,\beta]=[\tau,\beta]=1$ while
$\langle\alpha,\beta\rangle$ is a free group of rank $2$. We then compute
\begin{xalignat*}{2}
  B(f)\cdot\sigma&\cong\sigma\cdot B(f),& B(f)\cdot\tau&\cong\sigma^2\tau^3\cdot B(f),\\
  B(f)\cdot\alpha&\cong\alpha\sigma^2\cdot B(f),& B(f)\cdot\beta&\cong\beta\cdot B(f).
\end{xalignat*}
For the second equality, the recursion of
$\sigma^{-2}\tau^{-3}\cdot B(f)\cdot\tau$ in basis
$\{s^2t^3\ell_1,st^3\ell_2,st^3\ell_3,t^2\ell_4,t^2\ell_5,\ell_6\}$
coincides with~\eqref{eq:infinite centralizer}, while for the third
equality, the recursion of
$\sigma^{-2}\alpha^{-1}\cdot B(f)\cdot\alpha$ in basis
$\{s^2\ell_1,s^2\ell_2,\ell_3,\dots,\ell_6\}$ coincides
with~\eqref{eq:infinite centralizer}.

Consider the homomorphism $\phi\colon\langle\alpha,\beta\rangle\to\Z$
which counts the total exponent in $\alpha$ of a word; it is the
quotient by the normal closure of $\beta$. Then, for
$w\in\langle\alpha,\beta\rangle$, the element $w\sigma^m\tau^n$
belongs to the centralizer of $f$ if and only if
$(m,n)=(m+2n+\phi(w),3n)$, if and only if $n=0$ and
$w\in\ker(\phi)$. Therefore,
\[Z(f)=\langle\sigma\rangle\times\ker(\phi)=\langle\sigma\rangle\times\langle\beta,\beta^\alpha,\beta^{\alpha\beta},\dots\rangle\cong\Z\times F_\infty.\]

\end{document}